\documentclass[acmsmall,screen,nonacm]{acmart}

\usepackage{xspace}
\usepackage{mathpartir}
\usepackage{mathtools}
\usepackage{amsmath,amsfonts}
\usepackage{listings}
\usepackage{subcaption}
\usepackage{caption}
\usepackage{wrapfig}
\usepackage{xcolor}
\usepackage{bm}
\usepackage{stmaryrd}
\usepackage{enumitem}
\usepackage{hyperref}
\usepackage{multirow}
\usepackage[capitalise,noabbrev,nameinlink]{cleveref}
\usepackage{thmtools}

\hypersetup{
    colorlinks,
    citecolor=black,
    filecolor=black,
    linkcolor=blue,
    urlcolor=black
}

\newtheorem{theorem}{Theorem}[section]   % numbering within sections
\newtheorem{lemma}[theorem]{Lemma}
\newtheorem{corollary}[theorem]{Corollary}

\newtheorem{definition}[theorem]{Definition}

\crefname{theorem}{theorem}{theorems}
\Crefname{theorem}{Theorem}{Theorems}
\crefname{lemma}{lemma}{lemmas}
\Crefname{lemma}{Lemma}{Lemmas}
\crefname{corollary}{corollary}{corollaries}
\Crefname{corollary}{Corollary}{Corollaries}
\crefname{proposition}{proposition}{propositions}
\Crefname{proposition}{Proposition}{Propositions}
\crefname{definition}{definition}{definitions}
\Crefname{definition}{Definition}{Definitions}
\crefname{example}{example}{examples}
\Crefname{example}{Example}{Examples}
\crefname{remark}{remark}{remarks}
\Crefname{remark}{Remark}{Remarks}
\Crefname{hypothesis}{Hypothesis}{Hypotheses}

\newcommand{\rulename}[1]{\textnormal{\textsc{#1}}}
\newcommand{\rref}[1]{\hyperref[#1]{\rulename{#1}}} % hyperlink a rule
\newcommand{\inferenceRule}[3]{\inferrule*[lab={\label{#1}\rulename{#1}}]{#2}{#3}}

\makeatletter
\newcommand*{\inlineequation}[2][]{%
  \begingroup
    \refstepcounter{equation}%
    \ifx\\#1\\%
    \else
      \label{#1}%
    \fi
    \relpenalty=10000 %
    \binoppenalty=10000 %
    \ensuremath{%
      #2%
    }%
    ~\@eqnnum
  \endgroup
}
\makeatother

\newcommand{\bnflabel}[1]{\mbox{\textit{#1}}}
\newcommand{\bnfalt}{\mathrel{\bf \mid }}

\newcommand{\bnfdef}{\mathrel{\bf::=}}
\newcommand{\deftable}[1]{\[\begin{array}{l l l l} #1 \end{array}\]}
\newcommand{\defline}[3]{\bnflabel{#1} & #2 & \bnfdef  & #3 \\}
\newcommand{\defnewline}[1]{ & & & #1 \\}

\newcommand{\many}[1]{{\color{cyan} \overline{\color{black}#1}}}

\newcommand{\kw}[1]{\bm{\mathtt{{#1}}}}
\newcommand{\meta}[1]{{\textsf{\textbf{#1}}}}

\newcommand{\var}[1]{\mathit{#1}}

\newcommand{\payable}{\kw{payable}}

\newcommand{\opi}{\circ_i}
\newcommand{\opb}{\circ_b}
\newcommand{\cmp}{\bowtie}

\newcommand{\ite}[3]{
\kw{if} ~#1 ~\kw{then} ~ #2 ~\kw{else}~ #3
}

\newcommand{\contract}[4]{
\kw{contract}~#1:  #2 ~ #3 ~#4
}

\newcommand{\constrinlinecase}[5]{
 \kw{constructor} ~(#1)~\optkey{\payable}~ 
 \kw{iff}~#2 ~
 \caseblock{#3_{i}}{
 \kw{creates}~#4
 }~
 \kw{ensures}~#5
}

\newcommand{\constrcase}[4]{
 \kw{constructor} ~(#1)~\optkey{\payable} ~
 \kw{iff}~#2 ~
  #3~
 \kw{ensures}~#4
}

\newcommand{\transinlinecase}[6]{
\kw{transition}~id(#1)~\optkey{\payable}~ : \alpha ~
\kw{iff}~#2 ~
\caseblock{#3_{i}} 
{\kw{updates}~#4_{i}~
\kw{returns}~#5_{i}}~
\kw{ensures}~#6
}

\newcommand{\transcase}[4]{
\kw{transition}~id(#1)~\optkey{\payable}~ : \alpha ~
\kw{iff}~#2 ~
#3~
\kw{ensures}~#4
}

\newcommand{\asgn}{\mathrel{:=}}

\newcommand{\uint}[1]{\kw{uint}#1}

\newcommand{\intt}[1]{\kw{int}#1}
\newcommand{\bool}{\kw{bool}}
\newcommand{\mapping}[2]{\kw{mapping}(#1 \Rightarrow #2)}

\newcommand{\inrange}[2]{\kw{inrange}(#1,#2)}

\newcommand{\pre}[1]{\kw{pre}(#1)}
\newcommand{\post}[1]{\kw{post}(#1)}

\newcommand{\true}{\kw{true}}
\newcommand{\false}{\kw{false}}

\newcommand{\caseblock}[2]{\range{\kw{case}~#1: #2}{i \in [1,n]}}

\newcommand{\defeq}{\mathrel{\overset{\MyDef}{\resizebox{\widthof{\kern1.25pt\MyEqdefU}}{\heightof{$=$}}{$=$}}}}

\newcommand{\bigstep}{\Downarrow}
\newcommand{\psem}[3]{#1;#2;#3\bigstep_{\ell}}
\newcommand{\psemnoell}[3]{#1;#2;#3\bigstep}
\newcommand{\psemdummy}[3]{#1;#2;#3\bigstep_{\dummyloc}}
\newcommand{\vsem}[1]{\llbracket #1 \rrbracket}

\newcommand{\possible}[1]{\meta{possible}(#1)}

\newcommand{\partfun}{\rightharpoonup}
\newcommand{\fun}{\rightarrow}

\newcommand{\dom}[1]{\meta{dom}(#1)}

\newcommand{\ins}[2]{#1 \downarrow^{\kw{ins}}_{\ell} #2}

\newcommand{\range}[2]{\{ #1\}_{#2}}

\newcommand{\addr}[1]{\kw{addr}(#1)}

\newcommand{\fresh}[1]{\meta{fresh}(#1)}

\newcommand{\annot}[1]{{\color{gray} #1}}

\newcommand{\maybe}[1]{{#1}^{\,?}}

\newcommand{\optkey}[1]{?(#1)}

\newcommand{\none}{\bot}

\newcommand{\map}[2]{[\; #1 \;]_{#2}}

\newcommand{\spre}{s_{\kw{pre}}}
\newcommand{\spost}{s_{\kw{post}}}

\newcommand{\case}[1]{\item \underline{#1} \\}

\newenvironment{proofcases}%
  {\begin{itemize}[leftmargin=*]%
    }%
  {\end{itemize}}

\newcommand{\addrfield}{\meta{addr}}
\newcommand{\dummyloc}{\cdot}

\newcommand{\sstorage}{{\meta{Storage}}}
\newcommand{\scnstrs}{\meta{Cnstr}}
\newcommand{\stranss}{{\meta{Trans}}}

\newcommand{\Sigmastore}{\Sigma_{\sstorage}}
\newcommand{\Sigmacode}{\Sigma_{\scnstrs}}
\newcommand{\Sigmatranss}{\Sigma_{\stranss}}

\newcommand{\Sigmastorep}{\Sigma'_{\sstorage}}
\newcommand{\Sigmacodep}{\Sigma'_{\scnstrs}}
\newcommand{\Sigmatranssp}{\Sigma'_{\stranss}}

\newcommand{\specific}{\prec_{\kw{specific}}}
\newcommand{\specificeq}{\preceq_{\kw{specific}}}

\newcommand{\act}{\textsc{act}\xspace}
\begin{document}

%% end of the preamble, start of the body of the document source.

\title{\act: Technical Report}

\author{Zoe Paraskevopoulou}
\affiliation{
  \institution{National Technical University of Athens}
  \city{Athens}
  \country{Greece}
}
\affiliation{
  \institution{Argot Collective}
  \city{Zug}
  \country{Switzerland}
}
\email{zoe.paraskevopoulou@gmail.com}
\orcid{0009-0001-2352-9818}

\author{Anja Petković Komel}
\affiliation{
  \institution{Argot Collective}
  \city{Zug}
  \country{Switzerland}
}
\email{anja@argot.org}
\orcid{0000-0001-7203-6641}

\author{Sophie Rain}
\affiliation{
  \institution{Argot Collective}
  \city{Zug}
  \country{Switzerland}
}
\email{sophie.rain@argot.org}
\orcid{0000-0002-8940-4989}

\author{Lefteris Lazaropoulos}
\affiliation{
  \institution{National Technical University of Athens}
  \city{Athens}
  \country{Greece}
}
\affiliation{
  \institution{Argot Collective}
  \city{Zug}
  \country{Switzerland}
}
\email{lefterislazar@protonmail.com}

\author{Alexis Terry}
\affiliation{
  \institution{Argot Collective}
  \city{Zug}
  \country{Switzerland}
}
\email{alexis@argot.org}

\maketitle

% Short scope paragraph describing the document for readers.
\noindent\textbf{Scope.} This technical report contains the formal definitions
and metatheory for the \act specification and verification language. It documents
the syntax, the operational pointer semantics, the type system 
and the main metatheoretic results (type-safety).
The mechanised proofs are available in the \act repository\footnote{\url{https://github.com/argotorg/act}}.

\tableofcontents

\section{Syntax}
\deftable{
% Top-level spec
\defline{(act specification)}{spec}{\many{\var{contract}}}
% Contract
\defline{(Contract)}{\var{contract}}{
\contract{\var{Id}}
{\var{constructor}}
{\many{\var{transition}}}
{\var{invariants}}
}
% Constructor
\defline{(Constructor)}{\var{constructor}}{
\kw{constructor}~(\many{\var{id} : \alpha})~\optkey{\payable}~
% \kw{pointers} ~\many{\var{ptr}}~
\kw{iff}~\many{e}}
\defnewline{
    \many{\kw{case}~e: \kw{creates}~\many{\var{create}}}~
    \kw{ensures} ~ \many{e}
}
% % pointers 
% \defline{(Pointer)}{\var{ptr}}
% {
%     \ptr{\var{ref}}{\var{Id}}
% }
% Transition
\defline{(Transition)}{\var{transition}}{
\kw{transition}~\var{id}(\many{\var{id} : \alpha}) ~\optkey{\payable} ~\optkey{: \alpha} ~
% \kw{pointers} ~\many{\var{ptr}}~
\kw{iff}~\many{e}}
\defnewline{
    \many{\kw{case}~e: \kw{updates}~\many{\var{update}}~\optkey{\kw{returns}~e}}~
    \kw{ensures} ~ \many{e}
}
% % Invariants
\defline{(Invariants)}{invariants}
{\kw{invariants}~\many{\var{e}}}
% % create
\defline{(Create)}{\var{create}}{ \sigma~id \asgn \var{se} }
% % update
\defline{(Update)}{\var{update}}{ \var{ref_{\annot{\sigma}}} \asgn \var{se} }
% % Assign expressions 
\defline{(Slot expressions)}{\var{se}}{{\var{m}} \bnfalt \kw{new}~Id~\optkey{\{\kw{value}:~\var{se}\}}(\many{\var{se}})
\bnfalt \var{ref_{\annot{Id}}}
\bnfalt \addr{\var{se}}
}
% % Mapping expressions 
\defline{(Mapping expressions)}{\var{m}}{{\var{e}} \bnfalt 
\map{\many{\var{e} \Rightarrow \var{m}}}{\annot{\mu}} \bnfalt \var{ref_{\annot{\mu}}} \map{\many{\var{e} \Rightarrow \var{m}}}{\annot{\mu}} } 
% % expressions
\defline{(Base Expressions)}{\var{e}}{i \bnfalt \true \bnfalt \false \bnfalt \var{ref_{\annot{\alpha}}} \bnfalt  \addr{\var{ref_{\annot{Id}}}} \bnfalt}
\defnewline{ \var{e}~\opi~\var{e} \bnfalt \var{e}~\opb~\var{e} \bnfalt \var{e}~\cmp~\var{e} \bnfalt \neg \var{e} \bnfalt }
\defnewline{ \inrange{\iota}{e} \bnfalt \ite{\var{e}}{\var{e}}{\var{e}} \bnfalt \var{e} = \var{e} }
% % binary integer operators
\defline{(Binary Integer Operators)}{\opi}{+ \bnfalt - \bnfalt * \bnfalt \kw{div} \bnfalt \kw{mod} \bnfalt \kw{exp}}
% % binary boolean operators
\defline{(Binary Boolean Operators)}{\opb}{\wedge \bnfalt \vee \bnfalt \implies}
% % Integer comparisons
\defline{(Integer Comparisons)}{\cmp}{< \bnfalt \leq \bnfalt \geq \bnfalt >}
% envirentment variables
\defline{(Calling Environment)}{\var{env}}
{\kw{caller} \bnfalt \kw{origin} \bnfalt \kw{callvalue} \bnfalt \kw{this}}
% variables
\defline{(Variable References)}{\var{ref}}
{x \bnfalt \pre{x} \bnfalt \post{x} \bnfalt \var{ref} ~\kw{as}~\var{Id} \bnfalt \var{ref}.y \bnfalt \var{ref}[e] \bnfalt \var{env}}
% % types 
\defline{(Integer Types)}{\iota}{\uint{M} \bnfalt \intt{M} \bnfalt \intt{} }
& & & M \in \{ i * 8 ~ |~ \forall i \in \mathbb{N}, ~ 1 \leq i \leq 32 \} \\
% % types 
\defline{(Base type)}{\beta}{\iota \bnfalt \bool \bnfalt \kw{address}}
\defline{(Mapping type)}{\mu}{\mapping{\beta}{\mu} \bnfalt \beta}
\defline{(ABI type)}{\alpha}{\beta  \bnfalt \kw{address}_\var{Id}}
\defline{(Slot type)}{\sigma}{\mu \bnfalt \alpha \bnfalt \var{Id}}
}

\bigskip
% Zoe: move to here from syntax.tex to directly import syntax to the paper
Notes:
\begin{itemize}
    \item The non-terminals $\var{id},x, y$ denote identifiers that start
     with lowercase letter, and non-terminal $\var{Id},A$ denote identifiers
     that start with capital letters, i.e. contracts. The terminal $i$ denotes a numeric literal.
    \item We denote repetition with $\many{\text{overbar}}$. For example,
     $\many{e}$ denotes an ordered collection of expressions. When we want to be
     able explicitly refer to an item in a collection we index over a countable
     set, for example $\range{e_i}{i \in [1,N]}$ is a collection of $N$
     expressions and $e_i$ is the $i$-th expression.
    \item References in expressions and assignments have a type annotation.
    This is the type of the reference and it is filled during type checking.
    It is used in the bounds elaboration phase.

\end{itemize}

\pagebreak

\section{Pointer Semantics}
\subsection*{Values}
%<*refsemvals>
\[
\begin{array}{l}    
\begin{array}{l l}
\meta{Addr} \equiv \mathbb{N}  & 
\meta{State} \equiv \meta{Addr} \partfun \meta{String} \times (\meta{String} \partfun \meta {Value}) \\
\meta{BaseValue} \equiv  \mathbb{Z} + \mathbb{B} + \meta{Addr} & 
\meta{Env} \equiv \meta{String} \partfun \meta{BaseValue} \\
\end{array}\\
\; \meta{Value} \equiv (\mathbb{Z} \fun \meta{Value}) + (\mathbb{B} \fun \meta{Value}) + (\meta{Addr} \fun \meta{Value}) + \meta{BaseValue} 
\end{array} \] 
%</refsemvals>

\begin{flushleft}

The state maps each location $\ell$ to a pair of a contract type and a mapping
from storage variables to values.
For readability, we write $s(\ell).{\kw{type}}$ to denote the contract
type stored at location $\ell$ in state $s$ (i.e., $s(\ell).1$). We write
$s(\ell)(x)$ to denote the value of storage variable $x$ at location $\ell$ in
state $s$ (i.e., $(s(\ell).2)(x)$).

\subsection*{Semantics of Environment References}

\begin{flushleft}
\fbox{$\rho ; \var{env} \bigstep_{\ell} \var{v}$}
\end{flushleft}

\begin{mathpar}
    %% caller
    \inferenceRule{E-Caller}
    { 
        \kw{caller} \in \dom{\rho}
    }
    { 
       \rho ; \kw{caller} \bigstep_{\ell} \rho(\kw{caller})
    }
    \and
    %% origin
    \inferenceRule{E-Origin}
    {
        \kw{origin} \in \dom{\rho}
     }
    { 
        \rho ; \kw{origin} \bigstep_{\ell} \rho(\kw{origin})
    }
    \and
    %% callvalue
    %<*refsemCallvalue>
    \inferenceRule{E-Callvalue}
    { 
        \kw{callvalue} \in \dom{\rho}
    }
    { 
        \rho ; \kw{callvalue} \bigstep_{\ell} \rho(\kw{callvalue})
    }
    %</refsemCallvalue>
    \and
    %% this
    %<*refsemThis>
    \inferenceRule{E-This}
    { }
    { 
        \rho ; \kw{this} \bigstep_{\ell} \ell
    }
    %</refsemThis>
\end{mathpar}

\subsection*{Semantics of Variable References}

\begin{flushleft}
\fbox{$\psem{s^t}{\rho}{\var{ref}} (\var{v},t_p)$ }
\end{flushleft}
The state in the pointer semantics of variable references is annotated with a tag $t$ that takes values from the
set $\{ \kw{U}, \kw{T} \}$, denoting whether the reference appears in an untimed or timed judgement.
An untimed state is a single state $s^{\kw{U}} = s$ and a timed state is a pair of pre- and post-states $s^{\kw{T}} = (s_{\kw{pre}}, s_{\kw{post}})$.
We will omit the tag $t$ from the rest of the judgements, assuming they just propagate the same tag on the states, as given.
The tag $t_p \in \{\kw{U}, \kw{pre}, \kw{post}\}$ denotes the timing of the evaluated reference. 

%% TODO addr

\begin{mathpar}
    \inferenceRule{E-Environment}
    {
        \rho ; \var{env} \bigstep_{\ell} \var{v}
    }
    {
        \psem{s^t}{\rho}{\var{env}} (\var{v},\kw{U})
    }
    \and
    %% storage vars untimed
    \inferenceRule{E-Storage}
    {
        x \in \dom{s(\ell)} \\
        x \not \in \dom{\rho}
    }
    { 
        \psem{s^{\kw{U}}}{\rho}{x} (s(\ell)(x), \kw{U})
    }
    \and
    %% storage vars pre
    %<*refsemStoragePre>
    \inferenceRule{E-StoragePre}
    {
        x \in \dom{s_{\kw{pre}}(\ell)} \\
        x \not \in \dom{\rho}
    }
    { 
        \psem{(s_{\kw{pre}}, s_{\kw{post}})^{\kw{T}}}{\rho}{\pre{x}} (s_{\kw{pre}}(\ell)(x), \kw{pre})
    }
    %</refsemStoragePr>
    \and
    %% storage vars post
    \inferenceRule{E-StoragePost}
    {
        x \in \dom{s_{\kw{post}}(\ell)} \\
        x \not \in \dom{\rho}
    }
    { 
        \psem{(s_{\kw{pre}}, s_{\kw{post}})^{\kw{T}}}{\rho}{\post{x}} (s_{\kw{post}}(\ell)(x), \kw{post}) 
    }
    \and
    %% interface ref untimed
    \inferenceRule{E-Calldata}
    {
        x \in \dom{\rho}
    }
    { 
        \psem{s^{\kw{U}}}{\rho}{x} 
        (\rho(x), \kw{U})
    }
    \and
    %% interface ref timed
    \inferenceRule{E-CalldataTimed}
    {
        x \in \dom{\rho}
    }
    { 
        \psem{(s_{\kw{pre}}, s_{\kw{post}})^{\kw{T}}}{\rho}{x} 
        (\rho(x), \kw{pre})
    }
    \and
    %% casting
    \inferenceRule{E-Coerce}
    {
        \psem{s^t}{\rho}{\var{ref}} (\var{v}, t_p)
    }
    { 
        \psem{s^t}{\rho}{\var{\var{ref}} ~\kw{as}~ A}
        (\var{v}, t_p)
    }
    \and
    %% fields untimed
    \inferenceRule{E-Field}
    {
        \psem{s^{\kw{U}}}{\rho}{\var{ref}} (\var{\ell'}, \kw{U}) \\
        \var{\ell'} \in \meta{Addr} \\
        x \in \dom{s(\ell')}
    }
    { 
        \psem{s^{\kw{U}}}{\rho}{\var{ref.x}}
        (s(\var{\ell'})(x), \kw{U})
    }
    \and
    %% fields timed pre
    \inferenceRule{E-FieldPre}
    {
        \psem{(s_{\kw{pre}}, s_{\kw{post}})^{\kw{T}}}{\rho}{\var{ref}} (\var{\ell'}, \kw{pre}) \\
        \var{\ell'} \in \meta{Addr} \\
        x \in \dom{s_{\kw{pre}}(\ell')}
    }
    { 
        \psem{(s_{\kw{pre}}, s_{\kw{post}})^{\kw{T}}}{\rho}{\var{ref.x}}
        (s_{\kw{pre}}(\var{\ell'})(x), \kw{pre})
    }
    \and
     %% fields timed post
    \inferenceRule{E-FieldPost}
    {
        \psem{(s_{\kw{pre}}, s_{\kw{post}})^{\kw{T}}}{\rho}{\var{ref}} (\var{\ell'}, \kw{post}) \\
        \var{\ell'} \in \meta{Addr} \\
        x \in \dom{s_{\kw{post}}(\ell')}
    }
    { 
        \psem{(s_{\kw{pre}}, s_{\kw{post}})^{\kw{T}}}{\rho}{\var{ref.x}}
        (s_{\kw{post}}(\var{\ell'})(x), \kw{post})
    }
    \and
    %% mapping untimed
    \inferenceRule{E-RefMapping}
    {
        \psem{s^t}{\rho}{e} \var{v}_e \\
        \var{v}_e \in \meta{X} \\
        \psem{s^t}{\rho}{\var{ref}} (\var{v}_r, t_p) \\
        \var{v}_r \in \meta{X} \fun \meta{Value} \\
        \meta{X} \in \{\mathbb{Z},\mathbb{B},\meta{Addr}\} 
    }
    { 
        \psem{s^t}{\rho}{\var{ref}[e]}
        (\var{v}_r(\var{v}_e), t_p)
    }
\end{mathpar}

\subsection*{Insert value}

\begin{flushleft}
\fbox{$\ins{s ; \rho ; \var{ref} ; \var{v}}{s'}$}
\end{flushleft}    

%% TODO addr ? 

\begin{mathpar}
    %% storage vars
    \inferenceRule{E-InsStorage}
    { 
        x \in \dom{s(\ell)}
    }
    { 
        \ins{s ; \rho ; x ; \var{v}}{s[\ell \mapsto s(\ell)[x \mapsto \var{v}]]}
    }
    \and
    %% fields
    \inferenceRule{E-InsField}
    {
        \psem{s^{\kw{U}}}{\rho}{\var{ref}} (\var{\ell'}, \kw{U}) \\
        \var{\ell'} \in \meta{Addr} \\
        x \in \dom{s(\ell')}
    }
    { 
        \ins{s ; \rho ; \var{ref.x} ; \var{v}}{s[\ell' \mapsto s(\ell')[x \mapsto \var{v}]]}
    }
\end{mathpar}

\subsection*{Semantics of Expressions}
% semantics of expressions

We omit the typing-tags on the references $\var{ref_{\annot{\alpha}}}$ and $\var{ref_{\annot{A}}}$ for clarity, as they do not play a role here. 
We use $\meta{min},\meta{max}$ as follows

\begin{minipage}[t]{0.45\textwidth}
    \[\begin{array}{r l}
    % min definition
    \meta{min}(\uint{M}) = & \hspace{0.1em} 0, \\ 
    \meta{min}(\intt{M}) = & \hspace{-0.5em} -2^{M-1},\\
    \end{array}\]    
    where $ M \in \{ i * 8 ~ |~ \forall \, i \in \mathbb{N}. ~\, 1 \leq i \leq 32 \}$.
\end{minipage}
\begin{minipage}[t]{0.45\textwidth}
    \[\begin{array}{r l}
    % max definition
    \meta{max}(\uint{M}) = & \hspace{-0.5em} 2^M-1, \\
    \meta{max}(\intt{M}) = & \hspace{-0.5em} 2^{M-1}-1, \\
    \end{array}\]    
\end{minipage}

\begin{flushleft}
\fbox{$\psem{s}{\rho}{\var{e}} \var{v}$ }

\end{flushleft} 

\begin{mathpar}
    %% Integer literals
    \inferenceRule{E-Int}
    { n \in \mathbb{Z}}
    { 
        \psem{s}{\rho}{n} n 
    }
    \and
    %% booleans
    \inferenceRule{E-Bool}
    { b \in \mathbb{B}}
    { 
        \psem{s}{\rho}{b} b
    }
    \and
    %% References
    \inferenceRule{E-Ref}
    {
        \psem{s}{\rho}{\var{ref}} (\var{v}, t_p)
    }
    { 
        \psem{s}{\rho}{\var{ref}}
        \var{v} 
    }
    \and
    \inferenceRule{E-Addr}
    {
        \psem{s}{\rho}{\var{ref}} (\var{v}, t_p)
    }
    { 
        \psem{s}{\rho}{\addr{\var{ref}}}
        \var{v} 
    }
    \and
    %% inRange case true
    \inferenceRule{E-RangeTrue}
    {   
        \psem{s}{\rho}{e} \var{v} \\
        \var{v} \in \mathbb{Z} \\
        (\meta{min}(\iota) \leq \var{v} \leq \meta{max}(\iota)) ~\vee~ \iota = \intt{}
    }
    { 
        \psem{s}{\rho}{\inrange{\iota}{e}} \meta{True}
    }
    \and
    %% inRange case false
    \inferenceRule{E-RangeFalse}
    {
        \psem{s}{\rho}{e} \var{v} \\
        \var{v} \in \mathbb{Z} \\
        \var{v} < \meta{min}(\iota) \vee \meta{max}(\iota) < \var{v}
    }
    { 
        \psem{s}{\rho}{\inrange{\iota}{e}} \meta{False}
    }
    \and
    %% division case ok
    \inferenceRule{E-Div}
    {
        \psem{s}{\rho}{e_1} \var{v}_1 \\
        \psem{s}{\rho}{e_2} \var{v}_2 \\
        \var{v}_1, \var{v}_2 \in \mathbb{Z}\\
        \var{v}_2 \not = 0
    }
    { 
        \psem{s}{\rho}{e_1~\kw{div}~e_2}  \var{v}_1 /\var{v}_2
    }
    \and
     %% division by zero
    \inferenceRule{E-DivZero}
    {
        \psem{s}{\rho}{e_1} \var{v}_1 \\
        \psem{s}{\rho}{e_2} \var{v}_2 \\
        \var{v}_1, \var{v}_2 \in \mathbb{Z}\\
        \var{v}_2 = 0
    }
    { 
        \psem{s}{\rho}{e_1~\kw{div}~e_2}  0
    }
    \and
    %% modulo case ok
    \inferenceRule{E-Mod}
    {
        \psem{s}{\rho}{e_1} \var{v}_1 \\
        \psem{s}{\rho}{e_2} \var{v}_2 \\
        \var{v}_1, \var{v}_2 \in \mathbb{Z}\\
        \var{v}_2 \not = 0
    }
    { 
        \psem{s}{\rho}{e_1~\kw{mod}~e_2}  \var{v}_1~\%~\var{v}_2
    }
    \and
    %% modulo zero
    \inferenceRule{E-ModZero}
    {
        \psem{s}{\rho}{e_1} \var{v}_1 \\
        \psem{s}{\rho}{e_2} \var{v}_2 \\
        \var{v}_1, \var{v}_2 \in \mathbb{Z}\\
        \var{v}_2 = 0
    }
    { 
        \psem{s}{\rho}{e_1~\kw{mod}~e_2}  0
    }
    \and
    %% other binary integer operations
    \inferenceRule{E-BopI}
    {
        \psem{s}{\rho}{e_1} \var{v}_1 \\
        \psem{s}{\rho}{e_2} \var{v}_2 \\
        \var{v}_1, \var{v}_2 \in \mathbb{Z} \\
        \opi \notin \{\kw{div},\kw{mod}\}
    }
    { 
        \psem{s}{\rho}{e_1 \opi e_2}
        \var{v}_1 \opi \var{v}_2 
    }
    \and
    %% binary boolean operations
    \inferenceRule{E-BopB}
    {
        \psem{s}{\rho}{e_1} \var{v}_1 \\
        \psem{s}{\rho}{e_2} \var{v}_2 \\
        \var{v}_1, \var{v}_2 \in \mathbb{B}
    }
    { 
        \psem{s}{\rho}{e_1 \opb e_2}
        \var{v}_1 \opb \var{v}_2 
    }
    \and
    %% negation
    \inferenceRule{E-Neg}
    {
        \psem{s}{\rho}{e} \var{v} \\
        \var{v} \in \mathbb{B}
    }
    { 
        \psem{s}{\rho}{\neg e} 
        \neg \var{v}
    }
    \and
    %% integer comparisons
    \inferenceRule{E-Cmp}
    {
        \psem{s}{\rho}{e_1} \var{v}_1 \\
        \psem{s}{\rho}{e_2} \var{v}_2 \\
        \var{v}_1, \var{v}_2 \in \mathbb{Z}
    }
    { 
        \psem{s}{\rho}{e_1~\cmp~e_2}
        \var{v}_1~\cmp~\var{v}_2
    }
    \and
    %% if then else case true
    \inferenceRule{E-ITETrue}
    {
        \psem{s}{\rho}{e_1} \meta{True} \\
        \psem{s}{\rho}{e_2} \var{v}_2
    }
    { 
        \psem{s}{\rho}{\ite{e_1}{e_2}{e_3}}
        \var{v}_2
    }
    \and
    %% if then else case false
    \inferenceRule{E-ITEFalse}
    {
        \psem{s}{\rho}{e_1} \meta{False} \\
        \psem{s}{\rho}{e_3} \var{v}_3
    }
    { 
        \psem{s}{\rho}{\ite{e_1}{e_2}{e_3}}
        \var{v}_3
    }
    \and
    %% equations case true
    \inferenceRule{E-EqTrue}
    {
        \psem{s}{\rho}{e_1} \var{v}_1 \\
        \psem{s}{\rho}{e_2} \var{v}_2 \\
        \var{v}_1 = \var{v}_2
    }
    { 
        \psem{s}{\rho}{e_1 = e_2}
        \meta{True}
    }
    \and
    %% equations case false
    \inferenceRule{E-EqFalse}
    {
        \psem{s}{\rho}{e_1} \var{v}_1 \\
        \psem{s}{\rho}{e_2} \var{v}_2 \\
        \var{v}_1 \not = \var{v}_2
    }
    { 
        \psem{s}{\rho}{e_1 = e_2}
        \meta{False}
    }
\end{mathpar}

\subsection*{Semantics of mapping expressions}
% semantics of mapping expressions
\begin{flushleft}
\fbox{$\psem{s}{\rho}{\var{m}} \var{v}$ }
\end{flushleft} 

\begin{mathpar}
   \inferenceRule{E-Exp}
    { 
      \psem{s}{\rho}{e} \var{v}
    }    
    { 
      \psem{s}{\rho}{e} \var{v}
    } 
    \and
    \inferenceRule{E-Mapping}
    {
        \psem{s}{\rho}{e_i} \var{v_i}  \text{ for $i \in [1,n]$}\\ 
        \var{v_i}\in \meta{meta}(\beta)  \text{ for $i \in [1,n]$}\\ 
        \psem{s}{\rho}{\var{m_j}} \var{u_j} \text{ for $j \in [1,n]$}\\ 
        f = (x \in \meta{meta}(\beta) \mapsto \meta{ if } x = \var{v_i} \meta { then } \var{u_i}
        \meta { else }  \meta {default}(\mu) )
    } 
    {  
        \psem{s}{\rho}{[\range{e_i \Rightarrow \var{m}_i}{i \in [1,n]}]_{\annot{\mapping{\beta}{\mu}}}} f
    }
    \and
    \inferenceRule{E-MappingUpd}
    {
        \psem{s}{\rho}{\var{ref}} \var{f} \\
        \psem{s}{\rho}{e_i} \var{v_i}  \text{ for $i \in [1,n]$}\\
        \var{v_i}\in \meta{meta}(\beta)  \text{ for $i \in [1,n]$}\\ 
        \psem{s}{\rho}{\var{m_j}} \var{u_j} \text{ for $j \in [1,n]$}\\ 
        g = (x \in \meta{meta}(\beta) \mapsto \meta{ if } x = \var{v_i} \meta { then } \var{u_i}
        \meta { else } f(x) )
    } 
    {  
        \psem{s}{\rho}{\var{ref}\map{\range{e_i \Rightarrow \var{m}_i}{i \in [1,n]}}{\annot{\mapping{\beta}{\mu}}}} g
    }
\end{mathpar} where $\meta{meta}(\iota) = \mathbb{Z}$, $\meta{meta}(\kw{bool})=\mathbb{B}$ and 
$\meta{meta}(\kw{address})=\meta{Addr}$, and

\begin{minipage}[t]{0.45\textwidth}
\[\begin{array}{r l}
% default definition
\meta{default}(\iota) = & \hspace{-0.5em} 0 \\ 
\meta{default}(\bool) = & \hspace{-0.5em} \meta{False}\\
\meta{default}(\kw{address}) = & \hspace{-0.5em} 0 \\
\meta{default}(\mapping{\beta}{\mu}) = & \hspace{-0.5em} x \in \meta{meta}(\beta) \mapsto \meta{default}(\mu) \\
\end{array}\]    
\end{minipage}

\subsection*{Semantics of slot expressions}
% semantics of slot expressions
\begin{flushleft}
\fbox{$\psem{s}{\rho}{\var{se}} (\var{v}, s')$}
\end{flushleft} 
      
\begin{mathpar}
    % mapping expressions
    \inferenceRule{E-MapExp}
    { 
      \psem{s}{\rho}{m} \var{v}
    }    
    { 
      \psem{s}{\rho}{m} (\var{v}, s)
    } 
    \and
    % slot references
    \inferenceRule{E-SlotRef}
    { 
      \psem{s}{\rho}{\var{ref}} (\var{v}, \_)
    }    
    { 
      \psem{s}{\rho}{\var{ref}} (\var{v}, s)
    } 
    \and
    % slot address
    \inferenceRule{E-SlotAddr}
    { 
      \psem{s}{\rho}{\var{se}} (\var{v}, s')
    }    
    { 
      \psem{s}{\rho}{\addr{\var{se}}} (\var{v}, s')
    } 
    \and
    % contract call
    \inferenceRule{E-Create}
    { 
    A_{\kw{Cnstr}} = \var{ctor} \\
    \meta{isPayable}(\var{ctor}) = \meta{false} \\
    \var{ctor}_{\kw{Iface}} = \range{x_i : \alpha_i}{i \in [1,n]} \\    
    \psem{s_{i-1}}{\rho}{\var{se_i}} (\var{v_i}, s_i) \text{  for $i \in [1,n]$} \\
    \rho' = \range{x_i \mapsto \var{v}_i}{i \in [1,n]} \cup \{\kw{caller} \mapsto \ell, \kw{origin} \mapsto \rho(\kw{origin}), \kw{callvalue} \mapsto 0 \}\\
    s_n ; \rho' ; \var{ctor}_{\kw{cases}} \bigstep_{\var{A}} (\ell', s') 
    }    
    { 
      \psem{s_0}{\rho}{\kw{new}~A(\range{\var{se_i}}{i \in [1,n]})} (\ell',s')
    }
    \and
    % contract call payable
    \inferenceRule{E-CreatePayable}
    { 
    A_{\kw{Cnstr}} = \var{ctor} \\
    \meta{isPayable}(\var{ctor}) = \meta{true} \\
    \var{ctor}_{\kw{Iface}} = \range{x_i : \alpha_i}{i \in [1,n]} \\    
    \psem{s_{i-1}}{\rho}{\var{se_i}} (\var{v_i}, s_i) \text{  for $i \in [1,n+1]$} \\
    \rho' = \range{x_i \mapsto \var{v}_i}{i \in [1,n]} \cup \{\kw{caller} \mapsto \ell, \kw{origin} \mapsto \rho(\kw{origin}), \kw{callvalue} \mapsto v_{n+1}\}\\
    s_{n+1} ; \rho' ; \var{ctor}_{\kw{cases}} \bigstep_{\var{A}} (\ell', s') 
    }    
    { 
      \psem{s_0}{\rho}{\kw{new}~A\{\kw{value}:~\var{se}_{n+1}\}(\range{\var{se_i}}{i \in [1,n]})} (\ell',s')
    } 
\end{mathpar}

\subsection*{Semantics of Creates and Updates}
% semantics of creates 
\begin{flushleft}
\fbox{$s ; \rho ; \many{\var{create}} \bigstep_{\var{Id}} (\ell, s')$ }
\end{flushleft} 

\begin{mathpar}
    \inferenceRule{E-Creates}
    {
        s_{i-1} ; \rho ; \var{se}_i \bigstep_{\dummyloc} (\var{v_i}, s_i) \quad \text{ for } i \in [1,n] \\
        \ell = \fresh{s_n} \\
        s_{n+1} = s_n[\ell \mapsto_{\var{Id}} \range{ x_i \mapsto \var{v_i}}{i \in [1,n]}] \\
    }
    { 
        s_0 ; \rho ; \range{\sigma_i~x_i \asgn \var{se}_i}{i \in [1,n]} \bigstep_{\var{Id}}
        (\ell, s_{n+1})
    }
\end{mathpar}

where
$\fresh{s} = \max({\dom{s}}) + 1$

% % semantics of update 
% \begin{flushleft}
% \fbox{$\psem{s}{\rho}{\var{upd}} s'$ }
% \end{flushleft} 
% %

% \begin{mathpar}
%     \inferenceRule{E-Update}
%     {
%         \psem{s}{\rho}{\var{se}} (\var{v}, s') \\ 
%         \ins{s' ; \rho ; \var{ref} ; \var{v}}{s''} \\ 
%     }
%     {
%         \psem{s}{\rho}{\var{ref} \asgn \var{se}} s''
%     } 
% \end{mathpar}   

% semantics of updates 
\begin{flushleft}
\fbox{$\psem{s}{\rho}{\many{\var{upd}}} s'$ }
\end{flushleft} 
\begin{mathpar}
\inferenceRule{E-Updates}
    { 
        \psem{s_{i-1}}{\rho}{\var{se}_i} (\var{v_i}, s_i) \quad \text{ for } i \in [1,n] \\
        \ins{s_{n+i-1} ; \rho ; \var{ref}_i ; \var{v_i}}{s_{n+i}} \quad \text{ for } i \in [1,n]
    }
    {
        \psem{s_0}{\rho}{\range{\var{ref}_i \asgn \var{se}_i}{i \in [1,n]}} s_{2n}
    }
\end{mathpar}

\subsection*{Semantics of Constructor Cases}
When the (well-typed) expressions evaluated do not refer to the location $\ell$ in the evaluation, we call the judgment with $\dummyloc$ in place of the location.

% semantics of constructor cases 
\begin{flushleft}
\fbox{$s ; \rho ; \var{ccases} \bigstep_{\var{Id}} (\ell, s')$}
\end{flushleft} 

\begin{mathpar}
\inferenceRule{E-CtorCases}
    { 
        s^{\kw{U}} ; \rho ; e_j \bigstep_{\dummyloc} \meta{True} \\
        s^{\kw{U}} ; \rho ; e_i \bigstep_{\dummyloc} \meta{False} \text{ for } i \neq j \\
        s^{\kw{U}} ; \rho;\many{\var{create}_j} \Downarrow_{\var{Id}} (\ell, s')
    }
    {
        s ; \rho ; \caseblock{e_i}{\kw{creates}~\many{create_i}} \bigstep (\ell, s')
    }
\end{mathpar}

\subsection*{Semantics of Transition Cases}
% semantics of transition cases 
\begin{flushleft}
\fbox{$\psem{s}{\rho}{\var{tcases}} (\var{v}, s')$ }
\end{flushleft} 

\begin{mathpar}
\inferenceRule{E-TransCases}
    { 
        \psem{s^{\kw{U}}}{\rho}{e_j} \meta{True} \\
        \psem{s^{\kw{U}}}{\rho}{e_i} \meta{False} \text{ for } i \neq j \\
        s^{\kw{U}} ; \rho;\many{\var{upd}_j} \Downarrow_{\ell} s' \\ 
        \psem{(s,s')^{\kw{T}}}{\rho}{\var{ret_j}} \var{v} 
    }
    {
        \psem{s}{\rho}{\caseblock{e_i}{\kw{updates}~\many{\var{upd_i}}~\kw{returns}~\var{ret_i}}} (\var{v}, s')
    }
\end{mathpar}

\end{flushleft}

\begin{flushleft}
\subsection*{Semantics of Constructors and Transitions}
In the following we write $I$ shorthand for $\many{\var{id} : \alpha}$ to emphasize its role as interface environment.
\begin{flushleft}
\fbox{
    $s ; \rho ;  \var{cnstr} \bigstep_{\var{Id}} (\ell, s')$
}
\end{flushleft}        
\begin{mathpar}
%% constructor (cases)
\inferenceRule{E-Ctor}
{
    s^{\kw{U}}; \rho; \var{pre}_i \Downarrow_{\dummyloc} \meta{True} \quad \text{ for } i \in [1,n]\\
    s; \rho; \var{cases} \bigstep_{\var{Id}} (\ell, s')
}
{ 
    s ; \rho ; \constrcase{\var{I}}
    {\{\var{pre}_i\}_{i \in [1,n]}}
    {\var{cases}}
    {\many{post}}
    \bigstep_{\var{Id}}
    (\ell, s')
}
\end{mathpar}

\begin{flushleft}
\fbox{
    $s ; \rho ; \var{trans} \Downarrow_{\ell} (\var{v}, s')$
}
\end{flushleft}        
\begin{mathpar}

%% transition cases
\inferenceRule{E-Trans}
{
    s^{\kw{U}}; \rho; \var{pre}_i \Downarrow_{\ell} \meta{True}  \quad \text{ for } i \in [1,n] \\
    s; \rho;  \var{tcases} \Downarrow_{\ell} (\var{v},s')
}
{
    s; \rho ; \transcase{\var{I}}{\{\var{pre}_i\}_{i \in [1,n]}}
        {\var{tcases}}{\many{post}}
    \Downarrow_{\ell}
    (\var{v}, s')
}

\end{mathpar} 

\subsection*{Transitions between states}

% State transitions is well-defined using the measure $\meta{len}(\Sigma,A)$.
\begin{flushleft}
\fbox{
    $\Sigma \vdash  s \rightsquigarrow s'$
}
\end{flushleft}

\begin{mathpar}
\inferenceRule{E-Step-Cnstr}
{   
    \exists \rho.~ (\var{cnstr} = \Sigmacode(A) \wedge 
    \Sigma \vdash \rho :_{s} \var{cnstr}_{\meta{iface}} \wedge (s ; \rho ; \var{cnstr} \bigstep_{A} (\ell, s')))
}
{ 
    \Sigma \vdash_A s \rightsquigarrow_{\ell}^{\kw{C}} s'
}
\and  
\inferenceRule{E-Step-Trans}
{   \exists \rho.~ 
    \Sigma \vdash \ell :_s A \wedge
    \var{trans} \in \Sigmatranss(A) \wedge 
    \Sigma \vdash \rho :_s \var{trans}_{\meta{iface}} \wedge 
    s ; \rho ; \var{trans} \bigstep_\ell (v, s')
}
{ 
    \Sigma \vdash_A s \rightsquigarrow_{\ell}^{\kw{B}} s'
}

\end{mathpar} 

We define $\Sigma \vdash s \rightsquigarrow s' = (\exists A~\ell, ~\Sigma \vdash_A s \rightsquigarrow_{\ell}^{\kw{C}} s') \vee
    (\exists A~\ell, ~\Sigma \vdash_A s \rightsquigarrow_{\ell}^{\kw{B}} s')$ 

\smallskip
    
We express the reachability of the state taking the reflexive-transitive closure $\rightsquigarrow^*$.
We say that a state $s$ is \emph{possible}, written $\possible{s}$, when 
$\Sigma \vdash \emptyset \rightsquigarrow^{*} s$.

% \begin{flushleft}
% \fbox{
%     $\Sigma \vdash_A \rho : s \rightsquigarrow^{*} s''$
% }
% \end{flushleft}

% We also define the following relation to express the states 
% reachable after creating a constructor at state $s$ using a 
% well-typed environment $\rho$.

% \begin{mathpar}
% \inferenceRule{E-CtorReach}
% {
%     \Sigmacode(A) = \var{cnstr} \\
%     \possible{s} \\
%     \Sigma \vdash \rho :_s \var{cnstr}_{\meta{iface}} \\    
%     s ; \rho ; \var{cnstr} \bigstep_{\var{A}} (\ell, s') \\
%     \Sigma \vdash s' \rightsquigarrow^* s''
% }
% { 
%     \Sigma \vdash_A \rho : s \rightsquigarrow^{*}_\ell s''
% }
% \end{mathpar}

\subsection*{Semantics of Constructor Postconditions}

%% constructor postconditions without reachability

\begin{flushleft}
\fbox{
    $s ; \rho; \var{cnstr} \Downarrow^{\kw{post}} \meta{True}$
}
\end{flushleft}

\begin{mathpar}
\inferenceRule{E-CtorPost}
{
    \var{constr}_{\kw{post}} = \range{\var{post}_i}{i \in [1,n]} \\
    s ;\rho; \var{constr} \bigstep_{\var{A}} (\ell, s') \implies
    \psem{(s')^{\kw{U}}}{\rho}{\var{post_i}} \meta{True} \text{ for $i \in [1,n]$}
}
{ 
    s ; \rho; \var{cnstr} \Downarrow^{\kw{post}}_{\var{A}} \meta{True}
}
\end{mathpar}

%% constructor postconditions with reachability

\begin{flushleft}
\fbox{
    $\Sigma \vdash A \downarrow^{\kw{Cpost}} \meta{True}$
}
\end{flushleft}

\begin{mathpar}
\inferenceRule{E-CtorPostCheck}
{
    \var{cnstr} = \Sigmacode(A) \\
    \forall s, \rho.~\possible{s} ~\wedge~(\Sigma \vdash \rho :_s \var{cnstr}_{\kw{iface}}) \implies
     s ; \rho; \var{cnstr} \Downarrow^{\kw{post}}_{\var{A}} \meta{True}
}
{ 
    \Sigma \vdash A \downarrow^{\kw{Cpost}} \meta{True}
}
\end{mathpar}

\subsection*{Semantics of Transition postconditions}

%% transition postconditions without reachability

\begin{flushleft}
\fbox{
    $s ; \rho; \var{trans} \bigstep^{\kw{post}}_\ell \meta{True}$
}
\end{flushleft}

\begin{mathpar}
\inferenceRule{E-TransPost}
{
    \var{trans}_{\kw{post}} = \range{\var{post}_i}{i \in [1,n]} \\
    s ;\rho ; \var{trans} \bigstep_\ell (\var{v},s') \implies
    \psem{(s , s')^{\kw{T}}}{\rho}{\var{post_i}} \meta{True} \text{ for $i \in [1,n]$}
}
{ 
    s ; \rho; \var{trans} \bigstep^{\kw{post}}_\ell \meta{True}
}
\end{mathpar}

\begin{flushleft}
\fbox{
    $\Sigma \vdash A \downarrow^{\kw{Bpost}} \meta{True}$}
\end{flushleft}

%% transition postconditions with reachability
\begin{mathpar}
\inferenceRule{E-TransPostCheck}
{
    \Sigmatranss(A) = \range{\var{trans_i}}{i \in [1,n]} \\\\  
    \forall~i \in [1,n],~s, ~\ell, ~\rho. ~ \possible{s} \implies
     \Sigma \vdash \ell :_s A \implies  \\
     \Sigma \vdash \rho :_{s} \var{trans_i}_{\kw{iface}} \implies
      s ; \rho; \var{trans_i} \bigstep^{\kw{post}}_\ell \meta{True}
}
{ 
    \Sigma \vdash A \downarrow^{\kw{Bpost}} \meta{True}
}
\end{mathpar}

\subsection*{Semantics of Invariants}

\begin{flushleft}
\fbox{
    $\Sigma \vdash A \downarrow^{\kw{inv}} \meta{True}$
}
\end{flushleft}

%% invariants checking
\begin{mathpar}
\inferenceRule{E-InvCheck}
{   
    %\var{cnstr} = \Sigmacode(A) \\
    \forall \, s, \, \ell. \, \possible{s} \implies
    \Sigma \vdash \ell :_s A \implies
    s ; \cdot ; A_{\meta{inv}} \bigstep_\ell \meta{True}
    % \forall s, \rho, \ell, s'.~\possible{s} ~\wedge~(\Sigma \vdash \rho :_s \var{cnstr}_{\meta{iface}}) \implies
    %  \Sigma \vdash_A \rho : s \rightsquigarrow^{*}_{\ell} s' \implies
    % % \forall s, s',\rho.~(\Sigma \vdash \rho :_s \var{cnstr}_{\meta{iface}} \implies 
    % % s ; \rho ; \var{cnstr} \bigstep (\ell, s') \implies 
    % % \forall s''.~\Sigma \vdash s' \rightsquigarrow^* s'' \implies 
    % 
}
{
    \Sigma \vdash  A \downarrow^{\kw{inv}} \meta{True}
}
\end{mathpar}

\subsection*{Semantics of Contracts}

\begin{flushleft}
\fbox{
    $\Sigma \vdash A \downarrow \meta{True}$
}
\end{flushleft}

%% invariants checking
\begin{mathpar}
\inferenceRule{E-Contract}
{   
    \Sigma \vdash  A \downarrow^{\kw{Cpost}} \meta{True} \\
    \Sigma \vdash  A \downarrow^{\kw{Bpost}} \meta{True} \\
    \Sigma \vdash  A \downarrow^{\kw{inv}} \meta{True} \\
}
{ 
    \Sigma \vdash  A \downarrow \meta{True}
}
\end{mathpar}

\end{flushleft}

\section{Value and Environment Typing}
% Value typing
\begin{flushleft}
    \fbox{$\vdash v : \beta $}
\end{flushleft}
\begin{mathpar}
% Addresses 
\inferenceRule{V-Addr}
{
    \ell \in \meta{Addr}
}
{
   \vdash \ell : \kw{address}
}
\and
% Booleans 
\inferenceRule{V-Bool}
{
    b \in \mathbb{B}
}
{
    \vdash b : \kw{bool}
}
\and
% integers
\inferenceRule{V-Int}
{
    n \in \mathbb{Z} \\
    (\meta{min}(\iota) \leq n \leq \meta{max}(\iota)) ~\vee~ \iota = \intt{} \\
}
{
    \vdash n : \iota
}
\end{mathpar}

\begin{flushleft}
    \fbox{$\vdash v : \mu $}
\end{flushleft}

\begin{mathpar}
% mappings
\inferenceRule{V-BaseValMu}
{
    \vdash v : \beta \\
}
{
    \vdash v : \beta
}
\and
\inferenceRule{V-Mapping}
{
    f \in \meta{meta}(\beta) \rightarrow \meta{Value} \\
    \forall \, v. ~\;  \vdash v : \beta \Rightarrow \;
     \vdash f(v) : \mu
}
{
     \vdash f : \mapping{\beta}{\mu}
}
\end{mathpar}

% Value typing
\begin{flushleft}
    \fbox{$\Sigma \vdash v :_s \alpha $}
\end{flushleft}
\begin{mathpar}
% Addresses 
\inferenceRule{V-BaseValAlpha}
{
     \vdash v : \beta
}
{
    \Sigma \vdash v :_s \beta
}
\and
% contracts
\inferenceRule{V-AddrIsContract}
{ \ell \in \dom{s} \\
   A \in \dom{\Sigmastore} \\
   s(\ell).{\kw{type}} = A \\
   % If we keep the type in the store we can have stronger 
   % uniqueness of typing, which may be useful, even necessary
   % we keep the type of the contract in the state
   % \stype{\ell} = A \\
   %    
  (\forall \, x. ~\, x \in \dom{s(\ell)} \Leftrightarrow (\exists \, \sigma. ~\, \Sigmastore(A)(x) = \sigma)) \\
  (\forall \, x,\sigma. ~\, \Sigmastore(A)(x) = \sigma ~\Rightarrow ~ \Sigma \vdash s(\ell)(x) :_s \sigma) 
}
{
    \Sigma \vdash \ell :_s \kw{address}_A
}
\end{mathpar}

\begin{flushleft}
    \fbox{$\Sigma \vdash v :_s \sigma $}
\end{flushleft}

\begin{mathpar}
\inferenceRule{V-MappingVal}
{
     \vdash v : \mu \\
}
{
    \Sigma \vdash v :_s \mu
}
\and
\inferenceRule{V-ABIVal}
{
    \Sigma \vdash v :_s \alpha \\
}
{
    \Sigma \vdash v :_s \alpha
}
\and
% contracts
\inferenceRule{V-Contract}
{ 
    \Sigma \vdash \ell :_s \kw{address}_A
%     \ell \in \dom{s} \\
%    A \in \dom{\Sigmastore} \\
%  (\forall \, x. ~\, x \in \dom{s(\ell)} \Leftrightarrow (\exists \, \sigma. ~\, \Sigmastore(A)(x) = \sigma)) \\
%   (\forall \, x,\sigma. ~\, \Sigmastore(A)(x) = \sigma ~\Rightarrow ~ \Sigma \vdash s(\ell) (x) :_s \sigma) 
}
{
    \Sigma \vdash \ell :_s A
}
\end{mathpar}

We lift the above definition to optional slot types.
\begin{flushleft}
    \fbox{$\Sigma \vdash v :_s \maybe{\sigma} $}
\end{flushleft}

\begin{mathpar}
% Bot 
\inferenceRule{V-None}
{
    \\
}
{
    \Sigma \vdash v :_s \bot
}
% Bot 
\and 
\inferenceRule{V-Some}
{
   \Sigma \vdash v :_s \sigma
}
{
    \Sigma \vdash v :_s \sigma
}
\end{mathpar}

% Env typing
\begin{flushleft}
\fbox{$\Sigma \vdash \rho :_s I $}
\end{flushleft}
\begin{mathpar}
% s is not really relevant here as \rho only has base values
\inferenceRule{V-Env}
{
    \dom{\rho} = \dom{I} \cup \{\kw{caller}, \kw{origin}, \kw{callvalue}\} \\
    \forall \, x \in \dom{I}. ~\, \Sigma \vdash \rho(x) :_s I(x) \\ 
    \vdash \rho(\kw{caller}) : \kw{address} \\
    \vdash \rho(\kw{origin}) : \kw{address} \\
    \vdash \rho(\kw{callvalue}) : \kw{uint256} 
}
{
    \Sigma \vdash \rho :_s I
}
\end{mathpar}

% The following judgement asserts that a given state, environment, and entry location
% satisfy some pointer assertions.

% \begin{flushleft}
%     \fbox{$\Sigma \vdash (s, \rho, \ell) : P $}
% \end{flushleft}
% \begin{mathpar}
% \inferenceRule{V-Pointers}
% {
%     \forall \, i \in [1,n]. ~\, s^{\kw{U}}; \rho; \var{ref_i} \Downarrow_{\ell} (\ell_i, \kw{U}) \\
% %
%     \Sigma \vdash \ell_i :_s A_i
% }
% {
%     \Sigma \vdash (s, \rho, \ell) : \range{\var{ref_i} \mapsto A_i}{i \in [1,n]}
% }
% \end{mathpar}

\section{Semantic Entailment}
We formalize a notion of semantic entailment to check whether a collection of
boolean expressions is always true in a well-typed configuration. We use this
notion to perform necessary semantic checks during typing. 

% Semantic entailment
\begin{flushleft}
    \fbox{$\Sigma; I; \Phi \vDash_{\maybe{A}} \many{e} $}
\end{flushleft}
\begin{mathpar}
\inferenceRule{ValidExps}
{
    \forall \, s, \rho, \ell.  
    ~\, \Sigma \vdash \rho :_s I~~\wedge~~
    \Sigma \vdash \ell :_s \maybe{A}~~\wedge~~ 
    \psem{s}{\rho}{\Phi} \meta{True}
    ~~\Rightarrow~~ 
    \many{\psem{s}{\rho}{e} \meta{True}}
}
{
    \Sigma; I; \Phi \vDash_{\maybe{A}} \many{e}
}
\end{mathpar}

% Semantic entailment
\begin{flushleft}
    \fbox{$\Sigma; I; \Phi \vDash_{\maybe{A}} (\range{se_i}{i \in [1,n]}, \many{pre'}) $}
\end{flushleft}
\begin{mathpar}
\inferenceRule{ValidIffs}
{
    \forall \, s_0, \rho, \ell, \range{v_i}{i \in [1,n]} . ~\, \Sigma \vdash \rho :_s I~~\wedge~~
        \Sigma \vdash \ell :_{s_0} \maybe{A}~~\wedge~~ \\\\
        \psem{s_0}{\rho}{\Phi} \meta{True}~~\wedge~~  \\\\
        (\forall \, i \in [1,n]. ~\, \psem{s_{i-1}}{\rho}{se_i}(v_i, s_i)) ~~\Rightarrow~~ \\\\
        \many{\psemdummy{s_n}{\range{x_i \mapsto v_i}{i \in [1,n]}}{\var{pre'}} \meta{True}}
}
{
    \Sigma; I; \Phi \vDash_{\maybe{A}} (\range{se_i}{i \in [1,n]}, \many{pre'})
}
\end{mathpar}

\section{Typing}
\subsection*{Typing Environments}
%% Typing environments
\begin{minipage}{1\textwidth}
\deftable{
    \defline{(Contract Storage)}{C}{(\many{x : \sigma})}
    \defline{(State)}{\Sigma}{  \{ ~ \sstorage : (\many{\var{Id} : C}), }
        \defnewline{\hspace{0.24cm} \scnstrs :  
        (\many{\var{Id} : \var{cnstr}}) 
         } 
        \defnewline{\hspace{0.24cm} \stranss :
        (\many{\var{Id} : \var{\many{trans}} })
        ~\} }
    \defline{(Calldata)}{I}{(\many{\var{x} : \alpha})} 
    }
\end{minipage}

\subsection*{Top-level Specification and Contract Judgments}
%% Top-level spec
\begin{flushleft}
%% Spec judgement
% \begin{minipage}[t]{0.42\textwidth}
    \begin{flushleft}
    \fbox{$\vdash \var{spec} : \Sigma$}
    \end{flushleft}
\begin{mathpar}
\inferenceRule{T-Spec}
{
    \Sigma_0 = \varnothing \\
    \Sigma_{i-1} \vdash \var{contract_i} : \Sigma_i \\
}
{\vdash \range{\var{contract_i}}{i\in [1,n]} : \Sigma_n}
\end{mathpar}
% \end{minipage}
%%
%%
%% Contract judgment 
% \begin{minipage}[t]{0.57\textwidth}
\begin{flushleft}
    \fbox{$\Sigma \vdash \var{contract} : \Sigma'$}
\end{flushleft}
\begin{mathpar}
    \inferenceRule{T-Contract}
    {
        \Sigma \vdash_{\var{Id}} \var{cnstr} : C \\ 
        \Sigma' = \Sigma ~\meta{with}~ \{ \sstorage = (\Sigmastore, \var{Id} : C), 
             \scnstrs = (\Sigmacode, \var{Id} : \var{cnstr}) \} \\
        \many{(\Sigma' \vdash_{\var{Id}}  \var{trans})} \\ 
        \many{(\Sigma'; \cdot; \kw{true} \vdash_{\var{Id},\kw{U}} \var{inv} : \bool)} \\ 
        \Sigma'' = \Sigma' ~\meta{with}~ \{\stranss = (\Sigmatranss, \var{Id} : \many{\var{trans}})\}
    }
    {
        \Sigma \vdash
        \contract{\var{Id}}
        {\var{cnstr}}
        {\many{\var{trans}}}
        {\kw{invariants}~\many{\var{inv}}}
        : \Sigma''
    }
\end{mathpar}
% \end{minipage}
%%
%%
%% Constructors
\subsection*{Constructor Judgment}
\begin{flushleft}
    \fbox{$\Sigma \vdash_{\var{Id}} \var{constructor} : C $}
    \end{flushleft}
\begin{mathpar}
%% constructor case
\inferenceRule{T-Ctor}{
%%  interface is WF
\Sigma \vdash I ~\meta{wf} \\
% \Sigma; I \vdash_{\none}  P \\
%% preconds are well typed
\many{(\Sigma ; \var{I}; \kw{true} \vdash_{\none,\kw{U}} \var{pre} : \kw{bool})}  \\\\
%% case conditions are well typed
(\forall \, i \in [1,n].~\,\Sigma;  \var{I}; \kw{true} \vdash_{\none,\kw{U}} e_i : \kw{bool}) \and
%% Current constraints
\Phi_i = e_i \wedge \bigwedge_{\var{pre} \in \many{\var{pre}}}{\var{pre}} \\     
%% creates are well typed
(\forall \, i \in [1,n]. ~\, \Sigma; \var{I}; \Phi_i \vdash_{\var{Id}} \many{\var{create}_i} : C) \\\\
% defining \Sigma'
\Sigma' = \Sigma ~\meta{with}~ \{ ~ \sstorage = (\Sigmastore, \var{Id} : C) ~ \} \\\\
%% no overlap with interface variables
\dom{C} \cap \dom{I} = \emptyset \\
%% postconds are well typed
\many{( \Sigma'; \var{I}; \kw{true} \vdash_{\var{Id},\kw{U}} \var{post}: \kw{bool} )}   \\\\
%% cases are consistent
\Sigma; I; \kw{true} \vDash_{\bot} 
\bigwedge_{\var{pre \in \many{\var{pre}}}}{\var{pre}} \Rightarrow (\bigvee_{i \in [1,n]} e_i) ~~ \wedge ~~(\bigwedge_{i \in [1,n], j \in [i+1, n]} \neg (e_i \wedge e_j))
}
{\small \Sigma \vdash_{\var{Id}} 
  \constrinlinecase{\var{I}}{\many{\var{pre}}}{e}{\many{\var{create}_i}}{\many{post}} 
  : C
}
\end{mathpar}
\subsection*{Transition Judgment}
% transition
\begin{flushleft}
    \fbox{$\Sigma \vdash_{\var{Id}} \var{transition} $}
    \end{flushleft}
    \begin{mathpar}
        \inferenceRule{T-Trans}{
        %%  interface is WF
        \Sigma \vdash I ~\meta{wf} \\
        %% preconds are well typed
        \many{(\Sigma ; \var{I}; \kw{true} \vdash_{\var{Id},\kw{U}} \var{pre} : \kw{bool})}  \\\\
        %% case conditions are well typed
        (\forall \, i \in [1,n]. ~\, \Sigma; \var{I}; \kw{true} \vdash_{\var{Id},\kw{U}} e_i : \kw{bool}) \and
        %% Current constraints
        \Phi_i = e_i \wedge \bigwedge_{\var{pre} \in \many{\var{pre}}}{\var{pre}} \\ 
        %% updates are well typed
        (\forall \, i \in [1,n].~\, \many{\Sigma; \var{I}; \Phi_i  \vdash_{\var{Id}} \many{\var{upd}_i}}) \\\\
        %% returns are well typed
        (\forall \, i \in [1,n]. ~\, \Sigma; \var{I} \vdash_{\var{Id},\kw{T}} \var{ret}_i : \alpha) \and
        %% postconds are well typed
        \many{( \Sigma; \var{I}; \kw{true} \vdash_{\var{Id},\kw{T}} \var{post}: \kw{bool} )}   \\\\
        %% cases are consistent
        \Sigma; I; \kw{true} \vDash_{\var{Id}} 
        \bigwedge_{\var{pre \in \many{\var{pre}}}}{\var{pre}} \Rightarrow (\bigvee_{i \in [1,n]} e_i) ~~ \wedge ~~(\bigwedge_{i \in [1,n], j \in [i+1, n]} \neg (e_i \wedge e_j))
        }
        {
        {\small \Sigma \vdash_{Id} 
        \transinlinecase{\var{I}}{\many{\var{pre}}}
        {e}{\many{\var{upd}}}{\var{ret}}{\many{post}}}
        }
    \end{mathpar}

% \subsection*{Pointer Judgment}
% This judgment ensures that a pointer collection is well typed, meaning that
% references are always addresses and contract types are already defined.

% Note that, references that appear in pointers cannot use coercions themselves as
% they are typed in an empty pointer environment. 
% %
% In practice, this means that references of pointer constrains in constructors
% can only be calldata variables (we formally prove this as a lemma). % TODO
% %

\subsection*{Type Well-formedness}

\begin{flushleft}
    \fbox{$\Sigma \vdash \alpha ~\meta{wf}$}
    \end{flushleft}

    \begin{mathpar}
        \inferenceRule{WFInt}{
            \iota \not = \intt{}
        }
        {
            \Sigma \vdash \iota ~\meta{wf}
        }

        \inferenceRule{WFContractAddr}{
            \alpha = \kw{address}_A \Rightarrow A \in \dom{\Sigmastore}
        }
        {
            \Sigma \vdash \alpha ~\meta{wf}
        }
    \end{mathpar}

\begin{flushleft}
    \fbox{$\Sigma \vdash \sigma ~\meta{wf}$}
    \end{flushleft}
    \begin{mathpar}
        \inferenceRule{T-WFContract}{
             A \in \dom{\Sigmastore}
        }
        {
            \Sigma \vdash A ~\meta{wf}
        }
        \and
        \inferenceRule{T-WFAlpha}{
             \Sigma \vdash \alpha ~\meta{wf}
        }
        {
            \Sigma \vdash \alpha ~\meta{wf}
        }
        \and
        \inferenceRule{T-WFMapping}{
            \\
        }
        {
            \Sigma \vdash \mu ~\meta{wf}
        }
    \end{mathpar}

\begin{flushleft}
    \fbox{$\Sigma \vdash I ~\meta{wf}$}
    \end{flushleft}

    \begin{mathpar}
        \inferenceRule{T-WFEnv}{
            \forall~i~\in~[1,n],~\Sigma \vdash \alpha_i ~\meta{wf}
        }
        {
            \Sigma \vdash \range{\var{x_i} : \alpha_i}{i \in [1,n]} ~\meta{wf}
        }
    \end{mathpar}

\subsection*{Create and Update Judgments}

% Create
\begin{minipage}{0.62\textwidth}
\begin{flushleft}
    \fbox{$\Sigma ; I; \Phi \vdash_{\var{Id}} \many{\var{create}} : C $}
    \end{flushleft}
\begin{mathpar}
\inferenceRule{T-Creates}
{   \forall \, i \in [1,n].~\,\Sigma \vdash \sigma_i ~\meta{wf} \\
    C = \range{x_i : \sigma_i}{i \in [1,n]} \\
    \kw{balance} : \kw{uint256} \in C \\
    \forall \, i \in [1,n]. ~\,\Sigma; I; \Phi \vdash_{\bot} \var{se}_i : \sigma_i \\
}
{
\Sigma; I; \Phi  \vdash_{\var{Id}} \range{\sigma_i ~\var{x}_i \asgn \var{se}_i}{i \in [1,n]} : C
}
\end{mathpar}
\end{minipage}

% One update
\begin{minipage}{0.49\textwidth}
\begin{flushleft}
    \fbox{$\Sigma ; I ; \Phi \vdash_{\var{Id}} \var{update} $}
    \end{flushleft}
\begin{mathpar}
\inferenceRule{T-Update}
{
    \Sigma; I;\vdash^{\kw{S}}_{\var{Id}, \kw{U}} \var{ref} : \sigma \\ %% TODO fill
    \Sigma; I; \Phi  \vdash_{\var{Id}} \var{se} : \sigma \\
}
{
\Sigma; I; \Phi \vdash_{\var{Id}} \var{ref_{\annot{\sigma}}} \asgn \var{se}
}
\end{mathpar}
\end{minipage}
% Many updates
\begin{minipage}{0.48\textwidth}
\begin{flushleft}
    \fbox{$\Sigma ; I ; \Phi \vdash_{\var{Id}} \many{\var{update}} $}
    \end{flushleft}
\begin{mathpar}
\inferenceRule{T-Updates}
{
    \Sigma; I;  \Phi \vdash_{\var{Id}} \var{ref_{i}} \asgn \var{se_i} \text{ for $i \in [1,n]$}\\
    \neg (\var{ref}_j \specificeq \var{ref}_i) \text{ for $i \in [1,n]$ and $j \in [1,i]$}\\
}
{
\Sigma; I; \Phi \vdash_{\var{Id}} \range{\var{ref_{i}} \asgn \var{se_i}}{i \in [1,n]}
}
\end{mathpar}
\end{minipage}
where the relation $\specific$ is defined on the syntax as follows
\[\begin{array}{l}
    \var{ref}.x \specific \var{ref} \\
    \var{ref}_1 \specific \var{ref}_2 \implies \var{ref}_1.x \specific \var{ref}_2
\end{array}\]

We define $\specificeq$ to be the reflexive closure of $\specific$.

Note that in the~\rref{T-Ctor} we require that the created storage contains $\kw{balance}$ for 
the constructor to be well-typed, but it is the user's responsibility to ensure that the
storage is properly initialized.

\subsection*{Reference Judgment}
The reference judgement is annotated with a tag $k$ that takes values from the
set $\{ \kw{S}, \kw{N} \}$. $\kw{S}$ denotes a storage reference that is
reachable from the current contract's storage, without any coercion. References
tagged with $\kw{N}$ are references that are calldata references, references
that involve coercions or environment variables.

Only references tagged with $\kw{S}$ can be updated. %% TODO: motivate

\begin{flushleft}
    \fbox{$\Sigma ; I \vdash^{k}_{\var{\maybe{Id}}, t} \var{ref} : \sigma $}  
    \quad \quad where $k \in \{ \kw{S}, \kw{N} \}$
    \end{flushleft}
\begin{mathpar}
    \inferenceRule{T-Calldata}
    {
       x : \alpha \in I 
    }
    { 
      \Sigma; I \vdash^{\kw{N}}_{\maybe{\var{Id}}, t} x : \alpha
    }    
    \and\and
    \inferenceRule{T-Storage}
    {
       x : \sigma \in \Sigmastore(\var{Id}) \\
       x \not \in \dom{I}
    }
    { 
      \Sigma; I \vdash^{\kw{S}}_{\var{Id}, \kw{U}} x : \sigma
    }    
    \and   
    \inferenceRule{T-StoragePre}
    {
       x : \sigma \in \Sigmastore(\var{Id}) \\
       x \not \in \dom{I}
    }
    { 
      \Sigma; I \vdash^{\kw{S}}_{\var{Id}, \kw{T}} \pre{x} : \sigma
    }    
    \and\and
    \inferenceRule{T-StoragePost}
    {
       x : \sigma \in \Sigmastore(\var{Id}) \\
       x \not \in \dom{I}
    }
    { 
      \Sigma; I \vdash^{\kw{S}}_{\var{Id}, \kw{T}} \post{x} : \sigma
    }    
    \and   
    \inferenceRule{T-Coerce}
    {
        % NOTE only calldata variables can be casted to contracts?
        \Sigma; I \vdash^{k}_{\maybe{\var{Id}}, t} \var{ref} : \kw{address}_A \\
    }
    { 
      \Sigma; I \vdash^{\kw{N}}_{\maybe{\var{Id}}, t} \var{ref} ~\kw{as}~ A : A
    } 
    \and
    \inferenceRule{T-Upcast}
    {
        \Sigma; I \vdash^{k}_{\maybe{\var{Id}}, t} \var{ref} : \kw{address}_A \\
    }
    { 
        \Sigma; I \vdash^{\kw{N}}_{\maybe{\var{Id}}, t} \var{ref}  : \kw{address}
    } 
    \and\and
    \inferenceRule{T-Field}
    {
       \Sigma; I \vdash^{k}_{\maybe{\var{Id}}, t}  \var{ref} : \var{A} \\
       \Sigmastore(A)(x) = \sigma 
    }
    {
      \Sigma; I \vdash^{k}_{\maybe{\var{Id}}, t} \var{ref}.x : \sigma
    }
    \and
    \inferenceRule{T-MapIndex}
    {
       \Sigma; I \vdash^{k}_{\maybe{\var{Id}}, t}  \var{ref} : \mapping{\beta}{\mu} \\\\
       \Sigma; I \vdash_{\maybe{\var{Id}}, t}  e : \beta \\
    }
    { 
      \Sigma; I \vdash^{\kw{N}}_{\maybe{\var{Id}}, t} \var{ref}[e] : \mu
    }
    \and
    \inferenceRule{T-Environment}
    {
       \Sigma; I \vdash_{\maybe{\var{Id}}}  \var{env} : \alpha \\
    }
    { 
      \Sigma; I \vdash^{\kw{N}}_{\maybe{\var{Id}}, \kw{U}}  \var{env} : \alpha
    }    
\end{mathpar}

\subsection*{Environment References Judgment}
\begin{flushleft}
    \fbox{$\Sigma; I \vdash_{\maybe{\var{Id}}}  \var{env} : \alpha$}
\end{flushleft}        
\begin{mathpar}
    \inferenceRule{T-Caller}
    { }{\Sigma; I \vdash_{\maybe{\var{Id}}} \kw{caller} : \kw{address}}
    \and
    \inferenceRule{T-Origin}
    { }{\Sigma; I \vdash_{\maybe{\var{Id}}} \kw{origin} : \kw{address}}
    \and
    \inferenceRule{T-Callvalue}
    { }{\Sigma; I \vdash_{\maybe{\var{Id}}} \kw{callvalue} : \kw{uint256}}
    \and
    \inferenceRule{T-This}
    { }{\Sigma; I \vdash_{\var{Id}} \kw{this} : \kw{address}_{\var{Id}}}
\end{mathpar}

\subsection*{Slot Expression Judgment}
\begin{flushleft}
    \fbox{$\Sigma ; I; \Phi \vdash_{\maybe{\var{Id}}} \var{se} : \sigma $}
\end{flushleft}        
\begin{mathpar}
    \inferenceRule{T-MapExp}
    { 
      \Sigma; I; \Phi \vdash_{\maybe{\var{Id}}} m : \mu
    }    
    { 
      \Sigma; I; \Phi \vdash_{\maybe{\var{Id}}} m : \mu
    } 
    \and
    \inferenceRule{T-SlotRef}
    { 
      \Sigma; I \vdash^{k}_{\maybe{\var{Id}}, \kw{U}} \var{ref} : \var{A}
    }    
    { 
      \Sigma; I; \Phi \vdash_{\maybe{\var{Id}}} \var{ref_{\annot{\var{A}}}} : \var{A}
    } 
    \and
    \inferenceRule{T-SlotAddr}
    {
      \Sigma; I; \Phi \vdash_{\maybe{\var{Id}}} \var{se} : \var{A}
    }
    {
      \Sigma; I; \Phi \vdash_{\maybe{\var{Id}}} \addr{\var{se}} : \kw{address}_A
    }
    \and
    \inferenceRule{T-Create}
    %% TODO A_interface needs to be tracked by the typing environment
    { 
        \Sigmacode(A) = \var{ctor} \\
        \meta{isPayable}(\var{ctor}) = \meta{false} \\\\
         \var{ctor}_{\kw{Iface}} = \range{x_i : \alpha_i}{i \in [1,n]}\\
        \forall \, j \in [1,n]. ~\, \Sigma; I; \Phi \vdash_{\maybe{\var{Id}},\kw{U}} \var{se_j} : \alpha_j \\\\
        % %% constructor preconditions are guaranteed to hold
        \Sigma; I;  \Phi \vDash_{\maybe{\var{Id}}} (\range{se_i}{i \in [1,n]}, \var{ctor}_{\kw{Iff}}) 
        % %% syntactic check for pointers
        % (\forall \, i \in [1,n], B. ~\, x_i \mapsto B \in A_{\kw{Pointers}} \Rightarrow
        %     \Sigma; I; P \vdash_{\maybe{\var{Id}}} \var{se_i}  ~\meta{is}~ B)
    }
    {
        \Sigma; I; \Phi \vdash_{\maybe{\var{Id}}} \kw{new}~A(\var{se_1}, \dots, \var{se_n}) : A
    }
    \and
    \inferenceRule{T-CreatePayable}
    %% TODO A_interface needs to be tracked by the typing environment
    { 
        \Sigmacode(A) = \var{ctor} \\
        \meta{isPayable}(\var{ctor}) = \meta{true} \\\\
        \var{ctor}_{\kw{Iface}} = \range{x_i : \alpha_i}{i \in [1,n]} \\
        \forall \, j \in [1,n+1]. ~\, \Sigma; I; \Phi \vdash_{\maybe{\var{Id}},\kw{U}} \var{se_j} : \alpha_j \\
        \Sigma; I; \Phi \vdash_{\maybe{\var{Id}},\kw{U}} \var{se_{n+1}} : \kw{uint256} \\\\
        % %% constructor preconditions are guaranteed to hold
        \Sigma; I;  \Phi \vDash_{\maybe{\var{Id}}} (\range{se_i}{i \in [1,n]}, \var{ctor}_{\kw{Iff}}) 
        % %% syntactic check for pointers
        % (\forall \, i \in [1,n], B. ~\, x_i \mapsto B \in A_{\kw{Pointers}} \Rightarrow
        %     \Sigma; I; P \vdash_{\maybe{\var{Id}}} \var{se_i}  ~\meta{is}~ B)
    }
    {
        \Sigma; I; \Phi \vdash_{\maybe{\var{Id}}} \kw{new}~A\{\kw{value}:~\var{se_{n+1}}\}(\var{se_1}, \dots, \var{se_n}) : A
    }
\end{mathpar}

\subsection*{Mapping Expression Judgment}
\begin{flushleft}
    \fbox{$\Sigma ; I; \Phi \vdash_{\maybe{\var{Id}}} \var{m} : \mu $}
\end{flushleft}        
\begin{mathpar}
    \inferenceRule{T-Exp}
    { 
      \Sigma; I; \Phi \vdash_{\maybe{\var{A}}, \kw{U}} e : \beta
    }    
    { 
      \Sigma; I; \Phi \vdash_{\maybe{\var{A}}} e : \beta
    } 
    \and
    \inferenceRule{T-Mapping}
    {
        \forall \, i \in [1,n]. \\ 
        \Sigma; I; \Phi \vdash_{\maybe{\var{Id}},\kw{U}} e_i : \beta \\
        \Sigma; I; \Phi \vdash_{\maybe{\var{Id}}} \var{m}_i : \mu \\
    } 
    { \Sigma; I; \Phi \vdash_{\maybe{\var{Id}}} \map{\range{e_i \Rightarrow \var{m}_i}{i \in [1,n]}}{\annot{\mapping{\beta}{\mu}}} :
    \mapping{\beta}{\mu}
    }
    \and
    \inferenceRule{T-MappingUpd}
    {
        \Sigma; I; \Phi \vdash_{\maybe{\var{Id}},\kw{U}} \var{ref} : \mapping{\beta}{\mu} \\
        \forall \, i \in [1,n]. \\ 
        \Sigma; I; \Phi \vdash_{\maybe{\var{Id}},\kw{U}} e_i : \beta \\
        \Sigma; I; \Phi \vdash_{\maybe{\var{Id}}} \var{m}_i : \mu \\
    } 
    { \Sigma; I; \Phi \vdash_{\maybe{\var{Id}}} \var{ref}\map{\range{e_i \Rightarrow \var{m}_i}{i \in [1,n]}}{\annot{\mapping{\beta}{\mu}}} :
    \mapping{\beta}{\mu}
    }
\end{mathpar}

\subsection*{Expression Judgment}
\begin{flushleft}
\fbox{$\Sigma; I; \Phi \vdash_{\var{\maybe{Id}},t} \var{expr} : \beta $}
\end{flushleft}
\begin{mathpar}
\inferenceRule{T-Int}
{ n \in \mathbb{Z}\\
  (\meta{min}(\iota) \leq n \leq \meta{max}(\iota)) ~\vee~ \iota = \intt{}
}
{ 
  \Sigma; I; \Phi \vdash_{\var{\maybe{Id}},t} n : \iota
}    
\and
\inferenceRule{T-Bool}
{b \in \mathbb{B} \\
}
{ 
  \Sigma; I; \Phi \vdash_{\maybe{\var{Id}},t} b : \bool
}    
\and
\inferenceRule{T-Ref}
{
    \Sigma; I \vdash^{k}_{\maybe{\var{Id}},t} \var{ref} : \beta
}
{ 
  \Sigma; I; \Phi \vdash_{\maybe{\var{Id}},t} \var{ref_{\annot{\beta}}} : \beta
}
\and
\inferenceRule{T-Addr}
{
   \Sigma; I \vdash^{k}_{\maybe{\var{Id}},\kw{U}} \var{ref_{\annot{A}}} : A \\
}
{ 
  \Sigma; I; \Phi \vdash_{\maybe{\var{Id}},\kw{U}} \addr{\var{ref_{\annot{A}}}} : \kw{address} \\
}    
\and
\inferenceRule{T-Range}
{
    \Sigma; I; \Phi \vdash_{\maybe{\var{Id}},t} e : \iota_2
}
{ 
  \Sigma; I; \Phi \vdash_{\maybe{\var{Id}},t} \inrange{\iota_1}{e} : \bool
}    
\and
\inferenceRule{T-BopI}
{
    \Sigma; I; \Phi \vdash_{\maybe{\var{Id}},t} e_1 : \iota_1 \\
    \Sigma; I; \Phi \vdash_{\maybe{\var{Id}},t} e_2 : \iota_2 \\
    t = \kw{T} \implies \iota = \kw{int} \\ %% in untimed expressions we don't care about bitwidth
    \Sigma; I; \Phi \vDash_{\maybe{Id}} \inrange{\iota}{e_1~\opi~e_2} \\ 
}
{
  \Sigma; I; \Phi \vdash_{\maybe{\var{Id}},t} e_1~\opi~e_2 : \iota
}
\and
\inferenceRule{T-NumConv}
{
    \Sigma; I; \Phi \vdash_{\maybe{\var{Id}},t} e : \iota_1 \\
    \iota_2 \not = \kw{int} \Rightarrow \Sigma; I; \Phi \vDash_{\maybe{Id}} \inrange{\iota_2}{e}
}
{
  \Sigma; I; \Phi \vdash_{\maybe{\var{Id}},t} e : \iota_2
}
\and
\inferenceRule{T-BopB}
{
    \Sigma; I; \Phi \vdash_{\maybe{\var{Id}},t} e_1 : \bool \\
    \Sigma; I; \Phi \vdash_{\maybe{\var{Id}},t} e_2 : \bool
}
{
  \Sigma; I; \Phi \vdash_{\maybe{\var{Id}},t} e_1~\opb~e_2 : \bool
}
\and
\inferenceRule{T-Neg}
{
    \Sigma; I; \Phi \vdash_{\maybe{\var{Id}},t} e : \bool
}
{
  \Sigma; I; \Phi \vdash_{\maybe{\var{Id}},t} \neg e : \bool
}
\and
\inferenceRule{T-Cmp}
{
    \Sigma; I; \Phi \vdash_{\maybe{\var{Id}},t} e_1 : \iota \\
    \Sigma; I; \Phi \vdash_{\maybe{\var{Id}},t} e_2 : \iota 
}
{
  \Sigma; I; \Phi \vdash_{\maybe{\var{Id}},t} e_1~\cmp~e_2 : \bool
}
\and
\inferenceRule{T-ITE}
{
    \Sigma; I; \Phi \vdash_{\maybe{\var{Id}},t} e_1 : \bool \\
    \Sigma; I; \Phi \vdash_{\maybe{\var{Id}},t} e_2 : \beta\\
    \Sigma; I; \Phi \vdash_{\maybe{\var{Id}},t} e_3 : \beta\\
}
{
  \Sigma; I; \Phi \vdash_{\maybe{\var{Id}},t} (\ite{e_1}{e_2}{e_3}) : \beta
}
\and
\inferenceRule{T-Eq}
{
    \Sigma; I; \Phi \vdash_{\maybe{\var{Id}},t} e_1 : \beta\\
    \Sigma; I; \Phi \vdash_{\maybe{\var{Id}},t} e_2 : \beta\\
}
{
  \Sigma; I; \Phi \vdash_{\maybe{\var{Id}},t} e_1 = e_2 : \bool
}
\end{mathpar}

\end{flushleft}

\section{Well-typed State Environment}
\begin{definition}[Well-typed $\Sigma$]
    \label{def:well-typed-sigma}
    We say that $\Sigma$ is \emph{well-typed}
    \begin{enumerate}[start=1,label={(\bfseries S\arabic*)}]
    \item\label{i:wtsigma:cnstr}  for every $A \in \dom{\Sigmacode}$ exists $\Sigma' \subset \Sigma$, 
           such that $\Sigma'$ is well-typed and $\Sigma' \vdash_{A} \Sigmacode(A) : \Sigmastore(A)$.
    \item\label{i:wtsigma:trans} for every $A \in \dom{\Sigmatranss}$ exists $\Sigma' \subset \Sigma$, 
    such that $\Sigma'$ is well-typed and $\forall \var{trans} \in \Sigmatranss(A). \Sigma' \vdash_{A} \var{trans}$.
    \end{enumerate}
\end{definition}

\begin{lemma}[Extending $\Sigma$ Preserves Well-Typedness]
    \label{lem:extending-sigma-well-typed}
    Let $\Sigma$ be well-typed and $\Sigma \vdash \var{contract} : \Sigma'$. Then $\Sigma'$ is well-typed.
\end{lemma}

\begin{proof}
    Let $\var{contract} = \contract{\var{Id}}{\var{cnstr}}{\many{\var{trans}}}{\kw{invariants}~\many{\var{inv}}}$. By inversion on $\Sigma \vdash \var{contract} : \Sigma'$ we obtain
    \begin{enumerate}[start=1,label={(\bfseries I\arabic*)}]
        \item\label{i:wtsigma:Icnstr} $\Sigma \vdash_{\var{Id}} \var{cnstr} : C$
        \item $\Sigma'' = \Sigma ~\meta{with}~ \{ \sstorage = (\Sigmastore, \var{Id} : C), 
             \scnstrs = (\Sigmacode, \var{Id} : \var{cnstr}) \}$
        \item\label{i:wtsigma:Itranss} $\many{(\Sigma'' \vdash_{\var{Id}}  \var{trans})}$
        \item $\many{(\Sigma''; \var{cnstr}_{\kw{Iface}} \vdash_{\var{Id},\kw{U}} \var{inv} : \bool)}$
        \item $\Sigma' = \Sigma'' ~\meta{with}~ \{\stranss = (\Sigmatranss, \var{Id} : \many{\var{trans}})\}$
    \end{enumerate}
    To prove~\ref{i:wtsigma:cnstr} for $\Sigma'$, let $A\in \dom{\Sigmacodep}$. We distinguish the following two cases.
    \begin{proofcases}
    \case{$A \in \Sigmacode$}
        Since $\Sigma$ is well-typed we get exists $\Sigma^{(3)} \subset \Sigma \subset \Sigma'$, such that $\Sigma^{(3)}$ is well-typed and $\Sigma^{(3)} \vdash_{A} \Sigmacode(A) : \Sigmastore(A)$, and since $\Sigmacode(A) = \Sigmacodep(A)$ and $\Sigmastorep(A) = \Sigmastore(A)$, we also have $\Sigma^{(3)} \vdash_{A} \Sigmacodep(A) : \Sigmastorep(A)$.
    \case{$A = \var{Id}$}
    Then $\Sigma \subset \Sigma'$ and $\Sigma$ is well-typed, and by~\ref{i:wtsigma:Icnstr} we have that $\Sigma \vdash_{A} \var{cnstr} : C$.
    \end{proofcases}

    To prove~\ref{i:wtsigma:trans} for $\Sigma'$, let $A\in \dom{\Sigmatranssp}$. We distinguish the following two cases.
    \begin{proofcases}
    \case{$A \in \Sigmatranss$}
        Since $\Sigma$ is well-typed we get exists $\Sigma^{(3)} \subset \Sigma \subset \Sigma'$, such that $\Sigma^{(3)}$ is well-typed and $\forall \var{trans} \in \Sigmatranss(A). \Sigma^{(3)} \vdash_{A} \var{trans}$. Because $\Sigmatranss(A) = \Sigmatranssp(A)$, we also have $\forall \var{trans} \in \Sigmatranssp(A). \Sigma^{(3)} \vdash_{A} \var{trans}$.
    \case{$A = \var{Id}$}
    Then $\Sigma'' \subset \Sigma'$ and by~\ref{i:wtsigma:Itranss} we have that  $\many{(\Sigma'' \vdash_{A} var{trans})}$. We now only need to show that $\Sigma''$ is also well-typed. 

    Since $\Sigma''_{\scnstrs} = \Sigmacodep$ we prove~\ref{i:wtsigma:cnstr} for $\Sigma''$ as above for $\Sigma'$.
    Since $\Sigma''_{\stranss} = \Sigmatranss$ the property~\ref{i:wtsigma:trans} for $\Sigma''$ follows from the property~\ref{i:wtsigma:trans} for $\Sigma$, like above (in the case $A \in \Sigmatranss$). We thus showed $\Sigma''$ is well-typed, which concludes the proof.
    \end{proofcases}
\end{proof}

\begin{corollary}
    If $\vdash \var{spec} : \Sigma$, then $\Sigma$ is well-typed.
\end{corollary}

\begin{proof}
    Inverting $\vdash \var{spec} : \Sigma$ gives us $\Sigma_0 = \varnothing$, which is vacuously well-typed, and for $i\in[1,n]$ we have $\Sigma_{i-1} \vdash \var{contract_i} : \Sigma_i$, on which we invoke~\Cref{lem:extending-sigma-well-typed} iteratively to obtain that $\Sigma_n = \Sigma'$ is well-typed. 
\end{proof}

% \todo{Is this the correct place?} Zoe: Seems fine to me.
From here onwards we assume we are always working with a well-typed $\Sigma$.

\section{Environment Well-Foundedness}
We say that a storage typing environment $\Sigma$ is well-founded when there are
no circular references to contract types. All state typings that result from
well-typed act programs are well-founded. We formalize this notion here.

For a given state typing $\Sigma$, we define a relation $\prec_\Sigma$ as
follows. Intuitively, $B \prec_\Sigma A$ holds when the contract $B$ is directly
accessible from $A$.

    \begin{flushleft}
    \fbox{$\prec_\Sigma ~\in \meta{String} \times \meta{String}$}
    \end{flushleft}
    \[\begin{array}{r l }
    \prec_\Sigma ~= ~ \{~ (B, A) ~|~ \exists \, x. ~\, \Sigmastore(A)(x) = B ~~\vee~~\Sigmastore(A)(x) = \kw{address}_B~\}
    \end{array}\]
    
% reference https://www.cse.chalmers.se/~nad/listings/equality/Accessibility.html

We say that $\prec_\Sigma$ is well-founded if every element is accessible. 

\[\meta{WF}(\prec_\Sigma) = \forall \, x. ~\, \meta{Acc}_{\prec_\Sigma}~x\]

We
make use of the standard accessibility predicate~\cite{wellfounded}:
\[\meta{Acc}_{\prec_\Sigma}~x \Leftrightarrow (\forall y.\; y \prec_{\Sigma} x\; \rightarrow \; \meta{Acc}_{\prec_\Sigma}~y )\]

\begin{theorem}[Well-typed state is well-founded]\label{thm:wfstate}
If  $\Sigma \vdash \var{contract} : \Sigma'$ and $\meta{WF}(\prec_\Sigma)$ then 
$\meta{WF}(\prec_{\Sigma'})$
\end{theorem}
\begin{proof}
    We first observe that $\Sigmastorep = \Sigmastore \cup (Id: C)$ and therefore for all $A$, $B$, 
    such that $A \neq Id$ holds $B \prec_{\Sigma} A \Leftrightarrow B \prec_{\Sigma'} A $.
    Further, from the fact that $\Sigma$ is well-typed and \Cref{lem:extending-sigma-well-typed}, 
    follows that $Id \not\prec_{\Sigma'} A$ for any $A$.

    We proceed by showing $\meta{Acc}_{\prec_{\Sigma'}} A$ for all $A \neq Id$. We show this by
    induction on $\meta{len}(\Sigma',A)$, the length of the maximum chain of contract references in the storage state 
    $\Sigmastorep$ starting from $A$, as defined after \Cref{cor:wt-state-wf}. There cannot be infinite chains due to 
    $\meta{WF}(\prec_{\Sigma})$ and the fact that $\meta{len}(\Sigma,A) = \meta{len}(\Sigma',A)$, for $A \neq Id$. 

    Starting with the base case $\meta{len}(\Sigma',A)=0$, we observe that $\meta{Acc}_{\prec_{\Sigma'}} A$ trivially
    holds, since there are no contracts $B$ such that $B \prec_{\Sigma'} A$. Next, we study the inductive case and 
    assume $\meta{len}(\Sigma',A) = n$. We pick an arbitrary $B$ such that $B \prec_{\Sigma'} A$. 
    By construction of $\meta{len}$, we obtain $\meta{len}(\Sigma',B) < n$. Further, as observed at the beginning
     $Id \not\prec_{\Sigma'} A$ for any $A$, hence $B \neq Id$. Therefore, we can apply the induction hypothesis to 
     $B$ and conclude that $\meta{Acc}_{\prec_{\Sigma'}} B$. Since $B$ was picked arbitrary this holds for all 
     $B \prec_{\Sigma'} A$. By the definiton of $\meta{Acc}_{\prec_{\Sigma'}}$ it follows that  $\meta{Acc}_{\prec_{\Sigma'}} A$.
     This concludes the induction proof.
     
    What is left to show is that also  $\meta{Acc}_{\prec_{\Sigma'}} Id$.
     Following the defnition of $\meta{Acc}_{\prec_{\Sigma'}}$,
    we pick an arbitrary $B \prec_{\Sigma'} Id$. Since $B \neq Id$, we can use the above induction proof to obtain 
    $\meta{Acc}_{\prec_{\Sigma'}} B$, which implies  $\meta{Acc}_{\prec_{\Sigma'}} Id$.
\end{proof}

\begin{corollary}[Well-typed state is well-founded]\label{cor:wt-state-wf}
If $\cdot \vdash \var{spec} : \Sigma$ then $\meta{WF}(\prec_{\Sigma})$
\end{corollary}
\begin{proof}
 Follows from Theorem~\ref{thm:wfstate} and the fact that $\prec_\cdot$ is trivially well-founded.
\end{proof}

Let $\Sigma$ be well-founded. We define $\meta{len}(\Sigma, A)$ to be the length
of the maximum chain of contract references in the storage state $\Sigmastore$ starting from $A$. The
definition is by induction on the relation $\meta{Acc}_{\prec_\Sigma}~A$. 

    \[\begin{array}{r l }
    \meta{len}(\Sigma, A) ~= \max_{ B \in \{ B ~|~ B \prec_\Sigma A \}} (1 + \meta{len}(\Sigma, B)) 
    \end{array}\]

We extend this definition to slot types as follows.
    \[\begin{array}{r l r}
    \meta{len}(\Sigma, A) ~= & \meta{len}(\Sigma, A) \\
    \meta{len}(\Sigma, \kw{address}_A) ~= & \meta{len}(\Sigma, A) \\
    \meta{len}(\Sigma, \sigma) ~= & 0 & \text{otherwise}\\
    \end{array}\]

\section{Pointer Semantics Lemmas}
\begin{lemma}[Uniqueness of Contract Value Typing]\label{lem:uniqueness-of-value-typing}
  Let $\Sigma \vdash v :_s A$ and $\Sigma \vdash v :_s B$.
  Then $A = B$.
\end{lemma}
\begin{proof}
    We assume $\Sigma \vdash v :_s A$ and $\Sigma \vdash v :_s B$. In both derivations the last 
    rule that had to have been applied is \rref{V-Contract} and the second to last \rref{V-AddrIsContract},
    since they are the only ones that yield the required conclusion. Hence, inverting both rules in both 
    derivations yields that $s(v).\kw{type} = A$ but also $s(v).\kw{type} = B$. 
    Since $\kw{type}$ assigns a storage location $s(v)$ at most one type (it is a function), 
    follows $A = B$.
\end{proof}

% \begin{lemma}[Type Conversion]\label{lem:type-conv}
%   Let $\Sigma \vdash v :_s \sigma_1$ and $\sigma_1 \equiv \sigma_2$ then
%   $\Sigma \vdash v :_s \sigma_2$.
% \end{lemma}
% \begin{proof}
%     Consider $\sigma_1, \sigma_2$ such that $\sigma_1 \equiv \sigma_2$. Either $\sigma_1 = \sigma_2$, 
%     in which case the lemma trivially holds, or one of $\sigma_1, \sigma_2$ is $A$ and the other 
%     is $\kw{address}_A$.

%     First, assume $\sigma_1 = A$. The only way to reach $\Sigma \vdash v :_s \sigma_1$, was to applying
%     \rref{V-Contract} last. Inverting it yields $\Sigma \vdash v :_s \kw{address}_A$.
%     Second, we assume  $\sigma_1 = \kw{address}_A$. In this case we can just apply \rref{V-Contract}
%     to conclude $\Sigma \vdash v :_s A$.
% \end{proof}
 \begin{lemma}[Determinism of Base- and Mapping Expression Evaluation]\label{lem:exp-determinism}
Let $s$ be a storage and $\rho$ an environment. For any base expression or mapping expression $e$,
if
$$\psem{s}{\rho}{e} v \quad\text{and}\quad \psem{s}{\rho}{e} v'$$
then $v = v'$.
\end{lemma}

\begin{lemma}[Determinism of Reference Evaluation]\label{lem:ref-determinism}
Let $s$ be a storage and $\rho$ an environment. For any reference $\var{ref}$, if
$$\psem{s}{\rho}{\var{ref}} (v,t) \quad\text{and}\quad \psem{s}{\rho}{\var{ref}} (v',t')$$
then $(v,t) = (v',t')$.
\end{lemma}

\begin{proof}[Proof of \Cref{lem:exp-determinism} and \Cref{lem:ref-determinism}]
    By mutual structural induction on the form of $e$, respectively $\var{ref}$.
     For each possible shape according to the syntax, we identify all possible last derivation steps, use the induction hypothesis on premises (where applicable) and conclude
    that conclusions of the final rule coincide. We give details for the more interesting cases, 
    the others are similar. We start with the cases for references.
    \begin{proofcases}
        \case{$\var{ref}=x$}
        \begin{proofcases}
            \case{$x \in \dom{\rho}$}         In this case the only rules that could have been applied to $x$ are \rref{E-Calldata}
            and \rref{E-CalldataTimed} depending on whether the state $s$ is just one untimed state $s = s^{\kw{U}}$
            or a pair of states $s = (s_{\kw{pre}},s_{\kw{post}})^{\kw{T}}$, exactly one of the to rules applies.
            and yields $(\rho(x),\kw{U})$, respectively $(\rho(x), \kw{pre})$. Hence, the evaluation in this case is deterministic: 
            $v = \rho(x) = v'$ and all timings either $\kw{U}$ or $\kw{pre}$.

            \case{$x \notin \dom{\rho}$} The only derivation rule that could have been 
            applied is \rref{E-Storage}, which yields $(s(\ell)(x),\kw{U})$. Thus, the evaluation of $x$ is also deteministic in this case:
            $(v,t) = (s(\ell)(x), \kw{U}) = (v',t')$.
        \end{proofcases}

        \case{$\var{ref}=\kw{pre}(x)$ or $\var{ref}=\kw{post}(x)$} Similar to the $x \notin \dom{\rho}$ subcase of 
        $\var{ref}=x$.

        \case{$\var{ref}=\var{ref}\; \kw{as} \;\var{Id}$} The only rule that applies in this case is \rref{E-Coerce}.
        Hence both $(v,t)$ and $(v',t')$ had to have been derived using this rule last.
        Inversion of the rule yields $\psem{s}{\rho}{\var{ref}} (v,t)$ and $\psem{s}{\rho}{\var{ref}} (v',t')$.
        Applying the induction hypothesis on $\var{ref}$, though, we obtain $(v,t) = (v',t')$.

        \case{$\var{ref}=\var{ref}.y$} 
        \begin{proofcases}
            \case{ $s = s^{\kw{U}}$} In this case exactly one rule applies (\rref{E-Field}) and the reasoning is the same as in the previous case
            $\var{ref}=\var{ref}\; \kw{as} \;\var{Id}$.
    
            \case{ $s = (s_{\kw{pre}},s_{\kw{post}})^{\kw{T}}$} The only rules that apply in this case 
            are \rref{E-FieldPre} and \rref{E-FieldPost}. Hence both $(v,t)$ and $(v',t')$ had to be derived
            using one of them last. Inverting them, yields $\psem{s}{\rho}{\var{ref}} (\ell',t)$ 
            and $\psem{s}{\rho}{\var{ref}} (\ell'',t')$  where depending on which of the rules where used
            $t$ and $t'$ could independently of each other be $\kw{pre}$ or $\kw{post}$. However,
            applying the induction hypothesis to $\var{ref}$ yields that $\ell' = \ell''$ and $t=t'$.
            Hence, the same rule with identical premises had to have been applied to reach both $(v,t)$ and $(v',t')$.
            Thus, if $t = \kw{pre}$, then $(v,t) = (s_{\kw{pre}}(\ell')(x), \kw{pre}) = (v',t')$ and analog 
            for $t = \kw{post}$.
        \end{proofcases}

        \case{$\var{ref} = \var{ref}[x]$} Only \rref{E-RefMapping} applies, hence we can use the same arguments as above. 
        Note that in this case we rely on the mutual 
        induction hypothesis for references and base expressions when infering that expression $e$ evaluates deterministically.

        \case{$\var{ref} = \var{env}$} Only \rref{E-Environment} applies, hence the standard argumentation suffices.
    \end{proofcases} We continue with the cases for base expressions.
    \begin{proofcases}
        \case{$e = i$ or $e = \kw{true}$ or $e = \kw{false}$} In these base cases there is only one evaluation 
        step, which has to be either from \rref{E-Int} (for $e = i$) or \rref{E-Bool} (for $e = \kw{true}$ and 
        $e = \kw{false}$). Hence, the evaluation in these cases is determintistic: $v = e = v'$.

        \case{$e = \kw{inrange}(\iota,e)$} Both $v$ and $v'$ had to be deduced using \rref{E-RangeTrue} or \rref{E-RangeFalse}. 
        Inverting them yields $\psem{s}{\rho}{e} v_e$ and $\psem{s}{\rho}{e} v'_e$.   By the induction hypothesis, we know that $v_e = v'_e$.
        Therefore, $v$ and $v'$ had to be derived using the same rule (since $v_e$ and $v'_e$ have to satisfy the same premises as they
        are identical) and hence $v = v' \in \{\meta{True}, \meta{False}\}$. 

        \case{$e = e_1 \;\kw{div} \;e_2$ or $e = e_1 \;\kw{mod}\; e_2$} Let us assume $e = e_1 \;\kw{div} \;e_2$ first. the values
        $v$ and $v'$ had to have been derived using either \rref{E-Div} or \rref{E-DivZero}. Inverting these rules for both derivations 
        we get $\psem{s}{\rho}{e_2} v_2$ and $\psem{s}{\rho}{e_2} v'_2$. By the induction hypothesis, we get that $v_2=v'_2$ and hence
        the same rule was applied to reach $v$ and $v'$, but since also $e_1$ evaluates deterministically to $v_1$ by the induction hypothesis,
        we know that $v = v_1/v_2 = v'$, if $v_2 \neq 0$ and $v = 0 = v'$, otherwise.

        The same reasoning shows the determinism for $\kw{mod}$.

        \case{$e =  \ite{e_1}{e_2}{e_3}$} Since all of $e_1$, $e_2$ and $e_3$ evaluate deterministically by the induction
        hypothesis, only one of \rref{E-ITETrue}, \rref{E-ITEFalse} applies and yields the deterministic evaluation of $e_2$, respectively $e_3$.

        \case{$e = e_1 = e_2$} Similar argumentation as in the previous case.

        \case{other cases} In all other cases, which are $e = \var{ref}_{\alpha}$, \; $e = \kw{addr}(\var{ref}_{Id})$,\;
        $e = e_1 \opi e_2$, \;$e = e_1 \opb e_2$, \;$e = e_1 \cmp e_2$, \;and $e = \neg e$ exactly one rule applies and their
        deterministic evaluation can be shown using the induction argument.
    \end{proofcases} Lastly, we consider the mapping expression cases, refering directly to the base expression for 
    mapping expressions that are also base expressions.
    \begin{proofcases}
        \case{$e = [\{e_i \Rightarrow m_i\}_{i \in [1,n]}]_{\mu}$} The only rule that applies in this case is \rref{E-Mapping}. 
        Assuming $[\many{e \Rightarrow m}]_{\mu}$ evaluates to both $v$ and $v'$, we can invert the last rule 
        \rref{E-Mapping} of both derivations to obtain $\psem{s}{\rho}{e_i} v_i$ and $\psem{s}{\rho}{e_i} v'_i$. 
        But from the induction hypothesis follows that $v_i = v'_i$ for all $i$ and similar for the evaluations
        of $m_i$. Therefore, by applying the same rule to identical premises, we conclude that $v = v'$.

        \case{$e = \var{ref}_{\mu}[\many{e \Rightarrow m}]_{\mu}$} Similar argumentation as for the case 
        $e = [\{e_i \Rightarrow m_i\}_{i \in [1,n]}]_{\mu}$, but also using the induction hypothesis for references on $\var{ref}$. 

    \end{proofcases}
\end{proof}

\begin{lemma}[Determinism of Value Insertion]\label{lem:ins-determinism}
    Let $s$ be a storage and $\rho$ an environment. For any reference $\var{ref}$ and value $v$ if
    $$\ins{s;\rho;\var{ref}; v}{s'_1} \quad\text{and}\quad \ins{s;\rho;\var{ref}; v}{s'_2} $$
    then $s'_1 = s'_2$.
\end{lemma}
\begin{proof}
    Insertion can only happen in two cases: when $\var{ref} = x$ or when $\var{ref} = \var{ref}.x$.
    For the first case $\var{ref} = x$ clearly only \rref{E-InsStorage} applies and yields 
    $s[\ell \mapsto s(\ell)[x \mapsto v]] = s'_1 = s'_2$.

    For the other case only \rref{E-InsField} applies, hence both $s'_1$ and $s'_2$ were obtained by this rule. 
    Inversion of both instances tells us that $\psem{s}{\rho}{\var{ref}} (\ell'_1,\kw{U})$ and $\psem{s}{\rho}{\var{ref}} (\ell'_2,\kw{U})$.
    However, we can apply \Cref{lem:ref-determinism} to conclude that $\ell'_1 = \ell'_2$ and therefore also 
    $s'_1 = s[\ell'_1 \mapsto s(\ell'_1)[x \mapsto v]] = s'_2$.
\end{proof}

\begin{lemma}[Determinism of Slot Expression Evaluation]\label{lem:slotexp-determinism}
    Let $s$ be a storage and $\rho$ an environment. For any slot expression $\var{se}$,
    if
    $$\psem{s}{\rho}{se} (v_1,s'_1) \quad\text{and}\quad \psem{s}{\rho}{se} (v_2,s'_2)$$
    then $v_1 = v_2$ and $s'_1 = s'_2$.
\end{lemma}

\begin{lemma}[Determinism of Constructor Case Evaluation]\label{lem:ccase-determinism}
    Let $s$ be a storage and $\rho$ an environment. For any constructor cases expression $\var{ccases}$,
    if
    $$\psemnoell{s}{\rho}{\var{ccases}}_{\var{Id}} (\ell_1,s'_1) \quad\text{and}\quad \psemnoell{s}{\rho}{\var{ccases}}_{Id} (\ell_2,s'_2)$$
    then $\ell_1 = \ell_2$ and $s'_1 = s'_2$.
\end{lemma}

\begin{lemma}[Determinism of Creates Evaluation]\label{lem:creates-determinism}
    Let $s$ be a storage and $\rho$ an environment. For any creates expression $\many{\var{create}}$,
    if
    $$\psemnoell{s}{\rho}{\many{\var{create}}}_{Id} (\ell_1,s'_1) \quad\text{and}\quad \psemnoell{s}{\rho}{\many{\var{create}}}_{Id} (\ell_2,s'_2)$$
    then $\ell_1 = \ell_2$ and $s'_1 = s'_2$.
\end{lemma}

\begin{proof}[Proof of \Cref{lem:slotexp-determinism}, \Cref{lem:ccase-determinism} and \Cref{lem:creates-determinism}]
We proceed by mutual structural induction on the syntactic form of the slot expression  $\var{se}$ , constructor cases expression  $\var{ccases}$
 or creates expression $\many{\var{create}}$ and assume it evaluates to 
both $(v_1,s'_1)$ and $(v_2,s'_2)$, respectively $(\ell_1,s'_1)$ and $(\ell_2,s'_2)$. We start listing the cases for slot expressions.
\begin{proofcases}
    \case{$se = e$} The only slot expression evaluation rule that could have been applied here
    is \rref{E-MapExp}, yielding by inversion that $\psem{s}{\rho}{e} v_1 \quad\text{and}\quad \psem{s}{\rho}{e} v_2$.
    From \Cref{lem:exp-determinism}, we know that in this case $v_1 = v_2$. Further from \rref{E-MapExp}, we know that
    $s'_1 = s = s'_2$.

    \case{$se = \kw{new}~\var{Id}(\{\var{se}_i\}_{i \in [1,n]})$} The only rule that applies in this case is \rref{E-Create}.
    Hence, both $(v_1,s'_1)$ and $(v_2,s'_2)$ had to have been derived using this rule last. By inversion of both instances, we obtain
    $\psem{s_{0}}{\rho}{se_1} (v_{1,1},s_{1,1})$ and  $\psem{s_{0}}{\rho}{se_1} (v_{2,1},s_{2,1})$, where $s = s_0$. Applying the induction hypothesis to $se_1$,
    we obtain $v_{1,1} = v_{2,1}$ and $s_{1,1}= s_{2,1}$. Repeating this argument $n$ times yields deterministic $(v_i, s_i)$ and thus also 
    a deterministic evaluation of $s_{n-1};\rho;se_n$ to $(v_n,s_n)$. Finally applying the mutual induction hypothesis to the last premise 
    of the rule in both instances yields that $s_n;\rho';\var{ctor}_{\kw{cases}}$ deterministically evaluates to $(\ell',s')$. Therefore, 
    $\ell'_1 = \ell' = \ell'_2$ and $s'_1 = s' = s'_2$.

    \case{$se = \kw{new}~\var{Id}\{\kw{value}: \var{se'}\}(\{\var{se}_i\}_{i \in [1,n]}) $} Similar to the previous case.

    \case{$se = \var{ref}_{\var{Id}}$} In this case only \rref{E-SlotRef} applies and therefore had to be the last rule that was applied to reach both $(v_1,s'_1)$
    and $(v_2,s'_2)$. By the defintion of the rule follows $s'_1 = s = s'_2$. Inverting both instances implies $\psem{s}{\rho}{\var{ref}_{Id}} v_1$ as well as $\psem{s}{\rho}{\var{ref}_{Id}} v_2$.
    Using now \Cref{lem:ref-determinism}, we conclude that also $v_1 = v_2$.  

    \case{$se = \kw{addr}(se)$} For this case the only rule that could have been applied last is \rref{E-SlotAddr}.
    Hence, it had to have been applied last for both evaluations $(v_1,s'_1)$ and $(v_2,s'_2)$.
    Inverting the rule and using the induction hypothesis, however we conclude $v_1 = v_2$ and $s'_1 = s'_2$.
\end{proofcases}
We next consider constructor cases expressions and assume 
$$\psemnoell{s}{\rho}{\{\kw{case} \;e_i\;:\; \kw{creates} \; \many{\var{create}_i}\}_{i \in [1,n]}}_{Id}  (\ell_1,s'_1)$$
and also evaluates to $(\ell_2,s'_2)$. Both derivations had to end with \rref{E-CtorCases}. We can therefore invert the rule in 
both instances. From \Cref{lem:exp-determinism}, we know that the $e_i$ will evaluate the same way in both instances, hence also 
considering the same $\many{\var{create}_j}$ (for the $e_j$ that evaluated to \meta{True}). Applying the mutual induction
hypothesis to the evalution of $\many{\var{create}_j}$ yields a deterministic result $(\ell,s')$ of its evaluation.
Thus, using the definition of \rref{E-CtorCases}, we conclude $\ell_1 = \ell =\ell_2$ and $s'_1 = s' =s'_2$.

Finally, we consider creates expressions and assume $$\psemnoell{s_0}{\rho}{\{\sigma_i \;x_i\,:= se_i\}_{i \in [1,n]}}_{Id} (\ell_1, s'_1)$$
and also evaluates to $(\ell_2, s'_2)$. We invert the only applicable rule for both derivations to obtain
 $\psem{s_0}{\rho}{se_1} (v_{1,1}, s_{1,1})$ and  $\psem{s_0}{\rho}{se_1} (v_{2,1}, s_{2,1})$ . With the mutual induction
 hypothesis we get that $v_{1,1}=v_{2,1}$ and $s_{1,1} = s_{2,1}$. Applying the same argument $n$ times yields the same evaluation $(v_n, s_n)$
 of the $se_i$ in both derivations. Also the computation of $\ell$ is deterministic and therefore identical in both derivations.
 The same argument tells us that the $s_{n+1}$ have to be identical. Hence, by the definition of \rref{E-Creates}, we conclude $\ell_1 = \ell = \ell_2$ and 
 $s'_1 = s_{n+1} = s'_2$. 
\end{proof}

\begin{lemma}[Determinism of Updates Evaluation]\label{lem:updates-determinism}
    Let $s$ be a storage and $\rho$ an environment. For any updates expression $\many{\var{upd}}$,
    if
    $$\psem{s}{\rho}{\many{\var{upd}}} s'_1 \quad\text{and}\quad \psem{s}{\rho}{\many{\var{upd}}} s'_2$$
    then $s'_1 = s'_2$.
\end{lemma}
\begin{proof}
    Let $\many{\var{upd}} = \{ \var{ref}_i \; := \; se_i\}_{i \in [1,n]}$ and assume it evaluates to both $s'_1$
    and $s'_2$. By proceeding in the same way as in the proof of \Cref{lem:creates-determinism}, only skipping the 
    picking of a fresh location $\ell$ and the extra step of defining $s_{n+1}$, 
    we can infer that also $s'_1 = s'_2$.
\end{proof}

\begin{lemma}[Determinism of Transition Case Evaluation]\label{lem:bcase-determinism}
    Let $s$ be a storage and $\rho$ an environment. For any transition cases expression $\var{tcases}$,
    if
    $$\psem{s}{\rho}{\var{tcases}} (v_1,s'_1) \quad\text{and}\quad \psem{s}{\rho}{\var{tcases}} (v_2,s'_2)$$
    then $v_1 = v_2$ and $s'_1 = s'_2$.
\end{lemma}
\begin{proof}
    Let $\var{tcases} =  \{\kw{case} \;e_i\,:\, \kw{updates} \; \many{\var{upd}_i} \; \kw{returns}\; \var{ret}_i \}_{i \in [1,n]}$
    and assume it evaluates to both $ (v_1,s'_1)$ and $ (v_2,s'_2)$. Then both derivations had to end with 
    \rref{E-TransCases}. We invert both instances. With the same argumentation as in the proof of \Cref{lem:ccase-determinism}
    we reach the conclusion that $s'_1 = s'_2$. Finally, applying \Cref{lem:exp-determinism} to $\var{ret}_j$ 
    ($j$ such that $e_j$ evaluates to $\meta{True}$), we obtain that also $v_1 = v_2$.  
\end{proof} 

\begin{lemma}[Determinism of Constructor Evaluation]\label{lem:cnstr-determinism}
    Let $s$ be a storage and $\rho$ an environment. For any constructor $\var{cnstr}$,
    if
    $$\psemnoell{s}{\rho}{\var{cnstr}}_{Id} (\ell_1,s'_1) \quad\text{and}\quad \psemnoell{s}{\rho}{\var{cnstr}}_{Id} (\ell_2,s'_2)$$
    then $\ell_1 = \ell_2$ and $s'_1 = s'_2$.
\end{lemma}
\begin{proof}
    Let $\var{cnstr} = \kw{constructor} \,  (I)\, ? \,( \kw{payable} ) \,  \kw{iff} \, \{e_i\}_{i \in [1,n]} \, \var{ccases} \, \kw{ensures} \, \many{\var{post}}$
    and assume it evaluates to both $(\ell_1,s'_1)$ and $(\ell_2,s'_2)$. The rule \rref{E-Ctor} had to be the last in 
    both derivations. By inverting it and applying \Cref{lem:exp-determinism} to the $e_i$ and \Cref{lem:ccase-determinism}
    to $\var{ccases}$, we obtain that  $\ell_1 = \ell_2$ and $s'_1 = s'_2$, which completes the proof.
\end{proof}

\begin{lemma}[Determinism of Transition Evaluation]\label{lem:trans-determinism}
    Let $s$ be a storage and $\rho$ an environment. For any transition expression $\var{trans}$,
    if
    $$\psem{s}{\rho}{\var{trans}} (v_1,s'_1) \quad\text{and}\quad \psem{s}{\rho}{\var{trans}} (v_2,s'_2)$$
    then $v_1 = v_2$ and $s'_1 = s'_2$.
\end{lemma}
\begin{proof}
    Similar to the proof of \Cref{lem:cnstr-determinism}, only using \Cref{lem:bcase-determinism} instead of 
    \Cref{lem:ccase-determinism}.
\end{proof}

\section{Typing Lemmas}
\begin{lemma}[Weakening of Storage (Pointer Semantics)]\label{lem:sem-storage-weak}
Let $s$ and $s'$ be states for untimed judgements and pairs of states $s=(s_{\kw{pre}},s_{\kw{post}})$ and $s'=(s'_{\kw{pre}},s'_{\kw{post}})$ 
for timed judgements.  We assume $s \subseteq s'$. For the timed case that means $s_{\kw{pre}} \subseteq s'_{\kw{pre}}$ and 
$s_{\kw{post}} \subseteq s'_{\kw{post}}$. Further, let $\ell$ be a location and $\rho$ be an environment. 
% \begin{enumerate}
    % \item[(H1)] 
    % \item[(H2)] \textcolor{orange}{If $\ell \in \dom{s}$ then
    %  $\dom{s(\ell)} \cap \dom{\rho} = \emptyset$. For timed  that means $\dom{s_{\kw{pre}}(\ell)} \cap \dom{\rho} = \emptyset$ 
    % and  $\dom{s_{\kw{post}}(\ell)} \cap \dom{\rho} = \emptyset$. }
    % \item[(H3)] \textcolor{orange}{If $\ell \in \dom{s'}$ then
    %  $\dom{s'(\ell)}\cap \dom{\rho} = \emptyset$. For timed that means $\dom{s'_{\kw{pre}}(\ell)}\cap \dom{\rho} = \emptyset$ 
    % and  $\dom{s'_{\kw{post}}(\ell)}\cap \dom{\rho} = \emptyset$. }
% \end{enumerate}
  Then, the following hold.
    \begin{enumerate}
        % expressions
        \item If a base expression $e$ evaluates to $\psem{s}{\rho}{e} \var{v}$ then also $\psem{s'}{\rho}{e} \var{v}$.
        % references
        \item If a variable reference $\var{ref}$ evaluates to $\psem{s}{\rho}{\var{ref}} (\var{v}, \var{t_p})$
        then also $\psem{s'}{\rho}{\var{ref}} (\var{v}, \var{t_p})$.
        % mapping
        \item If a mapping expression $m$ evaluates to $\psem{s}{\rho}{m} \var{v}$ then also $\psem{s'}{\rho}{m} \var{v}$.
    \end{enumerate}
\end{lemma}
\begin{proof} We first prove statements (1) and (2) in a mutual structural induction on the semantic evaluation tree. 
    Therefore, we assume $\psem{s}{\rho}{e} \var{v}$, respectively  $\psem{s}{\rho}{\var{ref}} (\var{v}, \var{t_p})$ and $s \subseteq s'$.
    \begin{proofcases}
        \case{\rref{E-Int}} By inversion, we get that $e \in \mathbb{Z}$. Applying now \rref{E-Int} yields  $\psem{s'}{\rho}{e} \var{v}$.
        \case{\rref{E-Bool}} Inverting the rule, we obtain $e \in \mathbb{B}$. With \rref{E-Bool} follows  $\psem{s'}{\rho}{e} \var{v}$.
        \case{\rref{E-Ref}} Inversion yields $\psem{s}{\rho}{\var{ref}} (\var{v}, \var{t_p})$. By induction hypothesis, we get 
        $\psem{s'}{\rho}{\var{ref}} (\var{v}, \var{t_p})$ and applying \rref{E-Ref} we conclude $\psem{s'}{\rho}{\var{ref}} \var{v}$.
        \case{\rref{E-Addr}} We invert the rule to obtain $\psem{s}{\rho}{\var{ref}} (\var{v}, \var{t_p})$. With the induction hypothesis
        we get  $\psem{s'}{\rho}{\var{ref}} (\var{v}, \var{t_p})$ and applying \rref{E-Addr} we conclude $\psem{s'}{\rho}{\addr{\var{ref}}} \var{v}$.
        \case{\rref{E-RangeTrue} and \rref{E-RangeFalse}} By inversion, we obtain $\psem{s}{\rho}{e} \var{v}$, $v \in \mathbb{Z}$ and 
        either $v \in [\min(\iota),\max(\iota)] \vee \iota = \kw{int}$ or  $v \notin [\min(\iota),\max(\iota)]$. From the induction hypothesis
        follows $\psem{s'}{\rho}{e} \var{v}$, hence applying the rule \rref{E-RangeTrue} respectively \rref{E-RangeFalse},
        we get  $\psem{s'}{\rho}{\inrange{\iota}{e}} \meta{True}$, or $\psem{s'}{\rho}{\inrange{\iota}{e}} \meta{False}$ respectively.
        \case{\rref{E-Div} and \rref{E-DivZero}} We invert the rule to get  $\psem{s}{\rho}{e_1} \var{v_1}$ and 
        $\psem{s}{\rho}{e_2} \var{v_2}$. We further infer $v_1,v_2 \in \mathbb{Z}$ and $v_2 = 0$, respectively $v_2 \neq 0$. 
        From the induction hypothesis we learn $\psem{s'}{\rho}{e_1} \var{v_1}$ and 
        $\psem{s'}{\rho}{e_2} \var{v_2}$. Finally, applying \rref{E-Div} respectively \rref{E-DivZero}, we conclude
        $\psem{s'}{\rho}{e_1~\kw{div}~e_2} \var{v_1}/\var{v_2}$ respectively $\psem{s'}{\rho}{e_1~\kw{div}~e_2} 0$.
        \case{\rref{E-Mod} and \rref{E-ModZero}} Analog to \rref{E-Div} and \rref{E-DivZero}.
        \case{\rref{E-BopB}, \rref{E-BopI} and \rref{E-Cmp}} Let $* \in \{\opi,\opb,\cmp\}$. By inversion of the rule,
        we obtain $\psem{s}{\rho}{e_1} \var{v_1}$ and 
        $\psem{s}{\rho}{e_2} \var{v_2}$. We further infer $v_1,v_2 \in \mathbb{Z}$ respectively $v_1,v_2 \in \mathbb{B}$. We apply
        the induction hypothesis to get $\psem{s'}{\rho}{e_1} \var{v_1}$ and  $\psem{s'}{\rho}{e_2} \var{v_2}$. 
        Lastly, we use \rref{E-BopI},\rref{E-BopB}, respectively \rref{E-Cmp}, to conclude  $\psem{s'}{\rho}{e_1 * e_2} \var{v_1}*\var{v_2}$.
        \case{\rref{E-Neg}} We invert the rule to get $\psem{s}{\rho}{e} \var{v}$, then apply the induction hypothesis to reach $\psem{s'}{\rho}{e} \var{v}$
        and finally 
        apply \rref{E-Neg} to conclude  $\psem{s'}{\rho}{\neg e} \neg \var{v}$.
        \case{\rref{E-ITETrue}} By inversion, we obtain  $\psem{s}{\rho}{e_1} \meta{True}$ and $\psem{s}{\rho}{e_2} \var{v_2}$. 
        With the  induction hypothesis we get $\psem{s'}{\rho}{e_1} \meta{True}$ and $\psem{s'}{\rho}{e_2} \var{v_2}$. Applying 
        now \rref{E-ITETrue} yields $\psem{s'}{\rho}{\ite{e_1}{e_2}{e_3}} \var{v_2}$.
        \case{\rref{E-ITEFalse}} Analog to \rref{E-ITETrue}.
        \case{\rref{E-EqTrue} and \rref{E-EqFalse}} Inverting the rule yields  $\psem{s}{\rho}{e_1} \var{v_1}$ and 
        $\psem{s}{\rho}{e_2} \var{v_2}$, and further, either $v_1=v_2$ or $v_1 \neq v_2$. The induction hypothesis yields  $\psem{s'}{\rho}{e_1} \var{v_1}$ and 
        $\psem{s'}{\rho}{e_2} \var{v_2}$ and finally applying \rref{E-EqTrue} respectively \rref{E-EqFalse}, we conclude 
        $\psem{s'}{\rho}{e_1 = e_2} \meta{True}$, respectively  $\psem{s'}{\rho}{e_1 = e_2} \meta{False}$.
        \case{\rref{E-Storage}} We invert the rule to obtain $x \in \dom{s(\ell)}$ and $x \notin \dom{\rho}$. Since, $s \subseteq s'$, we know $s(\ell)=s'(\ell)$.
        and hence  $x \in \dom{s'(\ell)}$. Applying \rref{E-Storage} yields $\psem{s'}{\rho}{x} (s'(\ell)(x), \kw{U})$. Using $s(\ell)=s'(\ell)$ again, we conclude 
        $\psem{s'}{\rho}{x} (s(\ell)(x), \kw{U})$.

        \case{\rref{E-StoragePre}} In this case, the shape of the evaluation is $\psem{(s_{\kw{pre}},s_{\kw{post}})}{\rho}{\pre{x}} (s_{\kw{pre}}(\ell)(x), \kw{pre})$. By
        inversion we obtain $x \in \dom{s_{\kw{pre}}(\ell)}$ and $x \notin \dom{\rho}$. From $s_{\kw{pre}} \subseteq s'_{\kw{pre}}$ follows
        $s_{\kw{pre}}(\ell) = s'_{\kw{pre}}(\ell)$ and $x \in \dom{s'_{\kw{pre}}(\ell)}$. With that and \rref{E-StoragePre}, we conclude
        $\psem{(s'_{\kw{pre}},s'_{\kw{post}})}{\rho}{\pre{x}} (s_{\kw{pre}}(\ell)(x), \kw{pre})$.
        \case{\rref{E-StoragePost}} Similar to \rref{E-StoragePre}.

        \case{\rref{E-Calldata}} In this case, the shape of the evaluation is $\psem{s}{\rho}{x} (\rho(x), \kw{U})$.
        We invert the rule to get $x \in \dom{\rho}$ and apply \rref{E-Calldata} to conclude  $\psem{s'}{\rho}{x} (\rho(x), \kw{U})$.

        \case{\rref{E-CalldataTimed}} Similar to \rref{E-Calldata}.

        \case{\rref{E-Coerce}} Having $\psem{s}{\rho}{\var{ref} \kw{as} \;A} (v, t_p)$, we invert the rule to obtain 
        $\psem{s}{\rho}{\var{ref}} (v, t_p)$. From the induction hypothesis we know $\psem{s'}{\rho}{\var{ref}} (v, t_p)$.
        Applying now \rref{E-Coerce} leads us to $\psem{s'}{\rho}{\var{ref} \kw{as} \;A} (v, t_p)$.

        \case{\rref{E-Field}} In this case our assumption is  $\psem{s}{\rho}{\var{ref}.x} (s(\ell')(x), \kw{U})$. By inversion, we get 
        $\psem{s}{\rho}{\var{ref}} (\ell', \kw{U})$, $\ell' \in \meta{Addr}$ and $x \in \dom{s(\ell')}$.  From the induction hypothesis
        follows  $\psem{s'}{\rho}{\var{ref}} (\ell', \kw{U})$ and from $s \subseteq s'$ follows $x \in \dom{s'(\ell')}$. Thus,
        we can apply \rref{E-Field} to obtain $\psem{s'}{\rho}{\var{ref}.x} (s'(\ell')(x), \kw{U})$ and since $s'(\ell')=s(\ell')$
        is implied by  $s \subseteq s'$ and $x \in \dom{s(\ell')}$, this case is shown.

        \case{\rref{E-FieldPre} and \rref{E-FieldPost}} Similar to \rref{E-Field}, since $s_{\kw{pre}} \subseteq s'_{\kw{pre}}$
        and $s_{\kw{post}} \subseteq s'_{\kw{post}}$.

        \case{\rref{E-RefMapping}} We assume  $\psem{s}{\rho}{\var{ref}[e]} (v_r(v_e), t_p)$. Inverting the rule, we obtain 
        $\psem{s}{\rho}{e} v_e$, $v_e \in \meta{X}$, and also $\psem{s}{\rho}{\var{ref}} (v_r, t_p)$,
        $v_r \in \meta{X} \to \meta{Value}$. The induction hypothesis yields $\psem{s'}{\rho}{e} v_e$ and 
        $\psem{s'}{\rho}{\var{ref}} (v_r, t_p)$. Finally, applying \rref{E-RefMapping}, we conclude  
        $\psem{s'}{\rho}{\var{ref}[e]} (v_r(v_e), t_p)$.
        
    \end{proofcases} This concludes claims (1) and (2). We proceed by proving claim (3) also by structural induction.
    \begin{proofcases}
        \case{\rref{E-Exp}} If the mapping expression $m$ is just a base expression $e$, that is $\psem{s}{\rho}{e} \var{v}$, 
        we apply result (1) of the lemma to conclude  $\psem{s'}{\rho}{e} \var{v}$.
        \case{\rref{E-Mapping}} We have  
        $\psem{s}{\rho}{[\{e_i \Rightarrow m_i\}_{i \in [1,n]}]_{\annot{\mapping{\beta}{\mu}}}} \var{f}$. Inverting the rule yields
        for all $i \in [1,n]$ that $\psem{s}{\rho}{e_i} \var{v_i}$ and $\var{v_i}\in \meta{meta}(\beta)$, as well as $\psem{s}{\rho}{m_i} \var{u_i}$. By the induction 
        hypothesis, we obtain $\psem{s'}{\rho}{e_i} \var{v_i}$ and $\psem{s'}{\rho}{m_i} \var{u_i}$ for all $i$.
        Therefore, applying \rref{E-Mapping}, we conclude
        $\psem{s'}{\rho}{[\{e_i \Rightarrow m_i\}_{i \in [1,n]}]_{\annot{\mapping{\beta}{\mu}}}} \var{f}$.
        \case{\rref{E-MappingUpd}}
        Similar to the case of \rref{E-Mapping}.
     \end{proofcases}
\end{proof}

\begin{lemma}[Weakening of Storage (Typing)]\label{lem:valuetyp-storage-weak}
Assume that $s \subseteq s'$. Then, the following hold.
    \begin{enumerate}[start=1,label={(\bfseries P\arabic*)}]
        \item\label{weak-valuetyp-storage:alpha} If $\Sigma \vdash v :_s \alpha$ then $\Sigma \vdash v :_{s'} \alpha$ and 
        $\vsem{\Sigma \vdash v :_s \alpha} = \vsem{\Sigma \vdash v :_{s'} \alpha}$.
        \item\label{weak-valuetyp-storage:sigma} If $\Sigma \vdash v :_s \sigma$ then $\Sigma \vdash v :_{s'} \sigma$ and
        $\vsem{\Sigma \vdash v :_s \sigma} = \vsem{\Sigma \vdash v :_{s'} \sigma}$.
        \item\label{weak-valuetyp-storage:env} If $\Sigma \vdash \rho :_s I$ then $\Sigma \vdash \rho :_{s'} I$ and
        $\vsem{\Sigma \vdash \rho :_s I} = \vsem{\Sigma \vdash \rho :_{s'} I}$.
    \end{enumerate}
\end{lemma}
\begin{proof}
We fist prove statement~\ref{weak-valuetyp-storage:alpha} by inversion on the derivation $\Sigma \vdash v :_s \alpha$.
\begin{proofcases}
    \case{$\Sigma \vdash v :_s \beta$ (rule \rref{V-BaseValAlpha})}
    We have that $\vdash v : \beta$ from which we also derive $\Sigma \vdash v :_{s'} \beta$.
    We also have that 
    $\vsem{\Sigma \vdash v :_s \beta} = \vsem{\vdash v : \beta} = \vsem{\Sigma \vdash v :_{s'} \beta}$.
    \case{$ \Sigma \vdash \ell :_s \kw{address}_A$ (rule \rref{V-AddrIsContract})}
    We have that:
    \begin{enumerate}[start=1,label={(\bfseries H\arabic*)}]
    \item\label{weak-valuetyp-storage:addr:H1} $\ell \in \dom{s}$
    \item\label{weak-valuetyp-storage:addr:H2} $A \in \dom{\Sigmastore}$
    \item\label{weak-valuetyp-storage:addr:H3} $s(\ell).{\kw{type}} = A$
    \item\label{weak-valuetyp-storage:addr:H4} $(\forall \, x. ~\, x \in \dom{s(\ell)} \Leftrightarrow (\exists \, \sigma. ~\, \Sigmastore(A)(x) = \sigma))$
    \item\label{weak-valuetyp-storage:addr:H5} $(\forall \, x,\sigma. ~\, \Sigmastore(A)(x) = \sigma ~\Rightarrow ~ \Sigma \vdash s(\ell) (x) :_s \sigma)$
    \end{enumerate}
    We need to show that:
    \begin{enumerate}[start=1,label={(\bfseries G\arabic*)}]
    \item $\ell \in \dom{s'}$, which follows from~\ref{weak-valuetyp-storage:addr:H1} and the fact that $s \subseteq s'$.
    \item $A \in \dom{\Sigmastore}$ which follows from~\ref{weak-valuetyp-storage:addr:H2}.
    \item $s'(\ell).{\kw{type}} = A$ which follows from~\ref{weak-valuetyp-storage:addr:H3} and the fact that $s \subseteq s'$.
    \item $(\forall \, x. ~\, x \in \dom{s'(\ell)} \Leftrightarrow (\exists \, \sigma. ~\, \Sigmastore(A)(x) = \sigma))$,
    which follows from the fact that $s(\ell) = s'(\ell)$.
    \item $(\forall \, x,\sigma. ~\, \Sigmastore(A)(x) = \sigma ~\Rightarrow ~ \Sigma \vdash s(\ell) (x) :_{s'} \sigma)$. \\
    Let $x$ and $\sigma$ such that $\Sigmastore(A)(x) = \sigma$. From~\ref{weak-valuetyp-storage:addr:H5} we derive $\Sigma \vdash s(\ell) (x) :_s \sigma$.
    From the induction hypothesis we get $\Sigma \vdash s(\ell) (x) :_{s'} \sigma$.
    \end{enumerate}
    Lastly, we obtain the following. 
    \[\begin{array}{r l}
      & \vsem{\Sigma \vdash \ell :_s \kw{address}_A} \\
    = & \{ \addrfield = \ell, x_1 = \vsem{\Sigma \vdash s(\ell)(x_1) :_s \Sigmastore(A)(x_1)}, \dots, x_n =\vsem{\Sigma \vdash s(\ell)(x_n) :_s \Sigmastore(A)(x_n)} \} \\
    = & \{ \addrfield = \ell, x_1 = \vsem{\Sigma \vdash s(\ell)(x_1) :_{s'} \Sigmastore(A)(x_1)}, \dots, x_n =\vsem{\Sigma \vdash s(\ell)(x_n) :_{s'} \Sigmastore(A)(x_n)} \} \\
      & \hfill \text{by the induction hypothesis for the judgments $\Sigma \vdash s(\ell)(x_i) :_{s} \Sigmastore(A)(x_i)$}\\ 
    = & \vsem{\Sigma \vdash \ell :_{s'} \kw{address}_A} \\
    \end{array}\]
\end{proofcases}
    We then proceed with proving~\ref{weak-valuetyp-storage:sigma} by inversion of the derivation $\Sigma \vdash v :_s \sigma$.
\begin{proofcases}
    \case{$\Sigma \vdash v :_s \mu$ (rule \rref{V-MappingVal})}
    We have that $\vdash v : \mu$ from which we also derive $\Sigma \vdash v :_{s'} \mu$.
    We also have that 
    $\vsem{\Sigma \vdash v :_s \mu} = \vsem{\vdash v : \mu} = \vsem{\Sigma \vdash v :_{s'} \mu}$.
    \case{$\Sigma \vdash v :_s \alpha$ (rule \rref{V-ABIVal})}
    From~\ref{weak-valuetyp-storage:alpha} we have that $\Sigma \vdash v :_{s'} \alpha$.
    We also have that 
    $\vsem{\Sigma \vdash v :_s \alpha} = \vsem{\vdash v : \alpha} = \vsem{\Sigma \vdash v :_{s'} \alpha}$.
    \case{$\Sigma \vdash \ell :_s A$ (rule \rref{V-Contract})}
    We have that $\Sigma \vdash \ell :_s \kw{address}_A$. By~\ref{weak-valuetyp-storage:alpha} we get that $\Sigma \vdash \ell :_{s'} \kw{address}_A$,
    from which we obtain $\Sigma \vdash \ell :_{s'} A$ through the~\rref{V-Contract} rule. 
    We get $\vsem{\Sigma \vdash \ell :_s A} = \vsem{\Sigma \vdash \ell :_{s'} A}$ the same way as in the case of the~\rref{V-AddrIsContract} rule.
\end{proofcases}
Finally, we prove~\ref{weak-valuetyp-storage:env}. From the hypothesis $\Sigma \vdash \rho :_s I$ we have that 
\begin{enumerate}[start=1,label={(\bfseries H\arabic*)}]
\item\label{weak-valuetyp-storage:env:H1} $\dom{\rho} = \dom{I} \cup \{\kw{caller}, \kw{origin}, \kw{callvalue}\}$
\item\label{weak-valuetyp-storage:env:H2} $\forall \, x \in \dom{I}. ~\, \Sigma \vdash \rho(x) :_s I(x)$
\item\label{weak-valuetyp-storage:env:H3} $\vdash \rho(\kw{caller}) : \kw{address}$
\item\label{weak-valuetyp-storage:env:H4} $\vdash \rho(\kw{origin}) : \kw{address}$
\item\label{weak-valuetyp-storage:env:H5} $\vdash \rho(\kw{callvalue}) : \kw{uint256}$
\end{enumerate}
From~\ref{weak-valuetyp-storage:env:H1} and the proved statement~\ref{weak-valuetyp-storage:alpha} we have that 
$\forall \, x \in \dom{I}. ~\, \Sigma \vdash \rho(x) :_{s'} I(x)$ and that 
$\vsem{ \Sigma \vdash \rho(x) :_{s} I(x)} = \vsem{ \Sigma \vdash \rho(x) :_{s'} I(x)}$.
Together with~\ref{weak-valuetyp-storage:env:H1},~\ref{weak-valuetyp-storage:env:H3},~\ref{weak-valuetyp-storage:env:H4}, and~\ref{weak-valuetyp-storage:env:H5}, 
we can now establish $\Sigma \vdash \rho :_{s'} I$ and that 
$\vsem{\Sigma \vdash \rho :_{s} I} = \vsem{\Sigma \vdash \rho :_{s'} I}$.
\end{proof}

\begin{lemma}[Strengthening of Store Typing]\label{lem:valuetyp-styp-strength}
Assume that $\Sigma' \subseteq\Sigma$. Then, the following hold.
    \begin{enumerate}[start=1,label={(\bfseries P\arabic*)}]
        \item\label{valuetyp-styp-strength:alpha} If $\Sigma' \vdash \alpha~\meta{wf}$ and 
        $\Sigma \vdash v :_s \alpha$ then $\Sigma' \vdash v :_{s} \alpha$ and
        $\vsem{\Sigma \vdash v :_s \alpha} = \vsem{\Sigma' \vdash v :_{s} \alpha}$.        
        \item\label{valuetyp-styp-strength:sigma} If $\Sigma' \vdash \sigma~\meta{wf}$ and 
        $\Sigma \vdash v :_s \sigma$ then $\Sigma' \vdash v :_{s} \sigma$ and 
        $\vsem{\Sigma \vdash v :_s \sigma} = \vsem{\Sigma' \vdash v :_{s} \sigma}$.
        \item\label{valuetyp-styp-strength:env} If $\Sigma' \vdash I~\meta{wf}$ and 
        $\Sigma \vdash \rho :_s I$ then $\Sigma' \vdash \rho :_{s} I$ and
        $\vsem{\Sigma \vdash \rho :_s I} = \vsem{\Sigma' \vdash \rho :_{s} I}$.
    \end{enumerate}
\end{lemma}
\begin{proof}
We fist prove statements~\ref{valuetyp-styp-strength:alpha} and~\ref{valuetyp-styp-strength:sigma} by mutual induction on the derivations $\Sigma \vdash v :_s \alpha$ 
and $\Sigma \vdash v :_s \sigma$. We begin with the cases of $\Sigma \vdash v :_s \alpha$.
\begin{proofcases}
    \case{$\Sigma \vdash v :_s \beta$ (rule \rref{V-BaseValAlpha})}
    We have that $\vdash v : \beta$ from which we also derive $\Sigma' \vdash v :_s \beta$.
    We also have that 
    $\vsem{\Sigma \vdash v :_s \beta} = \vsem{\vdash v : \beta} = \vsem{\Sigma' \vdash v :_{s} \beta}$.    
    %%% 
    \case{$ \Sigma \vdash \ell :_s \kw{address}_A$ (rule \rref{V-AddrIsContract})}
    We have that:
    \begin{enumerate}[start=1,label={(\bfseries H\arabic*)}]
    \item\label{valuetyp-styp-strength:alpha:H1} $\ell \in \dom{s}$
    \item\label{valuetyp-styp-strength:alpha:H2} $A \in \dom{\Sigmastore}$
    \item\label{valuetyp-styp-strength:alpha:H3} $s(\ell).{\kw{type}} = A$
    \item\label{valuetyp-styp-strength:alpha:H4} $(\forall \, x. ~\, x \in \dom{s(\ell)} \Leftrightarrow (\exists \, \sigma. ~\, \Sigmastore(A)(x) = \sigma))$
    \item\label{valuetyp-styp-strength:alpha:H5} $(\forall \, x,\sigma. ~\, \Sigmastore(A)(x) = \sigma ~\Rightarrow ~ \Sigma \vdash s(\ell) (x) :_s \sigma)$
    \end{enumerate}
    We need to show that
    \begin{enumerate}[start=1,label={(\bfseries G\arabic*)}]
    \item\label{valuetyp-styp-strength:alpha:G1} $\ell \in \dom{s}$, which follows directly from~\ref{valuetyp-styp-strength:alpha:H1}.
    \item\label{valuetyp-styp-strength:alpha:G2} $A \in \dom{\Sigmastore}$, which follows from the fact that $\Sigma' \vdash \kw{address}_A~\meta{wf}$.
    \item \label{valuetyp-styp-strength:alpha:G3} $s(\ell).{\kw{type}} = A$, which follows directly from~\ref{valuetyp-styp-strength:alpha:H3}.
    \item\label{valuetyp-styp-strength:alpha:G4} $(\forall \, x. ~\, x \in \dom{s(\ell)} \Leftrightarrow (\exists \, \sigma. ~\, \Sigmastorep(A)(x) = \sigma))$.\\
    We have that $A \in \dom{\Sigmastore}$, $A \in \dom{\Sigmastorep}$ and $\Sigma' \subseteq\Sigma$. Therefore, 
    the goal follows from the fact that for all $x$, $\Sigmastorep(A)(x) = \Sigmastore(A)(x)$.
    \item\label{valuetyp-styp-strength:alpha:G5} $(\forall \, x,\sigma. ~\, \Sigmastorep(A)(x) = \sigma ~\Rightarrow ~ \Sigma' \vdash s(\ell) (x) :_s \sigma)$. \\
    Let $x$ and $\sigma$ such that $\Sigmastorep(A)(x) = \sigma$. We again have that $\Sigmastore(A)(x) = \sigma$. From~\ref{valuetyp-styp-strength:alpha:H4} we get that  $\Sigma \vdash s(\ell) (x) :_s \sigma$. By the induction hypothesis we get that 
     $\Sigma' \vdash s(\ell) (x) :_s \sigma$ which proves the goal.
    \end{enumerate}
    Lastly, we obtain the following. 
    \[\begin{array}{r l}
      & \vsem{\Sigma \vdash \ell :_s \kw{address}_A} \\
    = & \{ \addrfield = \ell, x_1 = \vsem{\Sigma \vdash s(\ell)(x_1) :_s \Sigmastore(A)(x_1)}, \dots, x_n =\vsem{\Sigma \vdash s(\ell)(x_n) :_s \Sigmastore(A)(x_n)} \} \\
    = & \{ \addrfield = \ell, x_1 = \vsem{\Sigma' \vdash s(\ell)(x_1) :_{s} \Sigmastore(A)(x_1)}, \dots, x_n =\vsem{\Sigma' \vdash s(\ell)(x_n) :_{s} \Sigmastore(A)(x_n)} \} \\
      & \hfill \text{by the induction hypothesis for the judgments $\Sigma \vdash s(\ell)(x_i) :_{s} \Sigmastore(A)(x_i)$}\\ 
    = & \vsem{\Sigma' \vdash \ell :_{s} \kw{address}_A} \\
    \end{array}\]
\end{proofcases}
     We then proceed with the cases of the derivation $\Sigma \vdash v :_s \sigma$.
\begin{proofcases}
    \case{$\Sigma \vdash v :_s \mu$ (rule \rref{V-MappingVal})}
    We have that $\vdash v : \mu$ from which we also derive $\Sigma' \vdash v :_{s'} \mu$.
    We also have that 
    $\vsem{\Sigma \vdash v :_s \mu} = \vsem{\vdash v : \mu} = \vsem{\Sigma' \vdash v :_{s} \mu}$.
    \case{$\Sigma \vdash v :_s \alpha$ (rule \rref{V-ABIVal})}
    From the induction hypothesis we derive that $\Sigma' \vdash v :_{s'} \alpha$.
    We also have that 
    $\vsem{\Sigma \vdash v :_s \alpha} = \vsem{\vdash v : \alpha} = \vsem{\Sigma' \vdash v :_{s} \alpha}$.
    \case{$\Sigma \vdash \ell :_s A$ (rule \rref{V-Contract})} % the proof is identical to rule \rref{V-AddrIsContract})
    We have that $\Sigma \vdash \ell :_s \kw{address}_A$. By the induction hypothesis we get that $\Sigma' \vdash \ell :_{s} \kw{address}_A$, and we conclude with~\rref{V-Contract} that $\Sigma' \vdash \ell :_{s} A$.
    We get $\vsem{\Sigma \vdash \ell :_s A} = \vsem{\Sigma' \vdash \ell :_{s} A}$ the same way as in the case of the~\rref{V-AddrIsContract} rule.
\end{proofcases}
Finally, we prove~\ref{valuetyp-styp-strength:env}. From the hypothesis $\Sigma \vdash \rho :_s I$ we have that 
\begin{enumerate}[start=1,label={(\bfseries H\arabic*)}]
\item\label{valuetyp-styp-strength:env:H1} $\dom{\rho} = \dom{I} \cup \{\kw{caller}, \kw{origin}, \kw{callvalue}\}$
\item\label{valuetyp-styp-strength:env:H2} $\forall \, x \in \dom{I}. ~\, \Sigma \vdash \rho(x) :_s I(x)$
\item\label{valuetyp-styp-strength:env:H3} $\vdash \rho(\kw{caller}) : \kw{address}$
\item\label{valuetyp-styp-strength:env:H4} $\vdash \rho(\kw{origin}) : \kw{address}$
\item\label{valuetyp-styp-strength:env:H5} $\vdash \rho(\kw{callvalue}) : \kw{uint256}$
\end{enumerate}
From~\ref{valuetyp-styp-strength:env:H2} and the proved statement~\ref{valuetyp-styp-strength:alpha} we have that 
$\forall \, x \in \dom{I}. ~\, \Sigma' \vdash \rho(x) :_{s} I(x)$ and that 
$\vsem{ \Sigma \vdash \rho(x) :_{s} I(x)} = \vsem{ \Sigma' \vdash \rho(x) :_{s} I(x)}$.
Together with~\ref{valuetyp-styp-strength:env:H1},~\ref{valuetyp-styp-strength:env:H3},~\ref{valuetyp-styp-strength:env:H4}, and~\ref{valuetyp-styp-strength:env:H5}, we can now establish $\Sigma' \vdash \rho :_{s} I$ and that 
$\vsem{\Sigma \vdash \rho :_{s} I} = \vsem{\Sigma' \vdash \rho :_{s} I}$.
\end{proof}

\begin{lemma}[Weakening of Store Typing]\label{lem:valuetyp-storagetyp-weak}
Assume that $\Sigma \subseteq \Sigma'$. Then, the following hold.
    \begin{enumerate}[start=1,label={(\bfseries P\arabic*)}]
        \item\label{valuetyp-styp-weak:alpha} If $\Sigma \vdash v :_s \alpha$ then $\Sigma' \vdash v :_{s} \alpha$ and 
        $\vsem{\Sigma \vdash v :_s \alpha} = \vsem{\Sigma' \vdash v :_{s} \alpha}$.
        \item\label{valuetyp-styp-weak:sigma} If $\Sigma \vdash v :_s \sigma$ then $\Sigma' \vdash v :_{s} \sigma$ and
        $\vsem{\Sigma \vdash v :_s \sigma} = \vsem{\Sigma' \vdash v :_{s} \sigma}$.
        \item\label{valuetyp-styp-weak:env} If $\Sigma \vdash \rho :_s I$ then $\Sigma' \vdash \rho :_{s} I$ and 
        $\vsem{\Sigma \vdash \rho :_s I} = \vsem{\Sigma' \vdash \rho :_{s} I}$.
    \end{enumerate}
\end{lemma}
\begin{proof}
We fist prove statements~\ref{valuetyp-styp-weak:alpha} and~\ref{valuetyp-styp-weak:sigma} by mutual induction on the 
derivations $\Sigma \vdash v :_s \alpha$ 
and $\Sigma \vdash v :_s \sigma$. We begin with the cases of $\Sigma \vdash v :_s \alpha$.
\begin{proofcases}
    \case{$\Sigma \vdash v :_s \beta$ (rule \rref{V-BaseValAlpha})}
    We have that $\vdash v : \beta$ from which we also derive $\Sigma' \vdash v :_s \beta$.
    We also have that 
    $\vsem{\Sigma \vdash v :_s \beta} = \vsem{\vdash v : \beta} = \vsem{\Sigma' \vdash v :_{s} \beta}$.    
    \case{$ \Sigma \vdash \ell :_s \kw{address}_A$ (rule \rref{V-AddrIsContract})}
    We have that:
     \begin{enumerate}[start=1,label={(\bfseries H\arabic*)}]
    \item\label{valuetyp-styp-weak:alpha:H1}  $\ell \in \dom{s}$
    \item\label{valuetyp-styp-weak:alpha:H2}  $A \in \dom{\Sigmastore}$
    \item\label{valuetyp-styp-weak:alpha:H3}  $s(\ell).{\kw{type}} = A$
    \item\label{valuetyp-styp-weak:alpha:H4} $(\forall \, x. ~\, x \in \dom{s(\ell)} \Leftrightarrow (\exists \, \sigma. ~\, \Sigmastore(A)(x) = \sigma))$
    \item\label{valuetyp-styp-weak:alpha:H5}  $(\forall \, x,\sigma. ~\, \Sigmastore(A)(x) = \sigma ~\Rightarrow ~ \Sigma \vdash s(\ell) (x) :_s \sigma)$
    \end{enumerate}
    We need to show that
    \begin{enumerate}[start=1,label={(\bfseries G\arabic*)}]
    \item\label{valuetyp-styp-weak:alpha:G1} $\ell \in \dom{s}$, which follows directly from~\ref{valuetyp-styp-weak:alpha:H1}.
    \item\label{valuetyp-styp-weak:alpha:G2} $A \in \dom{\Sigmastorep}$, which follows from~\ref{valuetyp-styp-weak:alpha:H2} and the fact that $\Sigma \subseteq \Sigma'$.
    \item \label{valuetyp-styp-weak:alpha:G3} $s(\ell).{\kw{type}} = A$, which follows directly from~\ref{valuetyp-styp-weak:alpha:H3}.
    \item\label{valuetyp-styp-weak:alpha:G4} $(\forall \, x. ~\, x \in \dom{s(\ell)} \Leftrightarrow (\exists \, \sigma. ~\, \Sigmastorep(A)(x) = \sigma))$.\\
        We have that $A \in \dom{\Sigmastore}$, $A \in \dom{\Sigmastorep}$ and $\Sigma \subseteq \Sigma'$. Therefore, 
        the goal follows from the fact that for all $x$, $\Sigmastorep(A)(x) = \Sigmastore(A)(x)$.
        % % => 
        % Let $x\in \dom{s(\ell)}$. From (H2) we get $\Sigma \vdash s(\ell) (x) :_s \sigma$. From the inductive hypothesis we 
        % get that $\Sigma' \vdash s(\ell) (x) :_s \sigma$, which proves the goal.
        
        % % <=
        % To prove the other direction, let $x$ such that $\Sigmastorep(A)(x) = \sigma$ for some $\sigma$.
        % Because $A \in \dom{\Sigmastore}$ and $\Sigma \subseteq \Sigma'$, we have that  $\Sigmastore(A)(x) = \sigma$. 
        % From (H3) we get that $x \in \dom{s(\ell)}$ which proves the goal. 
    \item\label{valuetyp-styp-weak:alpha:G5} $(\forall \, x,\sigma. ~\, \Sigmastorep(A)(x) = \sigma ~\Rightarrow ~ \Sigma' \vdash s(\ell) (x) :_s \sigma)$. \\
        Let $x$ and $\sigma$ such that $\Sigmastorep(A)(x) = \sigma$. We again have that $\Sigmastore(A)(x) = \sigma$. 
        From~\ref{valuetyp-styp-weak:alpha:H5} we get that  $\Sigma \vdash s(\ell) (x) :_s \sigma$. By the induction hypothesis we get that 
        $\Sigma' \vdash s(\ell) (x) :_s \sigma$ which proves the goal.
    \end{enumerate}
    Lastly, we obtain the following. 
    \[\begin{array}{r l}
      & \vsem{\Sigma \vdash \ell :_s \kw{address}_A} \\
    = & \{ \addrfield = \ell, x_1 = \vsem{\Sigma \vdash s(\ell)(x_1) :_s \Sigmastore(A)(x_1)}, \dots, x_n =\vsem{\Sigma \vdash s(\ell)(x_n) :_s \Sigmastore(A)(x_n)} \} \\
    = & \{ \addrfield = \ell, x_1 = \vsem{\Sigma' \vdash s(\ell)(x_1) :_{s} \Sigmastore(A)(x_1)}, \dots, x_n =\vsem{\Sigma' \vdash s(\ell)(x_n) :_{s} \Sigmastore(A)(x_n)} \} \\
      & \hfill \text{by the induction hypothesis for the judgments $\Sigma \vdash s(\ell)(x_i) :_{s} \Sigmastore(A)(x_i)$}\\ 
    = & \vsem{\Sigma' \vdash \ell :_{s} \kw{address}_A} \\
    \end{array}\]
\end{proofcases}
We then proceed with the cases of the derivation $\Sigma \vdash v :_s \sigma$.
\begin{proofcases}
    \case{$\Sigma \vdash v :_s \mu$ (rule \rref{V-MappingVal})}
    We have that $\vdash v : \mu$ from which we also derive $\Sigma' \vdash v :_{s'} \mu$.
    We also have that 
    $\vsem{\Sigma \vdash v :_s \mu} = \vsem{\vdash v : \mu} = \vsem{\Sigma' \vdash v :_{s} \mu}$.
    \case{$\Sigma \vdash v :_s \alpha$ (rule \rref{V-ABIVal})}
    From the induction hypothesis we derive that $\Sigma' \vdash v :_{s'} \alpha$.
    We also have that 
    $\vsem{\Sigma \vdash v :_s \alpha} = \vsem{\vdash v : \alpha} = \vsem{\Sigma' \vdash v :_{s} \alpha}$.
    \case{$\Sigma \vdash \ell :_s A$ (rule \rref{V-Contract})} 
    We have that $\Sigma \vdash \ell :_s \kw{address}_A$. By the induction hypothesis we get that 
    $\Sigma' \vdash \ell :_{s} \kw{address}_A$, and we conclude with~\rref{V-Contract} that 
    $\Sigma' \vdash \ell :_{s} A$.
    We get $\vsem{\Sigma \vdash \ell :_s A} = \vsem{\Sigma' \vdash \ell :_{s} A}$ the same way as in 
    the case of the~\rref{V-AddrIsContract} rule.
\end{proofcases}

    Finally, we prove~\ref{valuetyp-styp-weak:env}. From the hypothesis $\Sigma \vdash \rho :_s I$ we have that 
\begin{enumerate}[start=1,label={(\bfseries H\arabic*)}]
\item\label{valuetyp-styp-weak:env:H1} $\dom{\rho} = \dom{I} \cup \{\kw{caller}, \kw{origin}, \kw{callvalue}\}$
\item\label{valuetyp-styp-weak:env:H2} $\forall \, x \in \dom{I}. ~\, \Sigma \vdash \rho(x) :_s I(x)$
\item\label{valuetyp-styp-weak:env:H3} $\vdash \rho(\kw{caller}) : \kw{address}$
\item\label{valuetyp-styp-weak:env:H4} $\vdash \rho(\kw{origin}) : \kw{address}$
\item\label{valuetyp-styp-weak:env:H5} $\vdash \rho(\kw{callvalue}) : \kw{uint256}$
\end{enumerate}
From~\ref{valuetyp-styp-weak:env:H2} and the proved statement~\ref{valuetyp-styp-weak:alpha} we have that 
$\forall \, x \in \dom{I}. ~\, \Sigma' \vdash \rho(x) :_{s} I(x)$ and that 
$\vsem{ \Sigma \vdash \rho(x) :_{s} I(x)} = \vsem{ \Sigma' \vdash \rho(x) :_{s} I(x)}$.
Together with~\ref{valuetyp-styp-weak:env:H1},~\ref{valuetyp-styp-weak:env:H3},~\ref{valuetyp-styp-weak:env:H4}, and~\ref{valuetyp-styp-weak:env:H5}, we can now establish $\Sigma' \vdash \rho :_{s} I$ and that 
$\vsem{\Sigma \vdash \rho :_{s} I} = \vsem{\Sigma' \vdash \rho :_{s} I}$.
\end{proof}

\begin{lemma}\label{lem:value}
    If $\vdash u : \mu$ then $u \in \meta{Value}$.
\end{lemma}
\begin{proof}
    By the structure of $\mu$.
    \begin{proofcases}
        \case{Case $\beta$}
            If $\vdash u : \beta$, then  by inversion of \rref{V-Int},\rref{V-Bool}, or \rref{V-Addr},
             we obtain $u \in \meta{BaseValue}$ and therefore $u \in \meta{Value}$.
        \case{Case $\mapping{\beta}{\mu}$}
        From  $\vdash u : \mapping{\beta}{\mu}$ follows by inversion of \rref{V-Mapping} 
        that $u \in \meta{meta}(\beta) \to \meta{Value}$ and thus $u \in \meta{Value}$.
    \end{proofcases}
\end{proof}

\begin{lemma}\label{lem:defaultmu}
    It holds that $\vdash \meta{default}(\mu) : \mu$.
\end{lemma}
\begin{proof}
    By induction on the structure of $\mu$.
    \begin{proofcases}
        \case{Case $\beta$}
            \begin{proofcases}
                \case{Case $\iota$}        By definition $\meta{default}(\iota)=0 \in \mathbb{Z}$. Hence, by \rref{V-Int} we have $\vdash \meta{default}(\iota) : \iota$.
                \case{Case $\kw{bool}$}     By definition  $\meta{default}(\kw{bool})=\meta{True} \in \mathbb{B}$. Hence, by \rref{V-Bool}  we have $\vdash \meta{default}(\kw{bool}) : \kw{bool}$.
                \case{Case $\kw{address}$}  By definition  $\meta{default}(\kw{address})=0 \in \meta{Addr}$. Hence, by \rref{V-Addr} we have $\vdash \meta{default}(\kw{address}) : \kw{address}$.
            \end{proofcases}
        \case{Case $\mapping{\beta}{\mu}$}
            By definition $\meta{default}(\mapping{\beta}{\mu})= (x \in \meta{meta}(\beta) \mapsto \meta{default}(\mu))$.
            From the induction hypothesis follows $\vdash \meta{default}(\mu) : \mu$. 
            Let now $v$ be such that $\vdash v : \beta$. Consider that $\meta{default}(\mapping{\beta}{\mu})(v) = \meta{default}(\mu)$, 
            for which we just inferred from the hypothesis $\vdash \meta{default}(\mu) : \mu$. Hence, the second premise of \rref{V-Mapping}
            holds. For the first one, we apply \Cref{lem:value} to get that $\meta{default}(\mu) \in \meta{Value}$  and conclude that 
            the  $ (x \in \meta{meta}(\beta) \mapsto \meta{default}(\mu)) \; \in \; \meta{meta}(\beta) \to \meta{Value}$. 
            
            Therefore, we can apply \rref{V-Mapping} to obtain  $\vdash \meta{default}(\mapping{\beta}{\mu}) : \mapping{\beta}{\mu}$.
    \end{proofcases}
\end{proof}

\begin{lemma}[Well-typed Update Preserves Typing]\label{lem:update-preserves-typing}
Assume that
    \begin{enumerate}[start=1,label={(\bfseries H\arabic*)}]
        \item\label{wt_update_ptyp:H1} $\Sigma \vdash v :_s \sigma$
        \item\label{wt_update_ptyp:H2} $\Sigma \vdash \ell :_s A$ (or $\Sigma \vdash \ell :_s \kw{address}_A$)
        \item\label{wt_update_ptyp:H3} $\Sigma \vdash v' :_s \Sigmastore(A)(x)$
        \item\label{wt_update_ptyp:H4} $s' = s[\ell \mapsto s(\ell)[y \mapsto v']]$
    \end{enumerate}
Then, $\Sigma \vdash v :_{s'} \sigma$.
\end{lemma}
\begin{proof}
We proceed by induction on the well-founded measure $\meta{len}(\Sigma,\sigma)$.
We do inversion on~\ref{wt_update_ptyp:H1}. We consider the following cases. 
\begin{proofcases}
\case{$\Sigma \vdash v :_s \mu$ (rule \rref{V-MappingVal})}
From the hypothesis we obtain that $\vdash v : \mu$, from which we trivially 
derive that $\Sigma \vdash v :_{s'} \mu$ 
% and $\vsem{\Sigma \vdash v :_{s'} \mu} = \vsem{\vdash v : \mu} = \vsem{\Sigma \vdash v :_{s} \mu}$.
%
\case{$\Sigma \vdash v :_s \alpha$ (rule \rref{V-ABIVal})}
We invert the hypothesis and we consider the following cases.
\begin{proofcases}
\case{$\Sigma \vdash v :_s \beta$ (rule \rref{V-BaseValAlpha})}
We obtain that $\vdash v : \beta$ from which we trivially get $\Sigma \vdash v :_{s'} \beta$
\case{$\Sigma \vdash \ell_B :_s \kw{address}_B$ (rule \rref{V-AddrIsContract})}
From inversion we obtain that 
\begin{enumerate}[start=1,label={(\bfseries P\arabic*)}]
\item\label{wt_update_ptyp_addr:P1} $\ell_B  \in \dom{s}$
\item\label{wt_update_ptyp_addr:P2} $B \in \dom{\Sigmastore}$
\item\label{wt_update_ptyp_addr:P3} $s(\ell_B).{\kw{type}} = B$
\item\label{wt_update_ptyp_addr:P4} $(\forall \, x. ~\, x \in \dom{s(\ell_B)} \Leftrightarrow (\exists \, \sigma. ~\, \Sigmastore(B)(x) = \sigma))$
\item\label{wt_update_ptyp_addr:P5} $(\forall \, x,\sigma. ~\, \Sigmastore(B)(x) = \sigma ~\Rightarrow ~ \Sigma \vdash s(\ell_B) (x) :_s \sigma)$
\end{enumerate} We need to show that $\Sigma \vdash \ell_B :_{s'} : \kw{address}_B$, therefore we need to establish the following.
\begin{enumerate}[start=1,label={(\bfseries G\arabic*)}]
\item\label{wt_update_ptyp_addr:G1} $\ell_B  \in \dom{s'}$, which we get from~\ref{wt_update_ptyp_addr:P1} and the fact that $\dom{s} = \dom{s'}$.
\item\label{wt_update_ptyp_addr:G2} $B \in \dom{\Sigmastore}$, which we get directly from~\ref{wt_update_ptyp_addr:P2}.
\item\label{wt_update_ptyp_addr:G3} $s'(\ell_B).{\kw{type}} = B$, which we get directly from~\ref{wt_update_ptyp_addr:P3} and the definition of $s'$. 
\item \label{wt_update_ptyp_addr:G4} $(\forall \, x. ~\, x \in \dom{s'(\ell_B)} \Leftrightarrow (\exists \, \sigma. ~\, \Sigmastore(B)(x) = \sigma))$,
which we derive by using the fact that $\dom{s(\ell_B)} = \dom{s'(\ell_B)}$.
\item\label{wt_update_ptyp_addr:G5} $(\forall \, x,\sigma. ~\, \Sigmastore(B)(x) = \sigma ~\Rightarrow ~ \Sigma \vdash s'(\ell_B)(x) :_{s'} \sigma)$ \\
Let $x$ and $\sigma$ such that $\Sigmastore(B)(x) = \sigma$. We need to show that $\Sigma \vdash s'(\ell_B)(x) :_{s'} \Sigmastore(B)(x)$.
We distinguish two cases.
\begin{proofcases}
\case{$\ell_B = \ell$ and $x = y$}
Because $\ell_B = \ell$,  we have that $\Sigma \vdash \ell : A$ and $\Sigma \vdash \ell : B$.
From \Cref{lem:uniqueness-of-value-typing} we get that $A = B$.
So it follows that $\Sigmastore(A)(x) = \Sigmastore(B)(x)$.
Then, we have that $s'(\ell_B)(x) = s'(\ell)(y) = v'$ (by~\ref{wt_update_ptyp:H4}) and
from~\ref{wt_update_ptyp:H3} we get $\Sigma \vdash v' :_s \Sigmastore(A)(x)$ and 
therefore $\Sigma \vdash v' :_s \Sigmastore(B)(x)$.
\case{$\ell_B \not = \ell$ or $x \not = y$} 
Then we have that $s'(\ell_B)(x) = s(\ell_B)(x)$
and from~\ref{wt_update_ptyp_addr:P5} we get $\Sigma \vdash s(\ell_B)(x) :_{s} \Sigmastore(B)(x)$.  
\end{proofcases}
Thus, we obtain that  \inlineequation[eq:wt_update_ptyp_addr]{\Sigma \vdash s'(\ell_B)(x) :_{s} \Sigmastore(B)(x)\,}. 
Note that $\meta{len}(\Sigma, \Sigmastore(B)(x)) < \meta{len}(\Sigma, B)$ by the 
definition of $\meta{len}(\cdot)$.
Therefore, we use the induction hypothesis on~\eqref{eq:wt_update_ptyp_addr},~\ref{wt_update_ptyp:H2},~\ref{wt_update_ptyp:H3} and~\ref{wt_update_ptyp:H4} 
and we obtain that 
$\Sigma \vdash s'(\ell_B)(x) :_{s'} \Sigmastore(B)(x)$ which proves the goal.
\end{enumerate}
\end{proofcases}
\case{$\Sigma \vdash \ell_B :_s B$ (rule \rref{V-Contract})}
From the hypothesis we obtain that 
$\Sigma \vdash \ell_B :_s \kw{address}_B$. By proceeding as in the above case, we get that 
$\Sigma \vdash \ell_B :_{s'} \kw{address}_B$, and we conclude with~\rref{V-Contract} that 
$\Sigma \vdash \ell_B :_{s'} B$.
\end{proofcases}
\end{proof}

\begin{lemma}[Well-typed Insertion Preserves Typing]\label{lem:insertion-preserves-typing}
Assume that
    \begin{enumerate}[start=1,label={(\bfseries H\arabic*)}]
        \item\label{wt_insert_ptyp:H1}  $\Sigma \vdash v :_{s} \sigma$
        \item\label{wt_insert_ptyp:H2}  $\Sigma \vdash \ell :_{s} A$ (or $\Sigma \vdash \ell :_{s} \kw{address}_A$)
        \item\label{wt_insert_ptyp:H3}  $\Sigma \vdash \rho :_{s} I$
        \item\label{wt_insert_ptyp:H4}  $\Sigma ; I \vdash^{\kw{S}}_{\var{A}, \kw{U}} \var{ref} : \sigma'$
        \item\label{wt_insert_ptyp:H5}  $\Sigma \vdash v' :_{s} \sigma'$
        \item\label{wt_insert_ptyp:H6}  $\ins{s; \rho ; \var{ref} ; v'}{s'}$
    \end{enumerate}
Then,  $\Sigma \vdash v :_{s'} \sigma$.
\end{lemma}

\begin{proof}
    By case splitting on~\ref{wt_insert_ptyp:H4}. 
    \begin{proofcases}
        \case{$\Sigma; I \vdash^{\kw{S}}_{A, \kw{U}} x : \sigma'$ (rule~\rref{T-Storage})}
        From the premises we get that $x : \sigma' \in \Sigmastore(A)$ and thus by~\ref{wt_insert_ptyp:H5} we obtain
         that \linebreak 
         \inlineequation[eq:wt_insert_ptyp:tstore]{\Sigma \vdash v' :_{s} \Sigmastore(A)(x)}.
        By inversion on~\ref{wt_insert_ptyp:H6} (rule~\rref{E-InsStorage}) we get that \linebreak
        \inlineequation[eq:wt_insert_ptyp:insert]{s' = s[\ell \mapsto s(\ell)[x \mapsto \var{v'}]]} and $x \in \dom{s(\ell)}$.
        We apply~\Cref{lem:update-preserves-typing} with~\ref{wt_insert_ptyp:H1},~\ref{wt_insert_ptyp:H2},~\eqref{eq:wt_insert_ptyp:tstore}
        and~\eqref{eq:wt_insert_ptyp:insert} and we conclude $\Sigma \vdash v :_{s'} \sigma$.

        \case{$\Sigma; I \vdash^{\kw{S}}_{A, \kw{U}} \var{ref}.x : \sigma'$ (rule~\rref{T-Field})}
        From the premises we get 
        \begin{enumerate}[start=1,label={(\bfseries P\arabic*)}]
            \item\label{wt_insert_ptyp_field:P1} $\Sigma; I \vdash^{\kw{S}}_{A, \kw{U}}  \var{ref} : B$
            \item\label{wt_insert_ptyp_field:P2}  $\Sigmastore(B)(x) = \sigma'$
        \end{enumerate}
        By inversion on~\ref{wt_insert_ptyp:H6} (rule~\rref{E-InsField}) we get
        \begin{enumerate}[start=1,label={(\bfseries I\arabic*)}]
            \item\label{wt_insert_ptyp_field:I1} $\psem{s}{\rho}{\var{ref}} (\var{\ell'}, \kw{U})$
            \item\label{wt_insert_ptyp_field:I2} $\var{\ell'} \in \meta{Addr}$
            \item\label{wt_insert_ptyp_field:I3} $x \in \dom{s(\ell')}$
            \item\label{wt_insert_ptyp_field:I4} $s' = s[\ell' \mapsto s(\ell')[x \mapsto v']]$
        \end{enumerate}
        By~\Cref{lem:ref-typesafety-untimed} on~\ref{wt_insert_ptyp_field:P1},~\ref{wt_insert_ptyp:H3}, and~\ref{wt_insert_ptyp:H2}
         we get that there exists an $\ell''$, such that $\psem{s^{\kw{U}}}{\rho}{\var{ref}} (\var{\ell''}, \kw{U})$ and 
         $\Sigma \vdash \ell'' :_{s} B$. By~\ref{wt_insert_ptyp_field:I1} and determinism of pointer semantics we obtain 
         that $\ell'' = \ell'$ and 
        \begin{equation}
            \label{eq:well-typed-insertion-1}
            \Sigma \vdash \ell' :_{s} B
        \end{equation}
        We know that 
        \begin{equation}
            \label{eq:well-typed-insertion-2}
            \Sigma \vdash v' :_{s} \Sigmastore(B)(x)
        \end{equation}
        by~\ref{wt_insert_ptyp:H5} and~\ref{wt_insert_ptyp_field:P2}. 
        We then apply~\Cref{lem:update-preserves-typing} with~\ref{wt_insert_ptyp:H1},~\eqref{eq:well-typed-insertion-1}, 
        \eqref{eq:well-typed-insertion-2}, and~\ref{wt_insert_ptyp_field:I4} to conclude $\Sigma \vdash v :_{s'} \sigma$.
    \end{proofcases}
\end{proof}

\section{Type Safety}
\begin{lemma}[Environment References Type Safety]
\label{lem:ethenv-typesafety}
    Assume that
    \begin{enumerate}[start=1,label={(\bfseries H\arabic*)}]
        \item\label{typsf:ethenv:wt}  $\Sigma; I \vdash_{\maybe{A}}  \var{env} : \alpha$
        \item\label{typsf:ethenv:env} $\Sigma \vdash \rho :_s I$
        \item\label{typsf:ethenv:loc} $\Sigma \vdash \ell :_s \maybe{A}$
    \end{enumerate}
    Then there exists $v$ such that $\rho ; \var{env} \bigstep_\ell v$ and $\Sigma \vdash v :_s \alpha$.
\end{lemma}

\begin{proof}
    By cases on~\ref{typsf:ethenv:wt}.
    \begin{proofcases}
        \case{$\Sigma; I \vdash_{\maybe{A}} \kw{caller} : \kw{address}$ (rule~\rref{T-Caller})}
        By inversion on~\ref{typsf:ethenv:env} we obtain:
        \begin{enumerate}[start=1,label={(\bfseries I\arabic*)}]
            \item\label{typsf:ethenv:rhoDom} $\dom{\rho} = \dom{I} \cup \{\kw{caller}, \kw{origin}, \kw{callvalue}\}$
         
            \item\label{typsf:ethenv:caller} $\vdash \rho(\kw{caller}) : \kw{address}$
        \end{enumerate}
        Since $\kw{caller} \in \dom{\rho}$, we can use~\rref{E-Caller} rule to get $\rho ; \kw{caller} \bigstep_{\ell} \rho(\kw{caller})$. We conclude with~\ref{typsf:ethenv:caller}.

        \case{$\Sigma; I \vdash_{\maybe{A}} \kw{origin} : \kw{address}$ (rule~\rref{T-Origin})}
        Similarly as the~\rref{T-Caller} case.

        \case{$\Sigma; I \vdash_{\maybe{A}} \kw{callvalue} : \kw{uint256}$ (rule~\rref{T-Callvalue})}
        Similarly as the~\rref{T-Caller} case.

        \case{$\Sigma; I \vdash_{A} \kw{this} : \kw{address}_{A}$ (rule~\rref{T-This})}
        We use the~\rref{E-This} rule to obtain $\rho ; \kw{this} \bigstep_{\ell} \ell$.
        By inversion on~\ref{typsf:ethenv:loc} through~\rref{V-Contract} rule we also obtain 
        $\Sigma \vdash \ell :_s \kw{address}_{A}$.
    \end{proofcases}
\end{proof}

\begin{lemma}[References Type Safety (Untimed)]
\label{lem:ref-typesafety-untimed}
    Assume that
    \begin{enumerate}[start=1,label={(\bfseries H\arabic*)}]
        \item\label{typsf:uref:wt} $\Sigma; I \vdash_{\maybe{A}, \kw{U}}^k \var{ref} : \sigma$
        \item\label{typsf:uref:env} $\Sigma \vdash \rho :_s I$
        \item\label{typsf:uref:loc} $\Sigma \vdash \ell :_s \maybe{A}$
    \end{enumerate}
    Then there exists $v$ such that $\psem{s}{\rho}{\var{ref}} (v, \kw{U})$ and $\Sigma \vdash v :_s \sigma$.
\end{lemma}

\begin{lemma}[Expression Type Safety (Untimed)]
\label{lem:expr-typesafety-untimed}
    Let $\Phi$ be a boolean expression as defined in \rref{T-Ctor}, respectively \rref{T-Trans}. Assume that
    \begin{enumerate}[start=1,label={(\bfseries H\arabic*)}]
        \item\label{typsf:uexp:wt} $\Sigma; I; \Phi \vdash_{\maybe{A}, \kw{U}} e : \beta$
        \item\label{typsf:uexp:env} $\Sigma \vdash \rho :_s I$
        \item\label{typsf:uexp:loc} $\Sigma \vdash \ell :_s \maybe{A}$
        \item\label{typsf:uexp:phi} $\psem{s}{\rho}{\Phi} \meta{True}$
    \end{enumerate}
    Then there exists $v$ such that $\psem{s}{\rho}{e} v$ and $\vdash v : \beta$.

    \end{lemma}
    \begin{proof}[Proof of~\Cref{lem:ref-typesafety-untimed,lem:expr-typesafety-untimed}]
    By mutual induction on the typing derivation~\ref{typsf:uref:wt}. We start with the cases from~\Cref{lem:ref-typesafety-untimed}.

    \begin{proofcases}
        \case{$\Sigma; I \vdash^{\kw{N}}_{\maybe{A}, \kw{U}}  \var{env} : \alpha$}
        The typing derivation ends with the~\rref{T-Environment} rule, which gives us
         \inlineequation[eq:typesafety-refs-ethenv]{\Sigma; I \vdash_{\maybe{A}}  \var{env} : \alpha}. 
        We can then use~\Cref{lem:ethenv-typesafety} on~\eqref{eq:typesafety-refs-ethenv},
        \ref{typsf:uref:env}, and~\ref{typsf:uref:loc}, to obtain $v$ such that $\rho ; \var{env} \bigstep_\ell v$ and $\Sigma \vdash v :_s \alpha$. 
        We then use~\rref{E-Environment} rule to conclude $s; \rho ; \var{env} \bigstep_\ell (v,\kw{U})$.
        
        \case{$\Sigma; I \vdash^{\kw{S}}_{\var{A}, \kw{U}} x : \sigma$}
        The typing derivation ends with the~\rref{T-Storage} rule, which gives us
         $(x : \sigma) \in \Sigmastore(A)$ and $x \notin \dom{I}$.
        From inversion of derivation of~\ref{typsf:uref:loc} of~\Cref{lem:ref-typesafety-untimed} via
         the~\rref{V-Contract} rule and the~\rref{V-AddrIsContract} rule we
         then deduce that $ x \in \dom{s(\ell)}$. From~\ref{typsf:uref:env} and the fact
         that $x \notin \dom{I}$, we get that 
         $x \notin \dom{\rho}$ so the~\rref{E-Storage} rule can be applied to get the 
         candidate
         $v = s(\ell)(x)$ with tag $\kw{U}$. We need to derive $\Sigma \vdash s(\ell)(x) :_s \sigma$.
          Again using inversion twice on~\ref{typsf:uref:loc} 
         of~\Cref{lem:ref-typesafety-untimed}, the premises of the derivation give us that since
         $(x : \sigma) \in \Sigmastore(\var{A})$, we also have $\Sigma \vdash s(\ell)(x) :_s \sigma$, which concludes the case.

        \case{$\Sigma; I \vdash^{\kw{N}}_{\maybe{A}, \kw{U}} x : \alpha$}
        From the~\rref{T-Calldata} typing rule we see that $(x : \alpha \in I)$.
        From~\ref{typsf:uref:env} we then get that $x \in \dom{\rho}$, so we can
        apply~\rref{E-Calldata} to get $\psem{s}{\rho}{x} (\rho(x), \kw{U})$. We need
         to derive
        $\Sigma \vdash \rho(x) :_s \alpha$, which we can obtain from inverting the
         derivation of~\ref{typsf:uref:env}.

        \case{$\Sigma; I \vdash^{\kw{N}}_{\maybe{A}, \kw{U}} \var{ref} ~\kw{as}~ A : A$}
        From the premise of the derivation rule~\rref{T-Coerce} we obtain that
        $\Sigma; I \vdash^{k}_{\maybe{A}, \kw{U}} \var{ref} : \kw{address}_A$. From the induction hypothesis
        follows that there is a $v$ such that $\psem{s}{\rho}{\var{ref}} (v,\kw{U})$ and $\Sigma \vdash v:_s \kw{address}_A$.
        With \rref{E-Coerce} we get  $\psem{s}{\rho}{\var{ref} ~\kw{as}~ A} (v,\kw{U})$. Finally, 
        by applying~\rref{V-Contract} we conclude  $\Sigma \vdash v:_s A$.

        % From the premises of the derivation rule~\rref{T-Coerce} we obtain that $(\var{ref} \mapsto A) \in P$.
        % By inversion on (H4) (and since $(\var{ref} \mapsto A) \in P$) we then have that
        % $\psem{s^{\kw{U}}}{\rho}{\var{ref}} (\ell', \kw{U})$ and $\Sigma \vdash \ell' :_s A$. Since by~\rref{E-Coerce}
        % $\var{ref}~\kw{as}~A$ evaluates to $(\ell', \kw{U})$, the case is concluded.

        \case{$\Sigma; I \vdash^{k}_{\maybe{A}, \kw{U}} \var{ref} : \kw{address}$}
        The premise of \rref{T-Upcast} gives us $\Sigma; I \vdash^{k}_{\maybe{A}, \kw{U}} \var{ref} : \kw{address}_A$.
        By induction hypothesis we know there is a $v$ such that  $\psem{s}{\rho}{\var{ref}} (v,\kw{U})$ and $\Sigma \vdash v:_s \kw{address}_A$.
        We still have to show $\Sigma \vdash v:_s \kw{address}$: From the premises of \rref{V-AddrIsContract} we obtain
        $v \in \dom{s}$  and thus $v \in \meta{Addr}$. We now apply \rref{V-Addr} and \rref{V-BaseValAlpha} to conclude  $\Sigma \vdash v:_s \kw{address}$.

        \case{$\Sigma; I \vdash^{k}_{\maybe{A}, \kw{U}} \var{ref}.x : \sigma$}
        From the premises of the~\rref{T-Field} rule, we get that $\Sigma ; I  \vdash^k_{\maybe{A},\kw{U}} \var{ref} : B$ and
        $\Sigmastore(B)(x) = \sigma$. By induction on the first premise we get that
        $\psem{s}{\rho}{\var{ref}} (\ell', \kw{U})$ and $\Sigma \vdash \ell' :_s B$. 
        By inversion
        on~\rref{V-Contract} and~\rref{V-AddrIsContract} and  $\Sigmastore(B)(x) = \sigma$ we obtain 
        the following: $\ell' \in \dom{s}$ and thus $\ell' \in \meta{Addr}$, $x \in \dom{s(\ell')}$, and
        $\Sigma \vdash s(\ell')(x) :_s \sigma$. Since we can now apply~\rref{E-Field} to obtain 
        $\psem{s}{\rho}{\var{ref}.x} (s(\ell')(x),\kw{U})$ the case is concluded.

        \case{$\Sigma; I \vdash^{k}_{\maybe{A}, \kw{U}} \var{ref}[e] : \mu$}
        The second premise of the~\rref{T-MapIndex} rule gives us by induction $\psem{s}{\rho}{e} \var{v}_e$ and
        $\Sigma \vdash \var{v}_e :_s \beta$. Induction on the first premise of the~\rref{T-MapIndex} rule gives us
        $\psem{s}{\rho}{\var{ref}} (\var{v}_r, \kw{U})$ and $\Sigma \vdash \var{v}_r :_s \mapping{\beta}{\mu}$,
        which implies by inversion of~\rref{V-Mapping}
         that $\var{v}_r \in \meta{meta}(\beta) \to \meta{Value}$ and
        $\Sigma \vdash \var{v_r} (\var{v_e}) :_s \mu$. The candidate for the value $v$ is therefore
        $\var{v_r} (\var{v_e})$. In order to conclude that
        $\psem{s}{\rho}{\var{ref}[e]} (\var{v_r} (\var{v_e}), \kw{U})$ 
        with~\rref{E-RefMapping},
         we need to have that
        $\var{v_e} \in \kw{meta}(\beta)$, which follows by inversion from $\Sigma \vdash \var{v}_e :_s \beta$.

    \end{proofcases}
    We proceed with with the cases of  \Cref{lem:expr-typesafety-untimed}.
    \begin{proofcases}
        \case{$\Sigma; I; \Phi \vdash_{\maybe{\var{A}}, \kw{U}} n : \iota$}
        By inversion of the typing rule \rref{T-Int} we know that $n \in \mathbb{Z}$ and
        $(\meta{min}(\iota) \leq n \leq \meta{max}(\iota)) ~\vee~ \iota = \intt{}$.
        Hence, applying \rref{E-Int}, yields $\psem{s}{\rho}{n} n$. Further, since we have $n \in \mathbb{Z}$ and
        the bounds condition $(\meta{min}(\iota) \leq n \leq \meta{max}(\iota)) ~\vee~ \iota = \intt{}$,
        by \rref{V-Int} we get $\vdash n : \iota$.

        \case{$\Sigma; I; \Phi \vdash_{\maybe{\var{A}}, \kw{U}} b : \kw{bool}$}
        We invert the typing rule \rref{T-Bool} and get  $b \in \mathbb{B}$.
        Then, we apply \rref{E-Bool}, to reach $\psem{s}{\rho}{b} b$. Further, from \rref{V-Bool}, we get
        $\vdash b : \kw{bool}$.

        \case{$\Sigma; I; \Phi \vdash_{\maybe{\var{A}}, \kw{U}} \var{ref_{\beta}} : \beta$}
        By inversion of \rref{T-Ref} we get that $\Sigma; I \vdash^k_{\maybe{\var{A}}, \kw{U}} \var{ref} : \beta$.
        From the mutual induction hypothesis we know that there is a value $v$ such that $\psem{s}{\rho}{\var{ref}} (v,\kw{U})$
        and $\Sigma \vdash v :_s \beta$. Applying now \rref{E-Ref}, we conclude $\psem{s}{\rho}{\var{ref}} v$ and from inversion
         of \rref{V-BaseValAlpha} follows $\vdash v : \beta$.

        \case{$\Sigma; I; \Phi \vdash_{\maybe{\var{A}}, \kw{U}} \addr{\var{ref_{Id}}} : \kw{address}$}
        We invert \rref{T-Addr} to get $\Sigma; I \vdash^k_{\maybe{\var{A}}, \kw{U}} \var{ref_{Id}} : Id$.
        By the mutual induction hypothesis, there exists $v$ such that $\psem{s}{\rho}{\var{ref_{Id}}} (v,\kw{U})$ and
        $\Sigma \vdash v :_{s} \var{Id}$.
        Applying \rref{E-Addr} to $\psem{s}{\rho}{\var{ref_{Id}}} (v,\kw{U})$, we get $\psem{s}{\rho}{\addr{\var{ref_{Id}}}} v$.
        Further, inverting $\Sigma \vdash v :_{s} \var{Id}$ twice (through 
        \rref{V-Contract} and~\rref{V-AddrIsContract}),
        we know that $v \in \dom{s}$ and hence $v \in \meta{Addr}$. Therefore, we can use rule \rref{V-Addr} to conclude
        $\vdash v : \kw{address}$.

        \case{$\Sigma; I; \Phi \vdash_{\maybe{\var{A}}, \kw{U}} \kw{inrange}(\iota_1, e) : \kw{bool}$}
        We invert \rref{T-Range} to get $\Sigma; I; \Phi \vdash_{\maybe{\var{A}}, \kw{U}} e : \iota_2$.
        We then apply the induction hypothesis to obtain that there exists a $v$ such that $\psem{s}{\rho}{e} v$ and that
        $\vdash v : \iota_2$. Inverting now rule \rref{V-Int}, it follows that $v \in \mathbb{Z}$
        and that $\meta{min}(\iota_2) \leq v < \meta{max}(\iota_2) \lor \iota_2 = \kw{int}$.
        If $v \in [\meta{min}(\iota_1),\meta{max}(\iota_1)]$ or $\iota_1 = \kw{int}$, 
        then \rref{E-RangeTrue} applies and thus $\psem{s}{\rho}{\kw{inrange}(e,\iota_1)} \meta{True}$,
        otherwise \rref{E-RangeFalse} applies and thus $\psem{s}{\rho}{\kw{inrange}(e,\iota_1)} \meta{False}$.
        Either way the resulting value $b$ is an element of $\mathbb{B}$. Therefore, by \rref{V-Bool},   $\vdash b : \kw{bool}$.

        \case{$\Sigma; I; \Phi \vdash_{\maybe{\var{A}}, \kw{U}} e_1~\opi~e_2 : \iota$}
        By inversion of \rref{T-BopI}, both $\Sigma; I; \Phi \vdash_{\maybe{\var{A}}, \kw{U}} e_1 : \iota_1$ and $\Sigma; I; \Phi \vdash_{\maybe{\var{A}}, \kw{U}} e_2 : \iota_2$ hold,
        and moreover $\Sigma; I; \Phi \vDash_{\maybe{A}} \inrange{\iota}{e_1~\opi~e_2}$.
        With our induction hypothesis we get there are values $v_1,v_2$ such that $\psem{s}{\rho}{e_1} v_1$, $\psem{s}{\rho}{e_2} v_2$,
        $\vdash v_1 : \iota_1$, and $\vdash v_2 : \iota_2$. Hence, by inversion of \rref{V-Int} it follows that $v_1,v_2 \in \mathbb{Z}$.

        If $v_2 \neq 0$ or $\opi~\notin\{\kw{div},\kw{mod}\}$, then \rref{E-Div}, \rref{E-Mod}, respectively \rref{E-BopI} applies
        and we conclude  $\psem{s}{\rho}{e_1~\opi~e_2} v_1~\opi~v_2$. Moreover, since $v_1,v_2 \in \mathbb{Z}$ 
        and from the entailment $\Sigma; I; \Phi \vDash_{\maybe{A}} \inrange{\iota}{e_1~\opi~e_2}$
        together with~\ref{typsf:uexp:phi} we get
        $(\meta{min}(\iota) \leq v_1~\opi~v_2 \leq \meta{max}(\iota)) ~\vee~ \iota = \intt{}$.
        Applying \rref{V-Int} yields $\vdash v_1~\opi~v_2 : \iota$.

        Otherwise, if $v_2 = 0$ and $\opi~\in\{\kw{div},\kw{mod}\}$, either \rref{E-DivZero} or \rref{E-ModZero} apply and
        allow us to infer  $\psem{s}{\rho}{e_1~\opi~e_2} 0$. Since $0 \in \mathbb{Z}$ and $0 \in [\meta{min}(\iota), \meta{max}(\iota)]$
        for all bounded integer types, by \rref{V-Int} we get $\vdash 0 : \iota$.

        \case{$\Sigma; I; \Phi \vdash_{\maybe{\var{A}}, \kw{U}} e : \iota_2$ (\rref{T-NumConv})}
        By inversion of \rref{T-NumConv}, we have $\Sigma; I; \Phi \vdash_{\maybe{\var{A}}, \kw{U}} e : \iota_1$
        and $\Sigma; I; \Phi \vDash_{\maybe{A}} \inrange{\iota_2}{e}$.
        By the induction hypothesis there exists $v$ such that $\psem{s}{\rho}{e} v$ and $\vdash v : \iota_1$.
        By inversion of \rref{V-Int} it follows that $v \in \mathbb{Z}$.
        From the semantic entailment together with~\ref{typsf:uexp:phi} we obtain
        $(\meta{min}(\iota_2) \leq v \leq \meta{max}(\iota_2)) ~\vee~ \iota_2 = \intt{}$.
        Applying \rref{V-Int} yields $\vdash v : \iota_2$.

        \case{$\Sigma; I; \Phi \vdash_{\maybe{\var{A}}, \kw{U}} e_1~\opb~e_2 : \kw{bool}$}
        Inverting \rref{T-BopB} yields $\Sigma; I; \Phi \vdash_{\maybe{\var{A}}, \kw{U}} e_1 : \kw{bool}$ and $\Sigma; I; \Phi \vdash_{\maybe{\var{A}}, \kw{U}} e_2 : \kw{bool}$.
        We then apply the induction hypothesis to infer the existence of $v_1, v_2$ such that  $\psem{s}{\rho}{e_1} v_1$, $\quad\psem{s}{\rho}{e_2} v_2$,
        $\quad \vdash v_1 : \kw{bool}$ and  $\vdash v_2 : \kw{bool}$. From \rref{V-Bool} we know that therefore $v_1,v_2 \in \mathbb{B}$.
        Hence, we can use \rref{E-BopB} to conclude  $\psem{s}{\rho}{e_1~\opb~e_2} v_1~\opb~v_2$. Further, since $ v_1~\opb~v_2 \in \mathbb{B}$
        we use \rref{V-Bool} to reach  $\vdash v_1~\opb~v_2 : \kw{bool}$.

        \case{$\Sigma; I; \Phi \vdash_{\maybe{\var{A}}, \kw{U}} \neg e : \kw{bool}$}
        We invert \rref{T-Neg} and infer $\Sigma; I; \Phi \vdash_{\maybe{\var{A}}, \kw{U}} e : \kw{bool}$. Then, we apply the induction
        hypothesis to reason that there exists a value $v$ such that $\psem{s}{\rho}{e} v$ and that $\vdash v : \kw{bool}$.
        Inverting \rref{V-Bool}, shows $v \in \mathbb{B}$. Using \rref{E-Neg}, we know $\psem{s}{\rho}{\neg e} \neg v$.
        Since  $v \in \mathbb{B}$, also $\neg v \in \mathbb{B}$, thus we can finally
        apply \rref{V-Bool} to conclude  $\vdash \neg v : \kw{bool}$.

        \case{$\Sigma; I; \Phi \vdash_{\maybe{\var{A}}, \kw{U}} e_1~\cmp~e_2 : \kw{bool}$}
        Inverting \rref{T-Cmp} yields $\Sigma; I; \Phi \vdash_{\maybe{\var{A}}, \kw{U}} e_1 : \iota$ and $\Sigma; I; \Phi \vdash_{\maybe{\var{A}}, \kw{U}} e_2 : \iota$.
        We then apply the induction hypothesis to infer the existence of $v_1, v_2$ such that  $\psem{s}{\rho}{e_1} v_1$, $\quad\psem{s}{\rho}{e_2} v_2$,
        $\quad \vdash v_1 : \iota$ and  $\vdash v_2 : \iota$. From \rref{V-Int} we know that therefore $v_1,v_2 \in \mathbb{Z}$.
        Hence, we can use \rref{E-Cmp} to conclude  $\psem{s}{\rho}{e_1~\cmp~e_2} v_1~\cmp~v_2$. Further, since $ v_1~\cmp~v_2 \in \mathbb{B}$
        we use \rref{V-Bool} to reach  $\vdash v_1~\cmp~v_2 : \kw{bool}$.

        \case{$\Sigma; I; \Phi \vdash_{\maybe{\var{A}}, \kw{U}} \ite{e_1}{e_2}{e_3} : \beta$}
        Inversion of \rref{T-ITE} yields $\Sigma; I; \Phi \vdash_{\maybe{\var{A}}, \kw{U}} e_1 : \kw{bool}$ and $\Sigma; I; \Phi \vdash_{\maybe{\var{A}}, \kw{U}} e_2 : \beta$
        and $\Sigma; I; \Phi \vdash_{\maybe{\var{A}}, \kw{U}} e_3 : \beta$. We then use the induction hypothesis to infer that there are $v_1, v_2, v_3$
        such that  $\psem{s}{\rho}{e_1} v_1$ and $\;\psem{s}{\rho}{e_2} v_2$ and $\;\psem{s}{\rho}{e_3} v_3$ and
        $\;\vdash v_1 : \kw{bool}$ and $\;\vdash v_2 : \beta$ and $\vdash v_3 : \beta$.

        From inverting \rref{V-Bool} we know $v_1 \in \mathbb{B}$.
        If $v_1 = \meta{True}$, then \rref{E-ITETrue} applies and yields
        $\psem{s}{\rho}{\ite{e_1}{e_2}{e_3}} v_2$.
        We know $\vdash v_2 : \beta$ already from the hypothesis.
        Similar for $v_1=\meta{False}$ and \rref{E-ITEFalse}.

        \case{$\Sigma; I; \Phi \vdash_{\maybe{\var{A}}, \kw{U}} e_1~=~e_2 : \kw{bool}$}
        We invert \rref{T-Eq} and get $\Sigma; I; \Phi \vdash_{\maybe{\var{A}}, \kw{U}} e_1 : \beta$ and
        $\Sigma; I; \Phi \vdash_{\maybe{\var{A}}, \kw{U}} e_2 : \beta$ From the induction hypothesis follows that
        there are $v_1, v_2$ such that $\psem{s}{\rho}{e_1} v_1$ and $\;\psem{s}{\rho}{e_2} v_2$ and $\;\vdash v_1 : \beta$
        and $\vdash v_2 : \beta$. Inverting \rref{V-Int},\rref{V-Bool} or \rref{V-Addr} (depending on $\beta$),
        we know that $v_1$ and $v_2$ are comparable. If $v_1 = v_2$, then the evaluation rule \rref{E-EqTrue} applies and yields
        $\psem{s}{\rho}{e_1=e_2} \meta{True}$. Since $\meta{True} \in \mathbb{B}$, with \rref{V-Bool},
        we conclude  $\vdash \meta{True} : \kw{bool}$. Similar for $v_1 \neq v_2$ and \rref{E-EqFalse}.

    \end{proofcases}

    \end{proof}

\begin{lemma}[References Type Safety (Timed)]
\label{lem:ref-typesafety-timed}
    Assume that
    \begin{enumerate}[start=1,label={(\bfseries H\arabic*)}]
        \item \label{typsf:tref:wt} $\Sigma; I \vdash_{\maybe{A}, \kw{T}}^k \var{ref} : \sigma$
        \item \label{typsf:tref:env} $\Sigma \vdash \rho :_{\spre} I$ and $\Sigma \vdash \rho :_{\spost} I$
        \item \label{typsf:tref:loc} $\Sigma \vdash \ell :_{\spre} \maybe{A}$ and $\Sigma \vdash \ell :_{\spost} \maybe{A}$
    \end{enumerate}
    Then there exists $v$ such that $\psem{(\spre,\spost)^{\kw{T}}}{\rho}{\var{ref}} (v, t_p)$ and
    $t_p = \kw{pre} \Rightarrow \Sigma \vdash v :_{\spre} \sigma$ and
    $t_p = \kw{post} \Rightarrow \Sigma \vdash v :_{\spost} \sigma$.
\end{lemma}

    \begin{lemma}[Expression Type Safety (Timed)]
        \label{lem:expr-typesafety-timed}
    Assume that
    \begin{enumerate}[start=1,label={(\bfseries H\arabic*)}]
        \item \label{typsf:texp:wt} $\Sigma; I; \Phi \vdash_{\maybe{A}, \kw{T}} e : \beta$
        \item \label{typsf:texp:env} $\Sigma \vdash \rho :_{\spre} I$ and $\Sigma \vdash \rho :_{\spost} I$
        \item \label{typsf:texp:loc} $\Sigma \vdash \ell :_{\spre} \maybe{A}$ and $\Sigma \vdash \ell :_{\spost} \maybe{A}$
    \end{enumerate}
    Then there exists $v$ such that $\psem{(\spre,\spost)^{\kw{T}}}{\rho}{e} v$ and $\vdash v : \beta$.
    \end{lemma}

    \begin{proof}[Proof of~\Cref{lem:ref-typesafety-timed,lem:expr-typesafety-timed}]
    By mutual induction on the typing derivation~\ref{typsf:tref:wt}.
    \begin{proofcases}
        \case{$\Sigma; I \vdash^{\kw{S}}_{\var{A}, \kw{T}} \pre{x} : \sigma$}
        The typing derivation ends with the~\rref{T-StoragePre} rule, which gives us $(x : \sigma) \in \Sigmastore(A)$ and $x \notin \dom{I}$.
        From inversion of derivation of (H3) for $\spre$ in~\Cref{lem:ref-typesafety-timed} via the~\rref{V-Contract} rule we
        then deduce that $x \in \dom{\spre(\ell)}$ and from (H2) we get that $x \notin \dom{\rho}$ so the~\rref{E-StoragePre} rule can be applied to get the candidate
        $v = \spre(\ell)(x)$ with tag $\kw{pre}$. We need to derive $\Sigma \vdash \spre(\ell)(x) :_{\spre} \sigma$. Again using the inversion on the first
        judgement of (H3) of~\Cref{lem:ref-typesafety-timed}, the premises of the derivation give us that since
         $(x : \sigma) \in \Sigmastore(\var{A})$, we also have $\Sigma \vdash \spre(\ell)(x) :_{\spre} \sigma$, which concludes the case.

        \case{$\Sigma; I \vdash^{\kw{S}}_{\var{A}, \kw{T}} \post{x} : \sigma$}
        The typing derivation ends with the~\rref{T-StoragePost} rule, which gives us $(x : \sigma) \in \Sigmastore(A)$ and $x \notin \dom{I}$.
        From inversion of derivation of (H3) for $\spost$ in~\Cref{lem:ref-typesafety-timed} via the~\rref{V-Contract}
        rule we then deduce that $x \in \dom{\spost(\ell)}$, and from (H2) we get that $x \notin \dom{\rho}$, so the~\rref{E-StoragePost} rule can be applied to get the
        candidate $v = \spost(\ell)(x)$ with tag $\kw{post}$. We need to derive $\Sigma \vdash \spost(\ell)(x) :_{\spost} \sigma$. Again using the inversion on
        the second judgement of (H3) of~\Cref{lem:ref-typesafety-timed}, the premises of the derivation give us that since
         $(x : \sigma) \in \Sigmastore(\var{A})$, we also have $\Sigma \vdash \spost(\ell)(x) :_{\spost} \sigma$, which concludes the case.

        \case{$\Sigma; I \vdash^{\kw{N}}_{\maybe{\var{A}}, \kw{T}} x : \alpha$}
        The proof proceeds similarly to the untimed case in the proof of~\Cref{lem:ref-typesafety-untimed}, taking the first judgement in (H2) and applying the~\rref{E-CalldataTimed} rule.

        \case{$\Sigma; I \vdash^{k}_{\maybe{\var{A}}, \kw{T}} \var{ref}.x : \sigma$}
        The proof proceeds similarly to the untimed case in the proof of~\Cref{lem:ref-typesafety-untimed}, but depending on the result of the induction call giving us $\psem{(\spre, \spost)^{\kw{T}}}{\rho}{\var{ref}} (\ell', t_p)$ we proceed with~\rref{E-FieldPre} and the state $\spre$ if $t_p = \kw{pre}$; or with~\rref{E-FieldPost} and the state $\spost$ if $tp=\kw{post}$.

        \case{$\Sigma; I \vdash^{k}_{\maybe{\var{A}}, \kw{T}} \var{ref}[e] : \mu$}
        The proof proceeds similarly to the untimed case in the proof of~\Cref{lem:ref-typesafety-untimed}, but depending on the timing of the induction call on the premise, invoking either the first premise of (H3) for the $\kw{pre}$-case or the second premise of (H3) for the $\kw{post}$-case.

        \case{Else:}
        All other cases proceed in the exact same way as in the proof of~\Cref{lem:ref-typesafety-untimed,lem:expr-typesafety-untimed}.
        Not that timed expressions use only pure mathematical integer and will never semantically check the bounds of operations.
        Hence, the assumption that $\Phi$ holds is not necessary in these cases.
    \end{proofcases}
    \end{proof}

    \begin{lemma}[Mapping Expression Type Safety]
        \label{lem:mapexpr-typesafety}
            Assume that
            \begin{itemize}
                \item[(H1)] $\Sigma; I; \Phi \vdash_{\maybe{A}} m : \mu$
                \item[(H2)] $\Sigma \vdash \rho :_s I$
                \item[(H3)] $\Sigma \vdash \ell :_s \maybe{A}$
                \item[(H4)] $\psem{s}{\rho}{\Phi} \meta{True}$
            \end{itemize}
            Then there exists $v$ such that $\psem{s}{\rho}{m} v$ and $\vdash v : \mu$.
    \end{lemma}
    \begin{proof}
        By induction on the typing derivation (H1).
        \begin{proofcases}
            \case{$\Sigma; I; \Phi \vdash_{\maybe{A}} e : \beta$}
            We invert \rref{T-Exp} to get $\Sigma; I; \Phi \vdash_{\maybe{A},\kw{U}} e : \beta$. From
            \Cref{lem:expr-typesafety-untimed} follows that there is a $v$ such that $\psem{s}{\rho}{e} v$
            and $\vdash v : \beta$.

            \case{$\Sigma; I; \Phi \vdash_{\maybe{A}} [ \{ e_i \Rightarrow m_i\}_{i \in [1,n]}]_{\annot{\mapping{\beta}{\mu}}}  : \mapping{\beta}{\mu}$}
            From the first premise of \rref{T-Mapping} we obtain for all $i \in [1,n]$ that $\Sigma; I; \Phi \vdash_{\maybe{A},\kw{U}} e_i : \beta$.
            We can apply \Cref{lem:expr-typesafety-untimed}, to get $\psem{s}{\rho}{e_i} v_i$.
            From the second premise of \rref{T-Mapping} and the induction hypothesis we know  $\psem{s}{\rho}{m_i} u_i$ and $\vdash u_i : \mu$.
            By defining
            \[f = (x \in \meta{meta}(\beta) \mapsto \text{ if } v_i \in \meta{meta}(\beta) \land v_i = x \text{ then } u_i \text{ else } \meta{default}(\mu))\;,\]
            we can apply \rref{E-Mapping} to conclude $\psem{s}{\rho}{ [ \{ e_i \Rightarrow m_i\}_{i \in [1,n]}]} f$.

            We still have to show that $\vdash f : \mapping{\beta}{\mu}$.
            To do so, we need to show that $f \in \meta{meta}(\beta)\mapsto\meta{Value}$ and
            $\forall v.\; \vdash v : \beta \Rightarrow  \; \vdash f(v) : \mu$.
            The domain of $f$ is $\meta{meta}(\beta)$
            by construction. Let now $v \in \meta{meta}(\beta)$ (equivalent to $\vdash v : \beta$). From \Cref{lem:expr-typesafety-untimed}
            and inversion of \rref{V-BaseValAlpha}, we know that
            $ \vdash v_i : \beta$ and with \rref{V-Addr}, \rref{V-Int} and \rref{V-Bool} follows $v_i \in \meta{meta}(\beta)$.
            Hence, if $v = v_i$ then $f(v) = u_i$. From above, we know $ \vdash u_i : \mu$ and
            with \Cref{lem:value} we obtain $f(v) \in \meta{Value}$. Otherwise, if $v \neq v_i$ for all $i$, then $f(v) = \meta{default}(\mu)$.
            From \Cref{lem:defaultmu}, we get $\vdash f(v) : \mu$ and from \Cref{lem:value} follows $f(v) \in \meta{Value}$.
            Finally, we can apply \rref{V-Mapping} to conclude $\vdash f: \mapping{\beta}{\mu}$.

            \case{$\Sigma; I; \Phi \vdash_{\maybe{A}} \var{ref}\map{\range{e_i \Rightarrow \var{m}_i}{i \in [1,n]}}{\annot{\mapping{\beta}{\mu}}}  : \mapping{\beta}{\mu}$}
            From the first premise of \rref{T-MappingUpd} we obtain $\Sigma; I; \Phi \vdash_{\maybe{A},\kw{U}} \var{ref} : \mapping{\beta}{\mu}$.
            From the induction hypothesis we know $\psem{s}{\rho}{\var{ref}} f$ and $\vdash f : \mapping{\beta}{\mu}$.
            From the second premise of \rref{T-MappingUpd} we obtain for all $i \in [1,n]$ that $\Sigma; I; \Phi \vdash_{\maybe{A},\kw{U}} e_i : \beta$.
            We can apply \Cref{lem:expr-typesafety-untimed}, to get $\psem{s}{\rho}{e_i} v_i$   and $\vdash v_i : \beta$.
            From the third premise of \rref{T-MappingUpd} and the induction hypothesis we know $\psem{s}{\rho}{m_i} u_i$ and $\vdash u_i : \mu$.
            By defining
            \[g = (x \in \meta{meta}(\beta) \mapsto \text{ if } v_i = x \text{ then } u_i \text{ else } f(x))\;,\]
            and (depending on $\beta$) inverting \rref{V-Addr}, \rref{V-Bool} or \rref{V-Int} to obtain $v_i \in \meta{meta}(\beta)$, we can apply \rref{E-MappingUpd} to conclude $\psem{s}{\rho}{\var{ref}\map{\range{e_i \Rightarrow \var{m}_i}{i \in [1,n]}}{\annot{\mapping{\beta}{\mu}}}} g$.

            We still have to show that $\vdash g : \mapping{\beta}{\mu}$.
            To do so, we need to show that $g \in \meta{meta}(\beta)\mapsto\meta{Value}$ and
            $\forall v.\; \vdash v : \beta \Rightarrow  \; \vdash g(v) : \mu$.
            The domain of $g$ is $\meta{meta}(\beta)$ by construction. Let now $v \in \meta{meta}(\beta)$ (equivalent to $\vdash v: \beta $), then $g(v)$ 
            is either one of the $u_i$, for which we know  $ \vdash u_i : \mu$ and
            with \Cref{lem:value} we also obtain $u_i \in \meta{Value}$; or $g(v)=f(v)$ for which we also know by 
            inversion of \rref{V-Mapping} that 
            $ \vdash f(v) : \mu$ and hence  $f(v) \in \meta{Value}$. Finally, we can apply \rref{V-Mapping} to conclude $\vdash g: \mapping{\beta}{\mu}$.

            % From \Cref{lem:expr-typesafety-untimed} and inversion of \rref{V-BaseValAlpha}, we know that
            % $ \vdash v_i : \beta$ and with \rref{V-Addr}, \rref{V-Int} and \rref{V-Bool} follows $v_i \in \meta{meta}(\beta)$.
            % Hence, if $v = v_i$ then $g(v) = u_i$. From above, we know $ \vdash u_i : \mu$ and
            % with \Cref{lem:value} we obtain $g(v) \in \meta{Value}$. Otherwise, if $v \neq v_i$ for all $i$, then $g(v) = f'(v)$.
            % Since $\vdash f' : \mapping{\beta}{\mu}$, by inversion of \rref{V-Mapping} we get $\vdash f'(v) : \mu$ and $f'(v) \in \meta{Value}$.
            % Finally, we can apply \rref{V-Mapping} to conclude $\vdash g: \mapping{\beta}{\mu}$.
        \end{proofcases}
    \end{proof}

    \begin{lemma}[Slot Expression Type Safety]\label{lem:se-typesafety}
Assume that
    \begin{enumerate}[start=1,label={(\bfseries H\arabic*)}]
        \item\label{typsf:se:wt} $\Sigma; I; \Phi \vdash_{\maybe{A}} \var{se} : \sigma$
        \item\label{typsf:se:env} $\Sigma \vdash \rho :_s I$
        \item\label{typsf:se:loc} $\Sigma \vdash \ell :_s \maybe{A}$
        \item\label{typsf:se:pre} $\psem{s}{\rho}{\Phi} \meta{True}$
    \end{enumerate}
Then there exist $v$ and $s'$ such that
    \begin{itemize}
        \item $\psem{s}{\rho}{\var{se}} (v, s')$,
        \item $\Sigma \vdash v :_{s'} \sigma$,
        \item $\Sigma \vdash \rho :_{s'} I$,
        \item $\Sigma \vdash \ell :_{s'} \maybe{A}$,
        \item $s \subseteq s'$,
        \item $\forall \ell' \in \dom{s'} \setminus \dom{s}.~\exists B.~\Sigma \vdash \ell' :_{s'} B$.
        \item $\forall u, \sigma'.~\Sigma \vdash u :_s \sigma' \implies \Sigma \vdash u :_{s'} \sigma'$
    \end{itemize}
    \end{lemma}

    \begin{proof}
        By induction on the derivation of~\ref{typsf:se:wt} with mutual induction on the proof of~\Cref{lem:create-typesafety}.
        \begin{proofcases}
            \case{$\Sigma; I; \Phi \vdash_{\maybe{\var{A}}} m : \mu$ (rule~\rref{T-MapExp})}
            From the premises of the~\rref{T-MapExp} rule we obtain $\Sigma; I; \Phi \vdash_{\maybe{\var{A}}} m : \mu$, on which (together with~\ref{typsf:se:env}, \ref{typsf:se:loc}, and \ref{typsf:se:pre}) we apply~\Cref{lem:mapexpr-typesafety} to obtain $\psem{s}{\rho}{m} v$ and $\vdash v : \mu$. Setting $s'=s$, we derive $\psem{s}{\rho}{m} (v, s)$ through~\rref{E-MapExp}, we get $\Sigma \vdash v :_{s} \mu$ from~\rref{V-MappingVal} and the rest trivially follows from hypothesis and $s'=s$.

            \case{$\Sigma; I; \Phi \vdash_{\maybe{\var{A}}} \var{ref_{\annot{\var{B}}}} : \var{B}$ (rule~\rref{T-SlotRef})}
            From the premises we get that $\Sigma; I \vdash^{k}_{\maybe{\var{A}}, \kw{U}} \var{ref} : \var{B}$, and we can apply~\Cref{lem:ref-typesafety-untimed} to get a $v$ such that $\psem{s^{\kw{U}}}{\rho}{\var{ref}} (v, \kw{U})$ and $\Sigma \vdash v :_s B$. The rule~\rref{E-SlotRef} gives us $s'$ to be $s$ and $\psem{s}{\rho}{\var{ref}} (v, s)$ and we get the rest of the conditions from the premises (as the state $s'$ remains the same $s$).

            \case{$\Sigma; I; \Phi \vdash_{\maybe{\var{A}}} \addr{\var{se}} : \kw{address}_B$ (rule~\rref{T-SlotAddr})}
            The induction hypothesis is applied to the rule~\rref{T-SlotAddr} and we get the $(v, s')$ such that
            \begin{enumerate}[start=1,label={(\bfseries I\arabic*)}]
                \item $\psem{s}{\rho}{\var{se}} (v, s')$
                \item \label{typsf:se:I2} $\Sigma \vdash v :_{s'} B$
                \item $\Sigma \vdash \rho :_{s'} I$
                \item $\Sigma \vdash \ell :_{s'} \maybe{A}$
                \item $s \subseteq s'$
                \item $\forall \ell' \in \dom{s'} \setminus \dom{s}.~\exists C.~\Sigma \vdash \ell' :_{s'} C$
                \item $\forall u, \sigma'.~\Sigma \vdash u :_s \sigma' \implies \Sigma \vdash u :_{s'} \sigma'$
            \end{enumerate}
            We then apply the rule~\rref{E-SlotAddr} to obtain a derivation of
            $\psem{s}{\rho}{\addr{\var{se}}} (v, s')$ and we invert~\ref{typsf:se:I2} via the~\rref{V-Contract} rule, to get $\Sigma \vdash v :_{s'} \kw{address}_B$.

            \case{$\Sigma; I; \Phi \vdash_{\maybe{\var{A}}} \kw{new}~B(\var{se_1}, \dots, \var{se_n}) : B$ (rule~\rref{T-Create})}
            From the premises of the~\rref{T-Create} rule we get
            \begin{enumerate}[start=1,label={(\bfseries P\arabic*)}]
                \item\label{typsf:se:create:P1} $\Sigmacode(B) = \var{ctor} = \kw{constructor}~\var{I'} ~ \kw{iff}~\many{\var{pre}}~
                \caseblock{e_{i}}{\kw{creates}~\many{\var{create}_i}}~\kw{ensures}~\many{post}$
                \item\label{typsf:se:create:P2} $I' = \range{x_i : \alpha_i}{i \in [1,n]}$
                \item\label{typsf:se:create:P3} $\forall \, j \in [1,n]. ~\, \Sigma; I; \Phi \vdash_{\maybe{\var{A}},\kw{U}} \var{se_j} : \alpha_j$
                \item\label{typsf:se:create:P4} $\Sigma; I;  \Phi \vDash_{\maybe{\var{A}}} (\range{se_i}{i \in [1,n]}, \many{\var{pre}})$.
            \end{enumerate}
            By applying induction hypothesis on~\ref{typsf:se:create:P3}, we get for every $i \in [1,n]$
            \begin{enumerate}[start=1,label={(\bfseries I$_i$\arabic*)}]
                \item\label{typsf:se:create:I1} $\psem{s_{i-1}}{\rho}{se_i} (v_i, s_i)$,
                \item\label{typsf:se:create:I2} $\Sigma \vdash v_i :_{s_i} \alpha_i$,
                \item\label{typsf:se:create:I3} $\Sigma \vdash \rho :_{s_i} I$,
                \item\label{typsf:se:create:I4} $\Sigma \vdash \ell :_{s_i} \maybe{A}$,
                \item\label{typsf:se:create:I5} $s_{i-1} \subseteq s_i$,
                \item\label{typsf:se:create:I6} $\forall \ell' \in \dom{s_{i}} \setminus \dom{s_{i-1}}.~\exists D.~\Sigma \vdash \ell' :_{s_{i}} D$,
                \item\label{typsf:se:create:I7} $\forall u, \sigma'.~\Sigma \vdash u :_{s_{i-1}} \sigma' \implies \Sigma \vdash u :_{s_i} \sigma'$
            \end{enumerate}
            where we thread the state $s_i$, with $s_0 = s$, and we obtain $\psem{s_{i-1}}{\rho}{\Phi} \meta{True}$ through~\Cref{lem:sem-storage-weak} from~\ref{typsf:se:wt} and $(I_j5)$ for $j \in [1, i-1]$.
            We define $\rho' = \range{x_i \mapsto \var{v}_i}{i \in [1,n]}  \cup \{\kw{caller} \mapsto \ell, \kw{origin} \mapsto \rho(\kw{origin}), \kw{callvalue} \mapsto 0 \}$ and in order to apply~\rref{E-Create} rule we need to derive $s_n ; \rho' ; \var{ctor}_{\kw{cases}} \bigstep (\ell', s')$. To derive this, we apply~\rref{E-CtorCases}, by deriving
            \begin{enumerate}[start=1,label={(\bfseries D\arabic*)}]
                \item\label{typsf:se:create:D1} $s_n^{\kw{U}} ; \rho' ; e_j \bigstep_{\dummyloc} \meta{True}$
                \item\label{typsf:se:create:D2} $s_n^{\kw{U}} ; \rho' ; e_i \bigstep_{\dummyloc} \meta{False} \text{ for } i \neq j$
                \item\label{typsf:se:create:D3} $s_n^{\kw{U}} ; \rho' ; \many{\var{create}_j} \Downarrow (\ell',s')$
            \end{enumerate}
            as follows.
            Since $\Sigma$ is well-typed, let $C, \Sigma' \subset \Sigma$ be such that $\Sigma' \vdash \var{ctor} : C$. By inversion we know the only possible rule to be applied is~\rref{T-Ctor} and from the premises we obtain that \inlineequation[typsf:se:create:envpwf]{\Sigma' \vdash I'~\meta{wf}} and 
            \begin{equation*}
                \Sigma'; I'; \kw{true} \vDash_{\bot}
                \bigwedge_{\var{pre \in \many{\var{pre}}}}{\var{pre}} \Rightarrow (\bigvee_{i \in [1,n]} e_i) ~~ \wedge ~~(\bigwedge_{i \in [1,n], j \in [i+1, n]} \neg (e_i \wedge e_j))
            \end{equation*}
            which by inversion through the~\rref{ValidExps} means
            \begin{equation}
                \label{eq:typesafety-se-validexps}
                \begin{aligned}
                & \forall s'', \rho'', \ell''.~\rho'' :_{s''} I' \wedge \Sigma' \vdash \ell'' :_{s''} \bot \implies \\
                & s''; \rho''; \bigwedge_{\var{pre \in \many{\var{pre}}}}{\var{pre}} \Rightarrow (\bigvee_{i \in [1,n]} e_i) ~~ \wedge ~~(\bigwedge_{i \in [1,n], j \in [i+1, n]} \neg (e_i \wedge e_j)) \bigstep_{\ell''} \meta{True}
                \end{aligned}
            \end{equation}
            We take $s'' = s_n$ and $\rho'' = \rho'$ and any location $\ell'' = \dummyloc$. We get $\Sigma' \vdash \rho' :_{s_n} I'$ by first deriving $\Sigma \vdash \rho' :_{s_n} I'$ using the~\rref{V-Env} rule getting
            \begin{itemize}
                \item $\vdash \rho'(\kw{origin}) : \kw{address}$ from $\rho'(\kw{origin}) = \rho(\kw{origin})$ and inverting~\ref{typsf:se:env},
                \item $\vdash \rho'(\kw{caller}) : \kw{address}$ by deriving $\vdash \ell : \kw{address}$ via~\rref{V-Addr},
                \item $\vdash 0 : \kw{uint256}$ from~\rref{V-Int},
                \item $\Sigma \vdash v_i :_{s_n} \alpha_i$ by repeatedly applying~\Cref{lem:valuetyp-storage-weak} on $(I_i2)$ and $(I_k5)$ for all $k \in [i,n]$
            \end{itemize}
            We then apply~\Cref{lem:valuetyp-styp-strength} on $\Sigma \vdash \rho' :_{s_n} I'$ and $\Sigma' \subset \Sigma$ and~\eqref{typsf:se:create:envpwf} to get $\Sigma' \vdash \rho' :_{s_n} I'$.
            With this data we get from~\eqref{eq:typesafety-se-validexps} that
            \begin{equation}
                \label{eq:se-typesafety-precond}
                s_n ; \rho'; \bigwedge_{\var{pre \in \many{\var{pre}}}}{\var{pre}} \Rightarrow (\bigvee_{i \in [1,n]} e_i) ~~ \wedge ~~(\bigwedge_{i \in [1,n], j \in [i+1, n]} \neg (e_i \wedge e_j)) \bigstep_{\dummyloc} \meta{True}
            \end{equation}
            From~\ref{typsf:se:create:P4} by inversion (rule~\rref{ValidIffs}) we get from $\Sigma \vdash \rho :_{s_0} I$ and $\Sigma \vdash \ell :_{s_0} \maybe{A}$ and $\psem{s_0}{\rho}{\Phi} \meta{True}$ and  $s_{i-1} ; \rho ; se_i \bigstep_{\ell} (v_i, s_i)$ that
            \begin{equation}
                \many{s_n ; \rho' ; \var{pre} \bigstep_{\dummyloc} \meta{True}} \label{eq:se-typesafety-precs}
            \end{equation}
            From inverting the derivation of~\eqref{eq:se-typesafety-precond} using the~\rref{E-BopB} rule for implication we know that since an implication evaluates to $\meta{True}$ and the the lefthand-side also evaluates to $\meta{True}$ due to~\eqref{eq:se-typesafety-precs}, then the righthand-side evaluates to $\meta{True}$. With similar meta-level reasoning for conjunctions we obtain that
            \begin{align*}
                s_n ; \rho' ; \bigvee_{i \in [1,n]} e_i \bigstep_{\dummyloc} \meta{True} \\
                s_n ; \rho' ;\bigwedge_{i \in [1,n], j \in [i+1, n]} \neg (e_i \wedge e_j) \bigstep_{\dummyloc} \meta{True}
            \end{align*}
            These judgements encode that at least one and at most one $e_i$ evaluates to $\meta{True}$, so let $j \in [1,n]$ be such that
            \begin{align*}
                & s_n; \rho' ; e_j \bigstep_{\dummyloc} \meta{True} \\
                & s_n; \rho' ; e_i \bigstep_{\dummyloc} \meta{False} \quad \text{for $i \neq j$}
            \end{align*}
            which gives us~\ref{typsf:se:create:D1} and~\ref{typsf:se:create:D2}.
            From the constructor $\Sigma' \vdash \var{ctor} : C$ we get
            \begin{equation*}
                \Sigma' ; I' ; \Phi_j \vdash_B \many{\var{create_j}} : C \quad \text{where } \quad \Phi_j = e_j \wedge \bigwedge_{\var{pre} \in \many{\var{pre}}} \var{pre}.
            \end{equation*}
            We get $s_n ; \rho' ; \Phi_j \bigstep_{\dummyloc} \meta{True}$ by applying~\rref{E-BopB} on~\eqref{eq:se-typesafety-precs} and~\ref{typsf:se:create:D1}.
            We apply~\Cref{lem:create-typesafety} mutually inductively with $\Sigma' ; I' ; \Phi_j \vdash_B \many{\var{create_j}} : C$ (note that $\Sigma'$ is strictly smaller than $\Sigma$) and $\Sigma' \vdash \rho' :_{s_n} I'$ (derived above) to obtain $(\ell', s')$ such that
            \begin{enumerate}[start=3,label={(\bfseries D\arabic*)}]
                \item\label{typsf:se:create:D3} $s_n ; \rho' ; \many{\var{create_j}} \bigstep_{B} (\ell', s')$
                \item\label{typsf:se:create:D4} $\Sigma'' \vdash \ell' :_{s'} B$
                \item\label{typsf:se:create:D5} $\Sigma'' \vdash \rho' :_{s'} I'$
                \item\label{typsf:se:create:D6} $s_n \subseteq s'$,
                \item\label{typsf:se:create:D7} $\forall \ell'' \in \dom{s'} \setminus \dom{s_n}.~\exists D.~\Sigma'' \vdash \ell'' :_{s'} D$.
                \item\label{typsf:se:create:D8} $\forall u, \sigma'.~\Sigma' \vdash u :_{s_n} \sigma' \implies \Sigma' \vdash u :_{s'} \sigma'$.
            \end{enumerate}
            where $\Sigma'' = \Sigma'~\meta{with}~ \{ \sstorage = (\Sigmastorep,\, B : C)\}$.

            We then apply~\rref{E-Create} rule on~\ref{typsf:se:create:I1}, $\rho'$, and the result of~\rref{E-CtorCases} applied to~\ref{typsf:se:create:D1}-\ref{typsf:se:create:D3}  to obtain $\psem{s}{\rho}{\kw{new}~B(\range{\var{se_i}}{i \in [1,n]})} (\ell',s')$. Since $\Sigmastore(B) = C$ and $\Sigma' \subseteq \Sigma$, it holds that \inlineequation[eq:se-typesafety-sigmapp]{\Sigma'' \subseteq \Sigma}. We get that $\Sigma \vdash \ell' :_{s'} B$ by~\Cref{lem:valuetyp-storage-weak} on~\eqref{eq:se-typesafety-sigmapp} and~\ref{typsf:se:create:D4}. It holds that \inlineequation[eq:se-typesafety-subs]{s \subseteq s'} from combining~\ref{typsf:se:create:I5} for all $i \in [1,n]$ and~\ref{typsf:se:create:D6}. We get $\Sigma \vdash \rho :_{s'} I$ by~\Cref{lem:valuetyp-storage-weak} applied to~\ref{typsf:se:env} and $s \subseteq s'$. We derive $\Sigma \vdash \ell :_{s'} \maybe{A}$ from~\Cref{lem:valuetyp-storage-weak} applied to~\ref{typsf:se:loc} and $s \subseteq s'$.

            To show $\forall \ell'' \in \dom{s'} \setminus \dom{s}.~\exists D.~\Sigma \vdash \ell'' :_{s'} D$, let $\ell'' \in \dom{s'} \setminus \dom{s}$. Since $s \subseteq s_1 \subseteq s_2 \subseteq \cdots \subseteq s_n \subseteq s' = s_{n+1}$, and $\ell'' \notin \dom{s_0}$ let $k \in [1, n+1]$ be the smallest such that $\ell'' \in \dom{s_k}$. We distinguish the cases
            \begin{proofcases}
                \case{$k \neq n+1$}
                We use~\ref{typsf:se:create:I6} to get \inlineequation[eq:se-typesafety-ellpwt]{\exists D.~\Sigma \vdash \ell'' :_{s_k} D}. We use~\Cref{lem:valuetyp-storage-weak} on~\eqref{eq:se-typesafety-ellpwt} and $s_k \subseteq s'$ to conclude $\exists D.~\Sigma \vdash \ell'' :_{s'} D$.
                \case{$k = n+1$}
                Using\ref{typsf:se:create:D7}, we get \inlineequation[eq:se-typesafety-ellpwt1]{\exists D.~\Sigma'' \vdash \ell'' :_{s'} D}. We use~\Cref{lem:valuetyp-storagetyp-weak} on~\eqref{eq:se-typesafety-ellpwt1} and~\eqref{eq:se-typesafety-sigmapp} to conclude $\exists D.~\Sigma \vdash \ell'' :_{s'} D$.
            \end{proofcases}

            To show $\forall u, \sigma'.~\Sigma \vdash u :_s \sigma' \implies \Sigma \vdash u :_{s'} \sigma'$, let $u, \sigma'$ be such that \inlineequation[eq:se-typesafety-uwt]{\Sigma \vdash u :_s \sigma'}. Using~\Cref{lem:valuetyp-storage-weak} on~\eqref{eq:se-typesafety-uwt} and~\eqref{eq:se-typesafety-subs} to obtain $\Sigma \vdash u :_{s'} \sigma'$

            \case{$\Sigma; I; \Phi \vdash_{\maybe{\var{A}}} \kw{new}~B\{\kw{value}:~\var{se_{n+1}}\}(\var{se_1}, \dots, \var{se_n}) : B$ (rule~\rref{T-CreatePayable})}
            Similarly to the case~\rref{T-Create}, but using induction we also obtain $(I_{n+1}1)$ - $(I_{n+1}7)$ for $\var{se}_{n+1}$, with $(I_{n+1}2)$ restricted to $\Sigma \vdash v_{n+1} :_{s_{n+1}} \kw{uint256}$. By inversion via~\rref{V-BaseValAlpha}, we get \inlineequation[typsf:se:create:uint256]{\vdash v_{n+1} : \kw{uint256}}.
            We then define $\rho' = \range{x_i \mapsto \var{v}_i}{i \in [1,n]}  \cup \{\kw{caller} \mapsto \ell, \kw{origin} \mapsto \rho(\kw{origin}), \kw{callvalue} \mapsto v_{n+1} \}$. The proof then proceeds with $s_{n+1}$ in place of $s_n$ and we derive
            $\Sigma \vdash \rho' :_{s_{n+1}} I'$ by showing $\vdash \rho'(\kw{callvalue}) : \kw{uint256}$ from~\eqref{typsf:se:create:uint256}.

        \end{proofcases}
    \end{proof}

    \begin{lemma}[Create Type Safety]
        \label{lem:create-typesafety}
        Assume that
        \begin{enumerate}[start=1,label={(\bfseries H\arabic*)}]
            \item\label{typsf:create:wt} $\Sigma; I; \Phi \vdash_{A} \many{\var{create}} : C$
            \item\label{typsf:create:env} $\Sigma \vdash \rho :_s I$
            \item\label{typsf:creates:precond} $s ; \rho ; \Phi \bigstep_{\dummyloc} \meta{True}$
        \end{enumerate}
        Then there exist $s'$ and $\ell$ such that
        \begin{itemize}
            \item $s ; \rho ; \many{\var{create}} \bigstep_{A} (\ell, s')$,
            \item $\Sigma' \vdash \ell :_{s'} A$,
            \item $\Sigma' \vdash \rho :_{s'} I$,
            \item $s \subseteq s'$,
            \item $\forall \ell' \in \dom{s'} \setminus \dom{s}.~\exists B.~\Sigma' \vdash \ell' :_{s'} B$.
            \item $\forall u, \sigma'.~\Sigma \vdash u :_s \sigma' \implies \Sigma \vdash u :_{s'} \sigma'$.
        \end{itemize}
        where $\Sigma' = \Sigma ~\meta{with}~ \{ \sstorage = (\Sigmastore,\, A : C) \}$.
    \end{lemma}
    \begin{proof}
        By induction on the typing derivation~\ref{typsf:create:wt}, with mutual induction on the proof of~\Cref{lem:se-typesafety}.
        The only possible rule to apply is~\rref{T-Creates}, giving us \inlineequation[eq:creates-typesafety-IH1]{C = \range{x_i : \sigma_i}{i \in [1,n]}}\\ 
         and \inlineequation[eq:creates-typesafety-IH2]{\forall \, i \in [1,n]. ~\,\Sigma; I; \Phi \vdash_{\bot} \var{se}_i : \sigma_i}.

        We set $s_0 = s$. For every $i \in [1,n]$ we use mutual induction on~\Cref{lem:se-typesafety} with~\eqref{eq:creates-typesafety-IH2} and~\ref{typsf:create:env} or $\Sigma \vdash \rho :_{s_{i-1}} I$ (given from~\ref{typsf:create:IHenv} of the previous step), and  $\Sigma \vdash \dummyloc :_{s_{i-1}} \bot$ (derived through~\rref{V-None}), and $s_{i-1} ; \rho ; \Phi \bigstep_{\dummyloc} \meta{True}$ (obtained from the previous step and~\Cref{lem:sem-storage-weak}), and we obtain for all $i \in [1,n]$
        \begin{enumerate}[start=1,label={(\bfseries I\arabic*)}]
            \item\label{typsf:create:IHse} $s_{i-1} ; \rho ; se_i \bigstep_{\dummyloc} (v_i, s_i)$,
            \item\label{typsf:create:IHvi} $\Sigma \vdash v_i :_{s_i} \sigma_i$,
            \item\label{typsf:create:IHenv} $\Sigma \vdash \rho :_{s_i} I$,
            \item\label{typsf:create:IHloc} $\Sigma \vdash \dummyloc :_{s_i} \bot$,
            \item\label{typsf:create:IHsub}  $s_{i-1} \subseteq s_i$.
            \item\label{typsf:create:ellpwt}  $\forall \ell' \in \dom{s_{i}} \setminus \dom{s_{i-1}}.~\exists D.~\Sigma \vdash \ell' :_{s_{i}} D$,
        \end{enumerate}
        We get the first premise of the~\rref{E-Creates} rule from~\ref{typsf:create:IHse}. Let $\ell = \fresh{s_n}$ and 
        $s_{n+1} = s_n[\ell \mapsto_{A} \range{ x_i \mapsto v_i}{i \in [1,n]}]$. Since $\ell \notin \dom{s_n}$ we have 
        \inlineequation[eq:creates-typesafety-subsn]{s_n \subset s_{n+1}}.

        We set $s' = s_{n+1}$, applying the~\rref{E-Creates} rule to get $s ;\rho ;\many{\var{create}} \bigstep_A (\ell, s_{n+1})$.
        To derive \inlineequation[eq:creates-typesafety-locwt]{\Sigma' \vdash \ell :_{s_{n+1}} A} we use the~\rref{V-Contract} rule and~\rref{V-AddrIsContract} with
        \begin{itemize}
            \item $\ell \in \dom{s_{n+1}}$ by definition.
            \item $A \in \Sigmastorep$ by definition of $\Sigma'$.
            \item $s_{n+1}(\ell).{\kw{type}} = A$ by definition of $s_{n+1}(\ell)$.
            \item $\forall \, x. ~\, x \in \dom{s_{n+1}(\ell)}
                \Leftrightarrow (\exists \, \sigma. ~\, \Sigmastorep(A)(x) = \sigma)$
                from the definition of $s_{n+1}$.
            \item $(\forall \, x,\sigma. ~\, \Sigmastorep(A)(x) = \sigma ~\Rightarrow ~ \Sigma \vdash s_{n+1}(\ell) (x) :_{s_{n+1}} \sigma)$ we derive as follows: let $x_i, \sigma_i$ be such that $\Sigmastorep(A)(x_i) = \sigma_i$. By definition of $s_{n+1}$ we have $s_{n+1}(\ell)(x_i) = v_i$. By~\ref{typsf:create:IHvi}, $\Sigma \vdash v_i :_{s_i} \sigma_i$, and we can weaken this using~\Cref{lem:valuetyp-storage-weak} and $s_i \subseteq s_n$ that we get by transitivity on~\ref{typsf:create:IHsub}, to get $\Sigma \vdash v_i :_{s_n} \sigma_i$. Again by~\Cref{lem:valuetyp-storage-weak} and~\eqref{eq:creates-typesafety-subsn} we have $\Sigma \vdash v_i :_{s_{n+1}} \sigma_i$, which we weaken using~\Cref{lem:valuetyp-storagetyp-weak} and $\Sigma \subseteq \Sigma'$ to obtain the desired $\Sigma' \vdash v_i :_{s_{n+1}} \sigma_i$.
        \end{itemize}
        To derive $\Sigma' \vdash \rho :_{s_{n+1}} I$ we use the~\rref{V-Env} rule with
        \begin{itemize}
            \item $\dom{\rho} = \dom{I} \cup \{\kw{caller}, \kw{origin}, \kw{callvalue}\}$ from~\ref{typsf:create:IHenv} (with $i=n$).
            \item $\vdash \rho(\kw{caller}) : \kw{address}$ and $\vdash \rho(\kw{origin}) : \kw{address}$ and $\vdash \rho(\kw{callvalue}) : \kw{uint256}$ from~\ref{typsf:create:IHenv} (with $i=n$).
            \item $\forall \, x \in \dom{I}. ~\, \Sigma' \vdash \rho(x) :_{s_{n+1}} I(x)$ as follows: let $x \in \dom{I}$. From~\ref{typsf:create:IHenv} (with $i=n$) it holds that $\Sigma \vdash \rho(x) :_{s_{n}} I(x)$. We weaken this using~\Cref{lem:valuetyp-storage-weak} with~\eqref{eq:creates-typesafety-subsn} to get $\Sigma \vdash \rho(x) :_{s_{n+1}} I(x)$, and we further weaken using~\Cref{lem:valuetyp-storagetyp-weak} and $\Sigma \subseteq \Sigma'$ to get $\Sigma' \vdash \rho(x) :_{s_{n+1}} I(x)$.
        \end{itemize}
        We get \inlineequation[eq:creates-typesafety-subs]{s \subseteq s'} by transitivity on~\ref{typsf:create:IHsub} and~\eqref{eq:creates-typesafety-subsn}.

        We now show that $\forall \ell' \in \dom{s_{n+1}} \setminus \dom{s}.~\exists B.~\Sigma' \vdash \ell' :_{s_{n+1}} B$.
        Let $\ell' \in \dom{s_{n+1}} \setminus \dom{s}$. We know that $\dom{s_{n+1}} = \dom{s_n} \cup \{\ell\}$ by definition of $s_{n+1}$. If $\ell' = \ell$, then $B=A$ and we conclude with~\eqref{eq:creates-typesafety-locwt}. If $\ell' \neq \ell$, then $\ell' \in \dom{s_n}\setminus \dom{s}$. Since $s \subseteq s_1 \subseteq \cdots \subseteq s_n$, let $k \in [1,n]$ be the smallest such that $\ell \in \dom{s_k}$. From~\ref{typsf:create:ellpwt} for $i = k$ let $B$ be such that $\Sigma \vdash \ell' :_{s_{i}} B$. By weakening via~\Cref{lem:valuetyp-storage-weak} with $s_k \subset s_{n+1}$ (which we get by transitivity of inclusions), and then weakening via~\Cref{lem:valuetyp-storagetyp-weak} with $\Sigma' \subseteq \Sigma$ we conclude $\Sigma' \vdash \ell' :_{s_{n+1}} B$.

        Lastly, we show that $\forall u, \sigma'.~\Sigma \vdash u :_s \sigma' \implies \Sigma \vdash u :_{s'} \sigma'$, let $u, \sigma'$ be such that \inlineequation[eq:creates-typesafety-uwt]{\Sigma \vdash u :_s \sigma'}. Using~\Cref{lem:valuetyp-storage-weak} on~\eqref{eq:creates-typesafety-uwt} and~\eqref{eq:creates-typesafety-subs} to obtain $\Sigma \vdash u :_{s'} \sigma'$.
    \end{proof}

    \begin{lemma}[Update Type Safety]
        \label{lem:update-typesafety}
        Assume that
        \begin{enumerate}[start=1,label={(\bfseries H\arabic*)}]
            \item\label{typsf:upd:wt} $\Sigma; I; \Phi \vdash_{A} \many{\var{update}}$
            \item\label{typsf:upd:env} $\Sigma \vdash  \rho :_s I$
            \item\label{typsf:upd:loc} $\Sigma \vdash \ell :_{s} A$
            \item\label{typsf:upd:precond} $\psem{s}{\rho}{\Phi} \meta{True}$
        \end{enumerate}
        Then there exists $s'$ such that
        \begin{itemize}
            \item $\psem{s}{\rho}{\many{\var{update}}} s'$
            \item $\Sigma \vdash \ell :_{s'} A$
            \item $\Sigma \vdash \rho :_{s'} I$
            \item $\dom{s} \subseteq \dom{s'}$
            \item $\forall \ell' \in \dom{s'} \setminus \dom{s}.~\exists B.~\Sigma \vdash \ell' :_{s'} B$
            \item $\forall u, \sigma'.~\Sigma \vdash u :_s \sigma' \implies \Sigma \vdash u :_{s'} \sigma'$
        \end{itemize}
    \end{lemma}
    \begin{proof}
        By induction on the typing derivation~\ref{typsf:upd:wt}. Let $\many{\var{update}} =  \range{\var{ref_i} \asgn \var{se_i}}{i \in [1,n]}$. The only possible rule to derive~\ref{typsf:upd:wt} is~\rref{T-Updates}, and then by~\rref{T-Update} we get for every $i \in [1,n]$ that
        \begin{align}
         \Sigma; I;\vdash^{\kw{S}}_{\var{A}, \kw{U}} \var{ref_i} : \sigma_i \label{eq:update-typesafety-IH1}\\
         \Sigma; I; \Phi  \vdash_{\var{A}} \var{se_i} : \sigma_i \label{eq:update-typesafety-IH2} \\
          \neg (\var{ref}_j \specific \var{ref}_i) \text{ for $j \in [1,i]$} \label{eq:update-typesafety-IH3} 
        \end{align}
        We then repeatedly apply~\Cref{lem:se-typesafety} on~\eqref{eq:update-typesafety-IH2} with also~\ref{typsf:upd:env}-\ref{typsf:upd:precond} and we obtain for $i \in [1,n]$
        \begin{enumerate}[start=1,label={(\bfseries I$_i$\arabic*)}]
            \item\label{typsf:upd:IH:se} $\psem{s_{i-1}}{\rho}{se_i} (v_i, s_i)$,
            \item\label{typsf:upd:IH:vi} $\Sigma \vdash v_i :_{s_i} \sigma_i$,
            \item\label{typsf:upd:IH:env} $\Sigma \vdash \rho :_{s_i} I$,
            \item\label{typsf:upd:IH:loc} $\Sigma \vdash \ell :_{s_i} A$,
            \item\label{typsf:upd:IH:sub} $s_{i-1} \subseteq s_i$,
            \item\label{typsf:upd:IH:ellpwt} $\forall \ell' \in \dom{s_{i}} \setminus \dom{s_{i-1}}.~\exists D.~\Sigma \vdash \ell' :_{s_{i}} D$
            \item\label{typsf:upd:IH:preserve} $\forall u, \sigma'.~\Sigma \vdash u :_{s_{i-1}} \sigma' \implies \Sigma \vdash u :_{s_i} \sigma'$
        \end{enumerate}
        obtaining $s_{i-1} ; \rho ; \Phi \bigstep_{\dummyloc} \meta{True}$ from~\Cref{lem:sem-storage-weak} on~\ref{typsf:upd:precond} and $(I_{i-1}5)$. To apply~\rref{E-Updates} we derive
        \begin{equation}
            \label{eq:update-typesafety-ins1}
            \ins{s_{n+i-1} ; \rho ; \var{ref_i} ; v_i}{s_{n+i}}
        \end{equation}
        for $i \in [1,n]$, together with
        \begin{enumerate}[start=1,label={(\bfseries D$_i$\arabic*)}]
            \item\label{typsf:upd:D:loc} $\Sigma \vdash \ell :_{s_{n+i}} A$
            \item\label{typsf:upd:D:env} $\Sigma \vdash \rho :_{s_{n+i}} I$
            \item\label{typsf:upd:D:vj} $\Sigma \vdash v_j :_{s_{n+i}} \sigma_j$ for all $j \in [1,n]$
            \item\label{typsf:upd:D:ellpwt} $\forall \ell' \in \dom{s_{n+i}} \setminus \dom{s}.~\exists D.~\Sigma \vdash \ell' :_{s_{n+i}} D$
            \item\label{typsf:upd:D:preserve} $\forall u, \sigma'.~\Sigma \vdash u :_{s_{n+i-1}} \sigma' \implies \Sigma \vdash u :_{s_{n+i}} \sigma'$
        \end{enumerate}
        For $i=0$ we get ($D_{0}1$) from $(I_{n}4)$, ($D_{0}2$) from $(I_{n}3)$ and ($D_{0}3$) through~\Cref{lem:valuetyp-storage-weak} on $(I_{n}2)$ and $(I_{n}5)$. We get ($D_{0}4$) from repeatedly applying~\ref{typsf:upd:IH:ellpwt}, and ($D_{0}5$) from~\ref{typsf:upd:IH:preserve}. To derive the insertion judgement we proceed by induction on~\eqref{eq:update-typesafety-IH1}.
        \begin{proofcases}
            \case{$\Sigma; I \vdash^{\kw{S}}_{A, \kw{U}} x : \sigma_i$ (rule~\rref{T-Storage})}
            Inversion through~\rref{T-Storage} gives us $\Sigmastore(A)(x) =  \sigma_i$.
            To apply the~\rref{E-InsStorage} rule, we derive $x \in \dom{s_{n+i-1}(\ell)}$. Since $\Sigma \vdash \ell :_{s_{n+i-1}} A$ by ($D_{i-1}1$) we get by inversion through~\rref{V-Contract} and then~\rref{V-AddrIsContract} that $\ell \in \dom{s_{n+i-1}}$ and $A \in \dom{\Sigma}$. Further by inversion on~\eqref{eq:update-typesafety-IH1}, the only rule possible is~\rref{T-Storage}, which gives us $\Sigmastore(A)(x) =  \sigma_i$. From the fourth premise of the~\rref{V-AddrIsContract} rule we then get $x \in \dom{s_{n+i-1}(\ell)}$ and we apply~\rref{E-InsStorage} to obtain~\eqref{eq:update-typesafety-ins1}.
            \case{$\Sigma; I \vdash^{\kw{S}}_{A, \kw{U}} \var{ref}.x : \sigma_i$ (rule~\rref{T-Field})}
            Inversion through~\rref{T-Field} gives us
            \begin{align}
                 \Sigma; I \vdash^{\kw{S}}_{A, \kw{U}}  \var{ref} : B \label{eq:update-typesafety-IH3}\\
                 \Sigmastore(B)(x) = \sigma_i \label{eq:update-typesafety-IH4}
            \end{align}
            Using~\Cref{lem:ref-typesafety-untimed} on~\eqref{eq:update-typesafety-IH3}, ($D_{i-1}2$), and ($D_{i-1}1$) gives us $\ell'$ such that $\psem{s_{n+i-1}^{\kw{U}}}{\rho}{\var{ref}} (\ell', \kw{U})$ and $\Sigma \vdash \ell' :_{s_{n+i-1}} B$, which by inversion through~\rref{V-Contract} and~\rref{V-AddrIsContract} gives us $\ell' \in \dom{s_{n+i-1}}$ (and thus also $\ell' \in \meta{Addr}$) from the first premise, and together with~\eqref{eq:update-typesafety-IH4} we obtain $x \in \dom{s_{n+i-1}(\ell')}$ from the fourth premise. Using~\rref{E-InsField} we obtain~\eqref{eq:update-typesafety-ins1}.
        \end{proofcases}
        We derive~\ref{typsf:upd:D:loc} by~\Cref{lem:insertion-preserves-typing} on ($D_{i-1}1$), ($D_{i-1}1$),~\ref{typsf:upd:env} and ($D_{i-1}2$),~\eqref{eq:update-typesafety-IH1}, ($D_{i-1}3$) for $j = i$, and~\eqref{eq:update-typesafety-ins1}. Similarly we get~\ref{typsf:upd:D:vj} from~\Cref{lem:insertion-preserves-typing} on ($D_{i-1}3$) as the first hypothesis, and the rest of the arguments the same as above.
        To derive~\ref{typsf:upd:D:env} we proceed as follows. By inversion on ($D_{i-1}2$) through~\rref{V-Env} we get that 
        \begin{itemize}
            \item $\dom{\rho} = \dom{I}\cup \{\kw{caller}, \kw{origin}, \kw{callvalue}\}$
            \item $\forall \, x \in \dom{I}. ~\, \Sigma \vdash \rho(x) :_{s_{n+i-1}} I(x)$
            \item $\vdash \rho(\kw{caller}) : \kw{address}$
            \item $\vdash \rho(\kw{origin}) : \kw{address}$
            \item $\vdash \rho(\kw{callvalue}) : \kw{uint256}$
        \end{itemize}
        Let $x \in \dom{I}$. Applying~\Cref{lem:insertion-preserves-typing} on $\Sigma \vdash \rho(x) :_{s_{n+i-1}} I(x)$ as the first hypothesis, and the rest of the arguments the same as above, we obtain $\Sigma \vdash \rho(x) :_{s_{n+i}} I(x)$. We therefore derived $\forall \, x \in \dom{I}. ~\, \Sigma \vdash \rho(x) :_{s_{n+i}} I(x)$ The rule~\rref{V-Env} then concludes~\ref{typsf:upd:D:env}.

        To derive~\ref{typsf:upd:D:ellpwt}, let $\ell' \in \dom{s_{n+i}} \setminus \dom{s}$. By ($D_{i-1}4$) let $D$ be such that we get \inlineequation[eq:update-typesafety-ellpwt1]{\Sigma \vdash \ell' :_{s_{n+i-1}} D}. Using~\Cref{lem:insertion-preserves-typing} on~\eqref{eq:update-typesafety-ellpwt1} as the first hypothesis, and the rest of the arguments the same as above, we get $\Sigma \vdash \ell' :_{s_{n+i}} D$.

        To derive~\ref{typsf:upd:D:preserve}, let $u, \sigma'$ be such that \inlineequation[eq:update-typesafety-uwt]{\Sigma \vdash u :_{s_{n+i-1}} \sigma'}.
        Using~\Cref{lem:insertion-preserves-typing} on~\eqref{eq:update-typesafety-uwt} as the first hypothesis, and the rest of the arguments the same as above, we get $\Sigma \vdash u :_{s_{n+i}} \sigma'$.

        We conclude the proof through~\rref{E-Updates} on~\ref{typsf:upd:IH:se} and~\eqref{eq:update-typesafety-ins1}, obtaining
        \begin{itemize}
        \item $\Sigma \vdash \ell :_{s_{2n}} A$ from ($D_{n}1$),
        \item $\Sigma \vdash \rho :_{s_{2n}} I$ from ($D_{n}2$),
        \item $\forall \ell' \in \dom{s_{2n}} \setminus \dom{s}.~\exists B.~\Sigma \vdash \ell' :_{s_{2n}} B$ from ($D_{n}4$),
        \item $\forall u, \sigma'.~\Sigma \vdash u :_s \sigma' \implies \Sigma \vdash u :_{s'} \sigma'$ by transitivity on~\ref{typsf:upd:IH:preserve} and~\ref{typsf:upd:D:preserve} for all $i \in [1,n]$.
        \end{itemize}
        \end{proof}

    \begin{lemma}[Constructor Type Safety]
        \label{lem:constructor-typesafety}
        Assume that
        \begin{enumerate}[start=1,label={(\bfseries H\arabic*)}]
            \item\label{typsf:constr:wt} $\Sigma \vdash_{A} \var{cnstr}: C$
            \item\label{typsf:constr:env} $\Sigma \vdash \rho :_s I$
            \item\label{typsf:constr:pre} $s ; \rho ; \var{cnstr}_{\kw{iff}} \bigstep_{\dummyloc} \meta{True}$
        \end{enumerate}
        Then there exist $s'$ and $\ell$ such that
        \begin{itemize}
            \item $s ; \rho ; \var{cnstr} \bigstep_{A} (\ell, s')$,
            \item $\Sigma' \vdash \ell :_{s'} A$,
            \item $\Sigma' \vdash \rho :_{s'} I$,
            \item $s \subseteq s'$,
            \item $\forall \ell' \in \dom{s'} \setminus \dom{s}.~\exists B.~\Sigma' \vdash \ell' :_{s'} B$
            \item $\forall u, \sigma'.~\Sigma \vdash u :_s \sigma' \implies \Sigma \vdash u :_{s'} \sigma'$
        \end{itemize}
        where $\Sigma' = \Sigma ~\meta{with}~ \{ \sstorage = (\Sigmastore,\, A : C) \}$.
    \end{lemma}

    \begin{proof}
        Let $\var{cnstr} = \constrinlinecase{\var{I}}{\many{\var{pre}}}{e}{\many{\var{create}_i}}{\many{post}}$.
        By inversion on~\ref{typsf:constr:wt} through~\rref{T-Ctor} we obtain
        \begin{enumerate}[start=1,label={(\bfseries I\arabic*)}]
            \item\label{typsf:constr:pre} $\many{(\Sigma ; \var{I} \vdash_{\none,\kw{U}} \var{pre} : \kw{bool})}$
            \item\label{typsf:constr:case} $(\forall \, i \in [1,n].~\,\Sigma;  \var{I} \vdash_{\none,\kw{U}} e_i : \kw{bool})$
            \item\label{typsf:constr:phi} $\Phi_i = e_i \wedge \bigwedge_{\var{pre} \in \many{\var{pre}}}{\var{pre}}$
            \item\label{typsf:constr:create} $(\forall \, i \in [1,n]. ~\, \Sigma; \var{I}; \Phi_i \vdash_{A} \many{\var{create}_i} : C)$
            \item\label{typsf:constr:disj} $\dom{C} \cap \dom{I} = \emptyset$
            \item\label{typsf:constr:posts} $\many{( \Sigma'; \var{I} \vdash_{A,\kw{U}} \var{post}: \kw{bool} )}$
            \item\label{typsf:constr:onecase} $\Sigma; I; \kw{true} \vDash_{\bot} \bigwedge_{\var{pre \in \many{\var{pre}}}}{\var{pre}} \Rightarrow (\bigvee_{i \in [1,n]} e_i) ~~ \wedge ~~(\bigwedge_{i \in [1,n], j \in [i+1, n]} \neg (e_i \wedge e_j))$
        \end{enumerate}
        By inversion on~\ref{typsf:constr:onecase} through the~\rref{ValidExps} we get
        \begin{equation}
            \label{eq:typesafety-constr-validexps}
            \begin{aligned}
            & \forall s', \rho', \ell'.~\Sigma \vdash \rho' :_{s'} I \wedge \Sigma \vdash \ell' :_{s'} \bot \implies \\
            & s'; \rho'; \bigwedge_{\var{pre \in \many{\var{pre}}}}{\var{pre}} \Rightarrow (\bigvee_{i \in [1,n]} e_i) ~~ \wedge ~~(\bigwedge_{i \in [1,n], j \in [i+1, n]} \neg (e_i \wedge e_j)) \bigstep_{\ell'} \meta{True}
            \end{aligned}
        \end{equation}
        We take $s' = s$ and $\rho' = \rho$ and any location $\ell' = \dummyloc$. We get the lefthand-side of implication by~\ref{typsf:constr:env} and~\rref{V-None}. With this data we get from~\eqref{eq:typesafety-constr-validexps} that
        \begin{equation}
            \label{eq:typesafety-constr-precond}
            s ; \rho; \bigwedge_{\var{pre \in \many{\var{pre}}}}{\var{pre}} \Rightarrow (\bigvee_{i \in [1,n]} e_i) ~~ \wedge ~~(\bigwedge_{i \in [1,n], j \in [i+1, n]} \neg (e_i \wedge e_j)) \bigstep_{\dummyloc} \meta{True}
        \end{equation}
        From~\ref{typsf:constr:pre} we obtain
        \begin{equation}
            \many{s ; \rho ; \var{pre} \bigstep_{\dummyloc} \meta{True}} \label{eq:typesafety-constr-precs}
        \end{equation}
        From inverting the derivation of~\eqref{eq:typesafety-constr-precond} using the~\rref{E-BopB} rule for implication we know that since an implication evaluates to $\meta{True}$ and the the lefthand-side also evaluates to $\meta{True}$ due to~\eqref{eq:typesafety-constr-precs}, then the righthand-side evaluates to $\meta{True}$. With similar meta-level reasoning for conjunctions we obtain that
        \begin{align*}
            s ; \rho ; \bigvee_{i \in [1,n]} e_i \bigstep_{\dummyloc} \meta{True} \\
            s ; \rho ;\bigwedge_{i \in [1,n], j \in [i+1, n]} \neg (e_i \wedge e_j) \bigstep_{\dummyloc} \meta{True}
        \end{align*}
        These judgements encode that at least one and at most one $e_i$ evaluates to $\meta{True}$, so let $j \in [1,n]$ be such that
        \begin{align}
            & s; \rho ; e_j \bigstep_{\dummyloc} \meta{True} \label{eq:typesafety-constr-ej}\\
            & s; \rho ; e_i \bigstep_{\dummyloc} \meta{False} \quad \text{for $i \neq j$}\label{eq:typesafety-constr-ei}
        \end{align}
        A similar reasoning together with~\ref{typsf:constr:pre} gives us \inlineequation[eq:typesafety-constr-phi-j]{s; \rho ; \Phi_j \bigstep_{\dummyloc} \meta{True}}.
        We then use~\Cref{lem:create-typesafety} with~\ref{typsf:constr:create} for $i=j$,~\ref{typsf:constr:env} and~\eqref{eq:typesafety-constr-phi-j} to obtain such $s'$ and $\ell$ that
        \begin{itemize}
            \item $s ; \rho ; \many{\var{create}_j} \bigstep_A (\ell, s')$,
            \item $\Sigma' \vdash \ell :_{s'} A$,
            \item $\Sigma' \vdash \rho :_{s'} I$,
            \item $s \subseteq s'$,
            \item $\forall \ell' \in \dom{s'} \setminus \dom{s}.~\exists B.~\Sigma' \vdash \ell' :_{s'} B$.
            \item $\forall u, \sigma'.~\Sigma \vdash u :_s \sigma' \implies \Sigma \vdash u :_{s'} \sigma'$
        \end{itemize}
        We conclude the proof with~\rref{E-CtorCases} on~\ref{eq:typesafety-constr-ej}, \ref{eq:typesafety-constr-ei} and $s ; \rho ; \many{\var{create}_j} \bigstep_A (\ell, s')$.
    \end{proof}

    Note that the type safety statement for non-payable constructors does not require the $\kw{callvalue}$ of the environment
    $\rho$ to be set to $0$, since calling a constructor in a well-typed slot expression will do so correctly by
    setting $\rho'(\kw{callvalue}) = 0$ in the rule~\rref{E-Create}.

    \begin{lemma}[Transition Type Safety]
        \label{lem:transition-typesafety}
        Assume that
        \begin{enumerate}[start=1,label={(\bfseries H\arabic*)}]
            \item\label{typsf:trans:wt} $\Sigma \vdash_{A} \var{trans}$
            \item\label{typsf:trans:env} $\Sigma \vdash \rho :_s I$
            \item\label{typsf:trans:loc} $\Sigma \vdash \ell :_{s} A$
            \item\label{typsf:trans:pre} $s ; \rho ; \var{trans}_{\kw{iff}} \bigstep_{\ell} \meta{True}$
        \end{enumerate}
        Then there exist $v$ and $s'$ such that
        \begin{itemize}
            \item $s ; \rho ; \var{trans} \bigstep_\ell (v, s')$,
            \item $\Sigma \vdash \ell :_{s'} A$,
            \item $\Sigma \vdash \rho :_{s'} I$,
            \item $\dom{s} \subseteq \dom{s'}$,
            \item $\forall \ell' \in \dom{s'} \setminus \dom{s}.~\exists B.~\Sigma \vdash \ell' :_{s'} B$
            \item $\forall u, \sigma'.~\Sigma \vdash u :_s \sigma' \implies \Sigma \vdash u :_{s'} \sigma'$
            % \item \todo{Should $v$ be well-typed? The type of $v$ does not appear in the trans typing.}
        \end{itemize}
    \end{lemma}

    \begin{proof}
        By inversion on~\ref{typsf:trans:wt} we obtain
        \begin{enumerate}[start=1,label={(\bfseries I\arabic*)}]
            \item\label{typsf:trans:IH} $\Sigma \vdash I ~\meta{wf}$
            \item $\many{(\Sigma ; \var{I} \vdash_{A,\kw{U}} \var{pre} : \kw{bool})}$
            \item $(\forall \, i \in [1,n]. ~\, \Sigma; \var{I}  \vdash_{A,\kw{U}} e_i : \kw{bool})$
            \item $\Phi_i = e_i \wedge \bigwedge_{\var{pre} \in \many{\var{pre}}}{\var{pre}}$
            \item\label{typsf:trans:update} $(\forall \, i \in [1,n].~\, \many{\Sigma; \var{I}; \Phi_i  \vdash_{A} \many{\var{upd}_i}})$
            \item\label{typsf:trans:ret} $(\forall \, i \in [1,n]. ~\, \Sigma; \var{I} \vdash_{A,\kw{T}} \var{ret}_i : \beta)$
            \item $\many{( \Sigma; \var{I} \vdash_{A,\kw{T}} \var{post}: \kw{bool} )}$
            \item\label{typsf:trans:onecase} $\Sigma; I; \kw{true} \vDash_{A} \bigwedge_{\var{pre \in \many{\var{pre}}}}{\var{pre}} \Rightarrow (\bigvee_{i \in [1,n]} e_i) ~~ \wedge ~~(\bigwedge_{i \in [1,n], j \in [i+1, n]} \neg (e_i \wedge e_j))$
        \end{enumerate}
        By inversion on~\ref{typsf:trans:onecase} through the~\rref{ValidExps} we get
        \begin{equation}
            \label{eq:typesafety-trans-validexps}
            \begin{aligned}
            & \forall s', \rho', \ell'.~\Sigma \vdash \rho' :_{s'} I \wedge \Sigma \vdash \ell' :_{s'} A \implies \\
            & s'; \rho'; \bigwedge_{\var{pre \in \many{\var{pre}}}}{\var{pre}} \Rightarrow (\bigvee_{i \in [1,n]} e_i) ~~ \wedge ~~(\bigwedge_{i \in [1,n], j \in [i+1, n]} \neg (e_i \wedge e_j)) \bigstep_{\ell'} \meta{True}
            \end{aligned}
        \end{equation}
        We take $s' = s$ and $\rho' = \rho$ and location $\ell' = \ell$. We get the lefthand-side of implication by~\ref{typsf:trans:env} and~\ref{typsf:trans:loc}. With this data we get from~\eqref{eq:typesafety-trans-validexps} that
        \begin{equation}
            \label{eq:typesafety-trans-precond}
            s ; \rho; \bigwedge_{\var{pre \in \many{\var{pre}}}}{\var{pre}} \Rightarrow (\bigvee_{i \in [1,n]} e_i) ~~ \wedge ~~(\bigwedge_{i \in [1,n], j \in [i+1, n]} \neg (e_i \wedge e_j)) \bigstep_{\ell} \meta{True}
        \end{equation}
        From~\ref{typsf:trans:pre} we obtain
        \begin{equation}
            \many{s ; \rho ; \var{pre} \bigstep_{\ell} \meta{True}} \label{eq:typesafety-trans-precs}
        \end{equation}
        From inverting the derivation of~\eqref{eq:typesafety-trans-precond} using the~\rref{E-BopB} rule for implication we know that since an implication evaluates to $\meta{True}$ and the the lefthand-side also evaluates to $\meta{True}$ due to~\eqref{eq:typesafety-trans-precs}, then the righthand-side evaluates to $\meta{True}$. With similar meta-level reasoning for conjunctions we obtain that
        \begin{align*}
            s ; \rho ; \bigvee_{i \in [1,n]} e_i \bigstep_{\ell} \meta{True} \\
            s ; \rho ;\bigwedge_{i \in [1,n], j \in [i+1, n]} \neg (e_i \wedge e_j) \bigstep_{\ell} \meta{True}
        \end{align*}
        These judgements encode that at least one and at most one $e_i$ evaluates to $\meta{True}$, so let $j \in [1,n]$ be such that
        \begin{align}
            & s; \rho ; e_j \bigstep_{\ell} \meta{True} \label{eq:typesafety-trans-ej}\\
            & s; \rho ; e_i \bigstep_{\ell} \meta{False} \quad \text{for $i \neq j$}\label{eq:typesafety-trans-ei}
        \end{align}
        A similar reasoning together with~\ref{typsf:trans:pre} gives us \inlineequation[eq:typesafety-trans-phi-j]{s; \rho ; \Phi_j \bigstep_{\ell} \meta{True}}.
        We then use~\Cref{lem:update-typesafety} with~\ref{typsf:trans:update} for $i=j$,~\ref{typsf:trans:env}, \ref{typsf:trans:loc}, and~\eqref{eq:typesafety-trans-phi-j} to obtain such $s'$ that
        \begin{enumerate}[start=1,label={(\bfseries R\arabic*)}]
            \item\label{typsf:trans:newState} $s ; \rho ; \many{\var{upd}} \bigstep s'$,
            \item\label{typsf:trans:newloc} $\Sigma \vdash \ell :_{s'} A$,
            \item\label{typsf:trans:newenv} $\Sigma \vdash \rho :_{s'} I$.
            \item\label{typsf:trans:newlocs} $\forall \ell' \in \dom{s'} \setminus \dom{s}.~\exists B.~\Sigma \vdash \ell' :_{s'} B$.
            \item\label{typsf:trans:uwt} $\forall u, \sigma'.~\Sigma \vdash u :_s \sigma' \implies \Sigma \vdash u :_{s'} \sigma'$
        \end{enumerate}
        By applying~\Cref{lem:expr-typesafety-timed} on~\ref{typsf:trans:ret} for $i = j$, \ref{typsf:trans:env} and \ref{typsf:trans:newenv}, \ref{typsf:trans:loc} and \ref{typsf:trans:newloc} we obtain $v$, such that \inlineequation[typsf:trans:retval]{\psem{(s,s')^{\kw{T}}}{\rho}{\var{ret}_j} v} and $\vdash v : \beta$.
        We conclude the proof with~\rref{E-Trans} on~\ref{typsf:trans:pre} as the first premise and as the second the rule~\rref{E-TransCases} applied on~\ref{eq:typesafety-trans-ej}, \ref{eq:typesafety-constr-ei},~\ref{typsf:trans:newState} and~\eqref{typsf:trans:retval}. We get the desired properties of $s'$ from~\ref{typsf:trans:newloc}, \ref{typsf:trans:newenv},~\ref{typsf:trans:newlocs}, and~\ref{typsf:trans:uwt}.
    \end{proof}

    \begin{lemma}[State Transition Type Safety]
    \label{lem:state-transition-typesafety}
     Let $\Sigma \vdash s \rightsquigarrow^* s'$. Then
     \begin{enumerate}[start=1,label={(\bfseries S\arabic*)}]
        \item\label{typsf:reachState:preservation} $\forall v, \sigma.~\Sigma \vdash v :_s \sigma \implies \Sigma \vdash v :_{s'} \sigma$
        \item\label{typsf:reachState:domain} $\dom{s} \subseteq \dom{s'}$
     \end{enumerate}
    \end{lemma}

    \begin{proof}
        By induction on $\Sigma \vdash s \rightsquigarrow^* s'$.
        If $s = s'$, then trivially true.
        Otherwise let \inlineequation[typsf:reachState:midState]{\Sigma \vdash s \rightsquigarrow^* s''} and \inlineequation[typsf:reachState:onestep]{\Sigma \vdash s'' \rightsquigarrow s'}.
        By induction hypothesis on~\eqref{typsf:reachState:midState} we get that
        \inlineequation[typsf:reachState:IH1]{\forall v,\sigma.~\Sigma \vdash v :_s \sigma \implies \Sigma \vdash v :_{s''} \sigma} and
        \inlineequation[typsf:reachState:IH2]{\dom{s} \subseteq \dom{s''}}.
        Let $v,\sigma$ be such that $\Sigma \vdash v :_s \sigma$. By~\eqref{typsf:reachState:IH1} we have that $\Sigma \vdash v :_{s''} \sigma$. We need to show that $\Sigma \vdash v :_{s'} \sigma$ and $\dom{s''} \subseteq \dom{s'}$, from which by transitivity on~\eqref{typsf:reachState:IH2} we get $\dom{s} \subseteq \dom{s'}$.

        Unfolding~\eqref{typsf:reachState:onestep} we get $A$ and $\ell'$ and distinguish the following two cases:
        \begin{proofcases}
            \case{$\Sigma \vdash_A s'' \rightsquigarrow_{\ell'}^{\kw{C}} s'$}
                From the premise of the~\rref{E-Step-Cnstr}, let $\rho$ be such that
                \begin{enumerate}[start=1,label={(\bfseries P\arabic*)}]
                    \item\label{typsf:reachState:cnstr} $\Sigmacode(A) = \var{cnstr}~|~\exists \Sigma'.~\Sigma'\subset \Sigma ~\wedge~ \Sigma' \vdash_{A} \var{cnstr} : \Sigmastore(A)$
                    \item\label{typsf:reachState:env} $\Sigma \vdash \rho :_{s''} \var{cnstr}_{\meta{iface}}$
                    \item\label{typsf:reachState:eval} $s'' ; \rho ; \var{cnstr} \bigstep (\ell', s')$
                \end{enumerate}
                By inversion on~\ref{typsf:reachState:eval} we obtain \inlineequation[typsf:reachState:pre]{\psemdummy{s''}{\rho}{\var{cnstr}_{\kw{iff}}} \meta{True}}. Since $\Sigma$ is well-typed, let \inlineequation[typsf:reachState:sigp]{\Sigma' \subset \Sigma} be such that \inlineequation[typsf:reachState:cnstrWtp]{\Sigma' \vdash_{A} \var{cnstr} : \Sigmastore(A)}. Using~\Cref{lem:valuetyp-storagetyp-weak} on~\eqref{typsf:reachState:cnstrWtp} and~\eqref{typsf:reachState:sigp} we obtain \inlineequation[typsf:reachState:cnstrWt]{\Sigma \vdash_{A} \var{cnstr} : \Sigmastore(A)}. We then apply~\Cref{lem:constructor-typesafety} on~\eqref{typsf:reachState:cnstrWt}, \ref{typsf:reachState:env} and \ref{typsf:reachState:eval}, together with determinism of pointer semantics, to get that $s'' \subseteq s'$, from which we deduce $\dom{s''} \subseteq \dom{s'}$;  and $\forall u, \sigma'.~\Sigma \vdash u :_{s''} \sigma' \implies \Sigma \vdash u :_{s'} \sigma'$, which we instantiate with $u = v$ and $\sigma' = \sigma$, to get the $\Sigma \vdash v :_{s'} \sigma$.
            \case{$\Sigma \vdash_A s'' \rightsquigarrow_{\ell'}^{\kw{B}} s'$}
                From the premise of the~\rref{E-Step-Trans} rule, let $\rho$ be such that
                \begin{enumerate}[start=1,label={(\bfseries P\arabic*)}]
                    \item\label{typsf:reachState:trans:loc} $\Sigma \vdash \ell' :_{s''} A$
                    \item\label{typsf:reachState:trans} $\var{trans} \in \Sigmatranss(A)$
                    \item\label{typsf:reachState:trans:env} $\Sigma \vdash \rho :_{s''} \var{trans}_{\meta{iface}}$
                    \item\label{typsf:reachState:trans:eval} $s'' ; \rho ; \var{trans} \bigstep (v, s')$
                \end{enumerate}
                By inversion on~\ref{typsf:reachState:trans:eval} we obtain \inlineequation[typsf:reachState:trans:pre]{s'' ; \rho ; \var{trans}_{\kw{iff}} \bigstep_{\ell'} \meta{True}}. Since $\Sigma$ is well-typed, let \inlineequation[typsf:reachState:trans:sigp]{\Sigma' \subset \Sigma} be such that \inlineequation[typsf:reachState:transWtp]{\Sigma' \vdash_{A} \var{trans}}. Using~\Cref{lem:valuetyp-storagetyp-weak} on~\eqref{typsf:reachState:transWtp} and~\eqref{typsf:reachState:trans:sigp} we obtain \inlineequation[typsf:reachState:transWt]{\Sigma \vdash_{A} \var{trans}}. We then apply~\Cref{lem:transition-typesafety} on~\eqref{typsf:reachState:transWt}, \ref{typsf:reachState:trans:env}, \ref{typsf:reachState:trans:loc}, and \ref{typsf:reachState:trans:eval}, together with determinism of pointer semantics, to get that $\dom{s''} \subseteq \dom{s'}$ and $\forall u, \sigma'.~\Sigma \vdash u :_{s''} \sigma' \implies \Sigma \vdash u :_{s'} \sigma'$, which we instantiate with $u = v$ and $\sigma' = \sigma$, to get the $\Sigma \vdash v :_{s'} \sigma$.
        \end{proofcases}
    \end{proof}

   \begin{lemma}[Reachable Store is Well-typed]\label{lem:store-wt}
    Let $s$ such $\possible{s}$ and $\ell \in s$.
    Then there exists $A$, such that $\Sigma \vdash \ell :_s A$.
    \end{lemma}
    \begin{proof}
        Unfolding $\possible{s}$, we get that $\Sigma \vdash \emptyset \rightsquigarrow^{*} s$. Since by~\Cref{lem:state-transition-typesafety} part~\ref{typsf:reachState:domain} the state transitions potentially only increase domains, let $s'$ be first state in the transitive closure, which contains $\ell$, i.e.
        \begin{enumerate}[start=1,label={(\bfseries H\arabic*)}]
            \item $\Sigma \vdash \emptyset \rightsquigarrow^{*} s''$
            \item\label{typsf:reach:noloc} $\ell \notin \dom{s''}$
            \item\label{typsf:reach:onestep} $\Sigma \vdash s'' \rightsquigarrow s'$
            \item\label{typsf:reach:firstloc} $\ell \in \dom{s'}$
            \item\label{typsf:reach:transS} $\Sigma \vdash s' \rightsquigarrow^{*} s$
        \end{enumerate}
        Unfolding~\ref{typsf:reach:onestep} we get $A$ and $\ell'$ and distinguish the following two cases:
        \begin{proofcases}
            \case{$\Sigma \vdash_A s'' \rightsquigarrow_{\ell'}^{\kw{C}} s'$}
                From the premise of the~\rref{E-Step-Cnstr}, let $\rho$ be such that
                \begin{enumerate}[start=1,label={(\bfseries P\arabic*)}]
                    \item\label{typsf:reach:cnstr} $\Sigmacode(A) = \var{cnstr}$
                    \item\label{typsf:reach:env} $\Sigma \vdash \rho :_{s''} \var{cnstr}_{\meta{iface}}$
                    \item\label{typsf:reach:eval} $s'' ; \rho ; \var{cnstr} \bigstep (\ell', s')$
                \end{enumerate}
                By inversion on~\ref{typsf:reach:eval} we obtain \inlineequation[typsf:reach:pre]{\psemdummy{s''}{\rho}{\var{cnstr}_{\kw{iff}}} \meta{True}}. Since $\Sigma$ is well-typed, let \inlineequation[typsf:reach:sigp]{\Sigma' \subset \Sigma} be such that \inlineequation[typsf:reach:cnstrWtp]{\Sigma' \vdash_{A} \var{cnstr} : \Sigmastore(A)}. Using~\Cref{lem:valuetyp-storagetyp-weak} on~\eqref{typsf:reach:cnstrWtp} and~\eqref{typsf:reach:sigp} we obtain \inlineequation[typsf:reach:cnstrWt]{\Sigma \vdash_{A} \var{cnstr} : \Sigmastore(A)}. We then apply~\Cref{lem:constructor-typesafety} on~\eqref{typsf:reach:cnstrWt}, \ref{typsf:reach:env} and \ref{typsf:reach:eval}, together with determinism of pointer semantics, to get that $\forall \ell'' \in \dom{s'} \setminus \dom{s''}.~\exists B.~\Sigma \vdash \ell'' :_{s'} B$, which we instantiate with $\ell'' = \ell$, to get the $\exists B.~\Sigma \vdash \ell :_{s'} B$.
            \case{$\Sigma \vdash_A s'' \rightsquigarrow_{\ell'}^{\kw{B}} s'$}
                From the premise of the~\rref{E-Step-Trans} rule, let $\rho$ be such that
                \begin{enumerate}[start=1,label={(\bfseries P\arabic*)}]
                    \item\label{typsf:reach:trans:loc} $\Sigma \vdash \ell' :_{s''} A$
                    \item\label{typsf:reach:trans} $\var{trans} \in \Sigmatranss(A)$
                    \item\label{typsf:reach:trans:env} $\Sigma \vdash \rho :_{s''} \var{trans}_{\meta{iface}}$
                    \item\label{typsf:reach:trans:eval} $s'' ; \rho ; \var{trans} \bigstep (v, s')$
                \end{enumerate}
                By inversion on~\ref{typsf:reach:trans:eval} we obtain \inlineequation[typsf:reach:trans:pre]{s'' ; \rho ; \var{trans}_{\kw{iff}} \bigstep_{\ell'} \meta{True}}. Since $\Sigma$ is well-typed, let \inlineequation[typsf:reach:trans:sigp]{\Sigma' \subset \Sigma} be such that \inlineequation[typsf:reach:transWtp]{\Sigma' \vdash_{A} \var{trans}}. Using~\Cref{lem:valuetyp-storagetyp-weak} on~\eqref{typsf:reach:transWtp} and~\eqref{typsf:reach:trans:sigp} we obtain \inlineequation[typsf:reach:transWt]{\Sigma \vdash_{A} \var{trans}}. We then apply~\Cref{lem:transition-typesafety} on~\eqref{typsf:reach:transWt}, \ref{typsf:reach:trans:env}, \ref{typsf:reach:trans:loc}, and \ref{typsf:reach:trans:eval}, together with determinism of pointer semantics, to get that $\forall \ell'' \in \dom{s'} \setminus \dom{s''}.~\exists B.~\Sigma' \vdash \ell'' :_{s'} B$, which we instantiate with $\ell'' = \ell$, to get $\exists B.~\Sigma \vdash \ell :_{s'} B$.
        \end{proofcases}
        We then use~\Cref{lem:state-transition-typesafety} part~\ref{typsf:reachState:preservation} with~\ref{typsf:reach:transS} and $\exists B.~\Sigma \vdash \ell :_{s'} B$, to get the desired result.

    \end{proof}

% \section{Value Semantics}
% \input{valsem.tex}

% \section{Heap Reachability and Ownership}
% \input{reach.tex}

% \section{Semantic Soundness}
% \input{equiv-help.tex}
% \input{equiv.tex} % expression and references soundness
% \input{equiv-1.tex} % mapping expressions soundness
% \input{equiv-2.tex}
% \input{equiv-3.tex}

\bibliographystyle{ACM-Reference-Format}
\bibliography{bib}

\pagebreak
% Appendix with list of theorems
\listoftheorems[numwidth=2.5em,ignoreall,show={theorem,lemma,corollary,proposition,definition}]
\end{document}